\documentclass{article}
\usepackage[preprint]{log_2026}

\makeatletter
\renewcommand{\@noticestring}{Accepted at the Learning on Graphs Conference (LoG 2026), oral presentation.}
\makeatother

\usepackage{booktabs}
\usepackage{amsfonts}
\usepackage{amsmath}
\usepackage{amssymb}
\usepackage{amsthm}
\usepackage{graphicx}
\usepackage{pifont}
\usepackage{float}
\usepackage{array}

\newcommand{\cmark}{\ding{51}}
\newcommand{\xmark}{\ding{55}}

\usepackage[numbers,compress,sort]{natbib}

\newtheorem{theorem}{Theorem}
\newtheorem{definition}{Definition}
\newtheorem{proposition}{Proposition}
\newtheorem{lemma}{Lemma}
\newtheorem{corollary}{Corollary}
\newtheorem{remark}{Remark}

\newcommand{\multiset}[1]{\left\{\!\left\{#1\right\}\!\right\}}
\newcommand{\hash}{\operatorname{HASH}}
\newcommand{\one}{\mathbf{1}}
\newcommand{\relEmb}{\operatorname{Emb}_{r}}

\title[Expressive Power of Implicit Line-Graph WL]{%
\texorpdfstring{On the Expressive Power of Implicit Line-Graph\\
Higher-Order Weisfeiler--Leman}%
{On the Expressive Power of Implicit Line-Graph Higher-Order Weisfeiler--Leman}}

\author[Fan Yang]{%
Fan Yang\\
Independent Researcher\\
\email{lukeyf.ubc@gmail.com}
}

\begin{document}

\maketitle

\begin{abstract}
Whitney's theorem allows isomorphism testing for connected simple graphs,
apart from \(K_3\) and \(K_{1,3}\), to be formulated as distinguishing their
line graphs.  However, the relation between fixed-dimensional
Weisfeiler--Leman (WL) expressivity on line graphs and on their roots remains
unresolved.  We study this relation through Implicit Line-Graph WL
(ILG-\(k\)-WL), which is exactly \(k\)-WL on \(L(G)\), executed over the
edges of \(G\) with line-graph relations derived from endpoint incidence and
without explicitly constructing \(L(G)\).  On the Whitney-general class, the relation
between root-domain and line-graph WL depends on \(k\).  For \(k=1,2\),
ILG-\(k\)-WL adds no
distinguishing power beyond root-domain \(1\)-WL and misses some pairs that
\(1\)-WL separates.  For \(k=3\), we prove the backward containment
\(L(G)\equiv_{3\text{-WL}}L(H)\Rightarrow
G\equiv_{3\text{-WL}}H\).  Strongly regular witness pairs, including the
Shrikhande/rook pair, show that ILG-\(3\)-WL is strictly more expressive than
\(3\)-WL.  The backward containment also extends to disconnected graphs with
no isolated vertices when every connected component is Whitney-general.
Deterministic ILG-\(3\)-WL separates all three substructure-counting witness
pairs, all \(105\) pairs in SR25, and \(359\) of \(400\) BREC pairs.
An untrained dense ILG-\(3\)-GNN gives the same pairwise verdicts on these
evaluations.
\end{abstract}

\section{Introduction}

Graph neural networks (GNNs) are used in molecular property prediction
\citep{gilmer2017mpnn}, social-network analysis
\citep{hamilton2017graphsage}, combinatorial optimization
\citep{khalil2017combinatorial}, and program analysis
\citep{allamanis2018programs}.  Their ability to distinguish non-isomorphic
inputs matters independently of the learning objective: if a graph encoder
assigns two graphs the same representation, a downstream predictor cannot
distinguish them from that representation alone.  The Weisfeiler--Leman (WL)
hierarchy formalizes this limitation.  Standard message-passing neural
networks (MPNNs) are bounded by \(1\)-WL
\citep{xu2019gin,morris2019weisfeiler}, which fails to separate many regular
and highly symmetric graphs.  Higher-order GNNs use states indexed by multiple
vertices, as in higher-order WL refinement.  The \(k\)-GNN family operates on
vertex subsets \citep{morris2019weisfeiler}, whereas Provably Powerful Graph
Networks (PPGN) use pair-indexed feature matrices and have provable
\(3\)-WL expressive power \citep{maron2019ppgn}.

Graph transformations can also change what a fixed WL dimension detects.
The line graph \(L(G)\) represents every root edge by a vertex and joins two
such vertices when their edges share an endpoint.
For connected simple graphs \(G,H\), Whitney's theorem
\citep{whitney1932congruent} states that \(L(G)\cong L(H)\) implies
\(G\cong H\), except when \(\{G,H\}=\{K_3,K_{1,3}\}\).  This isomorphism
statement does not imply equal fixed-dimensional WL power on roots and their
line graphs.

Existing line-graph models address community detection
\citep{chen2019supervised}, link prediction \citep{cai2022linegraph}, joint
node--edge learning \citep{jiang2023coembedding}, and atomistic learning
\citep{choudhary2021alignn}.  Other methods transform the refinement domain:
LGAN forms local line graphs \citep{du2026lgan}, ESAN processes selected
subgraphs \citep{bevilacqua2022esan}, and \(\delta\)-\(k\)-LWL restricts tuple
replacements to local neighbors \citep{morris2020sparse}.  Simplicial and
cellular WL add cliques or cells and refine colors through their incidences
\citep{bodnar2021topological,bodnar2021cellular}.  These works do not determine
how fixed-dimensional WL on \(L(G)\) compares with the same dimension on
\(G\).

\begin{figure}[t]
    \centering
    \includegraphics{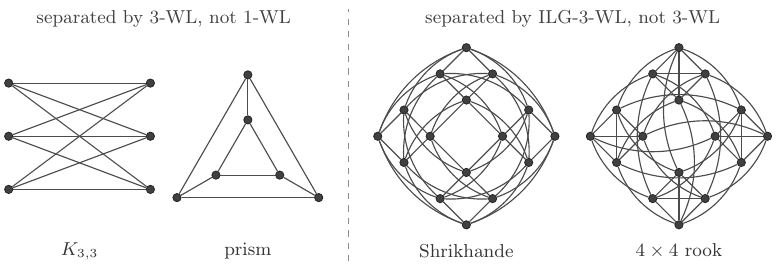}
    \caption{Graph pairs separated by \(3\)-WL (left) and ILG-\(3\)-WL
    (right).  \emph{Left:}
    \(3\)-WL separates \(K_{3,3}\) and the triangular prism, whereas
    \(1\)-WL does not.  \emph{Right:} ILG-\(3\)-WL separates the Shrikhande
    and \(4\times4\) rook graphs, whereas \(3\)-WL does not; their
    four-clique counts are \(0\) and \(8\), respectively.}
    \label{fig:hardness-pairs}
\end{figure}

This paper compares fixed-dimensional WL on \(G\) and \(L(G)\) through
Implicit Line-Graph WL (ILG-\(k\)-WL), which runs \(k\)-WL over root edges using
endpoint-derived line-graph relations without constructing \(L(G)\).
The Shrikhande and rook graphs illustrate why this comparison is nontrivial.
They are non-isomorphic strongly regular graphs with the same parameters and
form a standard hard case for \(3\)-WL
\citep{balcilar2021breaking,bodnar2021topological}.  Yet \(3\)-WL separates
their line graphs; line-graph transformation can break strong regularity
\citep{yang2024theoretical}.  Figure~\ref{fig:hardness-pairs} shows this pair.

\paragraph{Contributions.}
We determine the relation between root-domain and line-graph WL for
\(k\in\{1,2,3\}\) on the Whitney-general class.  ILG-\(k\)-WL is strictly
weaker than root-domain \(k\)-WL for \(k=1,2\)
(Corollary~\ref{cor:first-order-equiv}), whereas Theorem~\ref{thm:backward-transfer}
and strongly regular witnesses give
\(3\text{-WL}\prec\mathrm{ILG}\text{-}3\text{-WL}\)
(Corollary~\ref{cor:strict-three}).  One global pair update counts
four-cliques exactly on every finite simple graph
(Corollary~\ref{cor:universal-k4-counting}).  Backward transfer also extends
to disconnected graphs with no isolated vertices and Whitney-general
components (Theorem~\ref{thm:componentwise-backward}).  Both backward-transfer
theorems are formalized in Lean~4
(Appendix~\ref{app:machine-verification}).

Direct ILG-\(k\)-WL is \(k\)-WL on \(L(G)\).  For \(k\geq3\), its global-edge
form implements \((k{-}1)\)-FWL using \(\Theta(m^{k-1})\) root-edge tuple states
and endpoint-derived line-graph relations.  On bounded finite domains,
ILG-\(k\)-GNNs can match any fixed number of global-edge ILG-\(k\)-WL rounds,
and every parameterization is bounded by that refinement.

We evaluate deterministic ILG-\(3\)-WL and an untrained dense ILG-\(3\)-GNN on
the three substructure-counting witness pairs, all \(105\) SR25 pairs, and all
\(400\) BREC pairs.  Their verdicts coincide: \(3/3\), \(105/105\), and
\(359/400\) separations, respectively
(Section~\ref{sec:deterministic-verification}).

\section{Background and Conventions}
\label{sec:background}

Throughout, \(G\) and \(H\) denote finite, undirected, simple graphs, and
\(L(G)\) and \(L(H)\) denote their line graphs.  For one graph \(G\), we
write \(V(G)\) and \(E(G)\) for its vertex and edge sets and abbreviate
\(n=|V(G)|\) and \(m=|E(G)|\) when no second graph is present.  Edges are
unordered two-element subsets \(e=\{u,v\}\); an implementation may list the
edges in any order and store the two endpoints of each edge in either order.
For \(u\in V(G)\), let \(N_G(u)\) and \(d_G(u)\) denote its neighborhood and
degree; we omit the subscript only when the graph is unambiguous.
We write \(A_G\) for the adjacency matrix of \(G\), and \(I_S,J_S\) for the
identity and all-ones matrices indexed by a finite set \(S\).
Vertices and edges of \(G,H\) are called \emph{root vertices} and
\emph{root edges}; a root edge \(e\in E(G)\) is canonically the line-graph
vertex \(e\in V(L(G))\).  We write
\(\operatorname{Deg}(G)=\multiset{d_G(u):u\in V(G)}\) for the root degree
multiset.

The line graph \(L(G)\) has vertex set \(E(G)\), with two vertices adjacent
exactly when the corresponding distinct edges of \(G\) share an endpoint.

\begin{definition}[Whitney-general class]
\label{def:whitney-general}
A graph is \emph{Whitney-general} if it is connected, simple, and not
isomorphic to \(K_3\) or \(K_{1,3}\).  Whitney's theorem
\citep{whitney1932congruent} guarantees that its line graph determines it up
to isomorphism.  A graph is \emph{componentwise Whitney-general} if it has no
isolated vertices and every connected component is Whitney-general.
\end{definition}

\begin{definition}[Line-graph relation code]
For \(e,f\in E(G)\), write \(e\sim_{L(G)}f\) when \(e\neq f\) and
\(e\cap f\neq\emptyset\) --- that is, when \(e\) and \(f\) are adjacent in
\(L(G)\) --- and define the three-valued relation code
\[
    r_G(e,f)=
    \begin{cases}
        0, & e=f,\\
        1, & e\sim_{L(G)} f,\\
        2, & \text{otherwise}.
    \end{cases}
\]
\end{definition}

The \emph{graph-level signature} of a coloring \(c:S\to\mathcal C\) over a
finite state set \(S\) is its multiset of colors
\(\multiset{c(s):s\in S}\); a refinement procedure distinguishes two graphs
when their graph-level signatures differ after stabilization or after a
specified number of rounds.
All occurrences of \(\hash\) below mean a fixed injective encoding of the
displayed finite signature; only the induced color partition matters.  When
two graphs are compared, their refinements are run in parallel with the same
encoding of signatures, while every replacement sweep remains inside its own
graph.  Thus equal color names have a graph-independent meaning and no
vertexwise correspondence is assumed.

\paragraph{WL convention.}
Throughout, \(1\)-WL denotes ordinary vertex color refinement: from the
constant initial color \(c_0(v)\), it updates
\[
 c_{t+1}(v)=\hash\left(c_t(v),
   \multiset{c_t(w):w\in N_G(v)}\right).
\]
For every
\(k\geq2\), \(k\)-WL denotes the \emph{oblivious} \(k\)-dimensional WL
algorithm customary in graph learning
\citep{morris2019weisfeiler,feng2023fwl}: it colors ordered
\(k\)-tuples with \emph{separate per-coordinate replacement}. We use
``\(k\)-WL'' without further qualification for this convention throughout.
For a graph \(G\), a tuple
\(\mathbf{v}=(v_1,\ldots,v_k)\in V(G)^k\), and \(k\geq2\), initialize
\(c_0(\mathbf v)=\hash(\operatorname{atp}_G(\mathbf v))\), where the atomic
type records all equality and adjacency relations among the coordinates, and
update by
\[
    c_{t+1}(\mathbf{v})
    =
    \hash\left(
        c_t(\mathbf{v}),
        \left(
            \multiset{
                c_t(\mathbf{v}[i\leftarrow w])
                : w\in V(G)
            }
        \right)_{i=1}^{k}
    \right),
\]
where \(\mathbf{v}[i\leftarrow w]\) denotes the tuple with its \(i\)-th
coordinate replaced by~\(w\):
\[
 \mathbf v[i\leftarrow w]
 =(v_1,\ldots,v_{i-1},w,v_{i+1},\ldots,v_k).
\]
This procedure stores
\(\Theta(n^k)\) tuple states. For \(k\geq3\), the non-oblivious or folklore
algorithm on ordered \((k{-}1)\)-tuples, denoted \((k{-}1)\)-FWL, has the
same stable graph-distinguishing power as this separate-coordinate \(k\)-WL
algorithm \citep{feng2023fwl}. It uses
\emph{simultaneous} all-coordinate replacement and \(\Theta(n^{k-1})\) tuple
states. This is an arity-shift equivalence of graph distinguishability, not
a round-by-round identification of tuple colors. We use this
global-replacement form in Section~\ref{sec:global-edge}.

For either ordinary \(1\)-WL or the above \(k\)-WL convention, we write
\[
  G\equiv_{k\text{-WL}}H
\]
when the stable graph-level signatures of \(G\) and \(H\), computed with the
shared encoding convention above, are equal.  We write
\(G\not\equiv_{k\text{-WL}}H\) when they differ.  Unless features are stated
explicitly, all graphs are initially uncolored: \(1\)-WL starts from one
constant vertex color and \(k\)-WL for \(k\geq2\) starts from atomic types.
For two refinement procedures \(\mathsf A,\mathsf B\), we write
\(\mathsf A\preceq\mathsf B\) when every graph pair separated by \(\mathsf
A\) is also separated by \(\mathsf B\), and \(\mathsf A\prec\mathsf B\) when
the containment is strict.

Two low-dimensional conventions matter below. On finite simple graphs,
\(1\)-WL and \(2\)-WL have the same graph-distinguishing power
\citep{feng2023fwl}; and substituting \(k=1\) into the tuple rule
above would give a trivial global multiset, which is why ordinary \(1\)-WL
is defined separately.

\section{Implicit Line-Graph WL}

Since \(V(L(G))=E(G)\), \(k\)-WL on \(L(G)\) can be executed over the root
edges.  We give its first-order, direct higher-order, and arity-shifted
global-edge forms.

\subsection{ILG-\texorpdfstring{\(1\)}{1}-WL}

ILG-\(1\)-WL is ordinary \(1\)-WL on \(L(G)\): each root edge \(e\) starts
from a common color (or its given edge label) and updates from
\(\multiset{c_t(f):f\neq e,\ f\cap e\neq\emptyset}\).  Replacing the hash by
neural message, aggregation, and update maps gives ILG-\(1\)-GNN.

\begin{lemma}[ILG-\(1\)-GNN is an MPNN on the line graph]
\label{lem:ilg1gnn-linegraph-mpnn}
With shared initial features and message, update, aggregation, and readout
maps, an ILG-\(1\)-GNN on \(G\) and an MPNN on \(L(G)\) agree at every layer
and at graph readout.
\end{lemma}

\begin{proof}
Their neighborhoods coincide; see
Appendix~\ref{app:ilg1gnn-linegraph-mpnn}.
\end{proof}

\subsection{ILG-\texorpdfstring{\(k\)}{k}-WL}

For \(k\geq2\), direct ILG-\(k\)-WL applies the \(k\)-WL update of
Section~\ref{sec:background} to ordered tuples in \(E(G)^k\), with atomic
types computed from \(r_G\) and replacements ranging over \(E(G)\).

\begin{proposition}[ILG-\(k\)-WL equals \(k\)-WL on the line graph]
\label{prop:ilg-kwl}
For every graph \(G\) and \(k\geq1\), ILG-\(k\)-WL on \(G\) and \(k\)-WL on
\(L(G)\) produce identical state partitions after each round, up to canonical
renaming of colors.
\end{proposition}

\begin{proof}
Under \(V(L(G))=E(G)\), the domains, atomic types, and replacement ranges
coincide; see Appendices~\ref{app:atomic-types} and~\ref{app:ilg-kwl-proof}.
\end{proof}

\subsection{Global-Edge Form with \texorpdfstring{\(\Theta(m^{k-1})\)}{Theta(m\^{}(k-1))} Tuple-State Storage}
\label{sec:global-edge}

Direct ILG-\(k\)-WL uses \(\Theta(m^k)\) tuple states and
\(\Theta(m^{k+1})\) replacement entries per round.  Applying the standard
arity shift to \(L(G)\) gives the same stable distinguishing power on
\(\Theta(m^{k-1})\) tuples \citep{feng2023fwl}; its \(k=3\) instance underlies
PPGN \citep{maron2019ppgn}.  Global-edge ILG-\(k\)-WL is this
\((k{-}1)\)-FWL formulation on ordered root-edge tuples: for each
\(g\in E(G)\), its message collects the \(k{-}1\) coordinate-wise
replacements \(c_t(\mathbf e[i\leftarrow g])\).  Appendix
\ref{app:global-edge-2fwl} gives the full recurrence.

This form uses \(\Theta(m^{k-1})\) tuple states and \(\Theta(m^k)\)
replacement entries per round.  Endpoint comparisons evaluate \(r_G(e,f)\)
in \(O(1)\) time from a \(\Theta(m)\)-space \(m\times2\) table, avoiding
\(\Theta(m+m_L)\) stored line-graph adjacency without changing tuple-state or
update asymptotics.  Appendices~\ref{app:complexity}
and~\ref{app:empirical-complexity} give the details and measurements.

\begin{corollary}[Global-edge ILG-\(k\)-WL has \(k\)-WL graph-distinguishing power on \(L(G)\)]
\label{cor:global-edge-fwl}
For every simple graph \(G\) and every \(k\geq3\), global-edge ILG-\(k\)-WL
has the same graph-level distinguishing power as \(k\)-WL on \(L(G)\), using
\(\Theta(m^{k-1})\) rather than \(\Theta(m^k)\) tuple states.
\end{corollary}

\begin{proof}
The recurrence is exactly \((k{-}1)\)-FWL on \(L(G)\)
(Appendix~\ref{app:global-edge-2fwl}), and the arity shift
\citep{feng2023fwl} gives the claim.
\end{proof}

\section{Expressivity Analysis}
\label{sec:expressivity}

ILG-\(k\)-WL is strictly weaker than root-domain \(k\)-WL for \(k=1,2\), but
strictly stronger for \(k=3\), on Whitney-general graphs.  At \(k=1,2\), some
pairs separated on the roots are not separated on their line graphs, even
when the roots have equal order.  At \(k=3\), joint pair refinement on the
line graphs determines the stable ordered-pair coloring of the roots and
separates some pairs that root-domain \(3\)-WL does not.  The backward result
extends componentwise to the disconnected class in
Definition~\ref{def:whitney-general}.

\subsection{First-Order: Loss at \texorpdfstring{\(k=1\)}{k=1}}
\label{sec:first-order-gap}

The restriction to Whitney-general roots removes the classical
\(K_3,K_{1,3}\) ambiguity.  Any loss shown below therefore occurs even though
the line graph determines its root up to isomorphism.

\begin{proposition}[\(1\)-WL equivalence transfers to line graphs]
\label{prop:forward-transfer}
For all simple graphs \(G,H\) (without restriction):
\[
    G\equiv_{1\text{-WL}} H \quad\Longrightarrow\quad L(G)\equiv_{1\text{-WL}} L(H).
\]
Thus, ILG-\(1\)-WL can separate a pair only if root-domain \(1\)-WL also
separates it.
\end{proposition}

\begin{proof}[Proof sketch]
By the fractional-isomorphism/equitable-partition characterization
\citep{ramana1994fractional,grohe2017descriptive},
\(G\equiv_{1\text{-WL}} H\) iff their stable \(1\)-WL colorings have
matching class sizes and between-class neighbor counts. The table of these
counts is the equitable quotient. The induced edge coloring
\(\widehat c_G(e)=\multiset{c_G(u),c_G(v)}\), where \(c_G\) is the stable
vertex coloring, is an equitable partition of \(L(G)\) whose class sizes and
equitable quotient are determined by the \(1\)-WL data of~\(G\). Since these
agree for \(G\) and \(H\), the stable \(1\)-WL colorings of \(L(G)\) and \(L(H)\)
coincide, giving \(L(G)\equiv_{1\text{-WL}} L(H)\). See
Appendix~\ref{app:forward-transfer} for the full argument.
\end{proof}

\begin{proposition}[Equal-order backward failure at first order]
\label{prop:first-order-loss}
There exist Whitney-general graphs \(G,H\) with
\(L(G)\equiv_{1\text{-WL}} L(H)\) but \(G\not\equiv_{1\text{-WL}} H\).
\end{proposition}

\begin{figure}[t]
    \centering
    \includegraphics{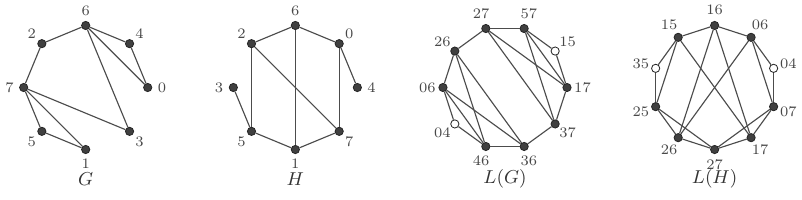}
    \caption{Equal-order backward failure at first order.  Root \(1\)-WL
    separates \(G,H\) by their degree multisets, but \(L(G),L(H)\) have the
    same stable equitable quotient with cells of sizes \(2,4,4\).  Line
    vertices are labeled by their root edges; hollow and filled vertices have
    degrees two and four.}
    \label{fig:first-order-pair}
\end{figure}

\begin{proof}
The pair and its line graphs are displayed in
Figure~\ref{fig:first-order-pair}; both graphs are Whitney-general, and
their root degree multisets differ, so \(1\)-WL separates \(G\) and
\(H\) in its first round. In either line graph, degree
refinement followed by one neighbor-color refinement stabilizes with cells
of sizes \(2,4,4\) and the same equitable quotient
\(A{:}\,(0,2,0)\), \(B{:}\,(1,1,2)\), \(C{:}\,(0,2,2)\), so
\(L(G)\equiv_{1\text{-WL}}L(H)\). The computation is verified in
Appendix~\ref{app:first-order-transfer}.
\end{proof}

\begin{remark}[Why \(k=1\) and \(k=2\) fail]
\label{rem:first-order-failure}
Root order is not the only reason backward transfer fails: a line color records
endpoint-degree sums and an equitable quotient of edge neighborhoods, but
not how those colors are distributed among root stars. Joint pair
refinement at \(k=3\) additionally records common-neighbor data for ordered
edge pairs, from which Section~\ref{sec:backward-transfer} recovers the
stable ordered-pair coloring of the root.
\end{remark}

\begin{corollary}[ILG is strictly weaker for \(k=1,2\)]
\label{cor:first-order-equiv}
On the Whitney-general class:
\[
  \text{ILG-}1\text{-WL} \prec 1\text{-WL},
  \qquad
  \text{ILG-}2\text{-WL} \prec 2\text{-WL}.
\]
\end{corollary}

\begin{proof}
The first statement follows from Propositions~\ref{prop:forward-transfer}
and~\ref{prop:first-order-loss}. For the second, apply the equivalence of
\(1\)-WL and \(2\)-WL \citep{feng2023fwl} both to the roots and to
their line graphs, and use Proposition~\ref{prop:ilg-kwl}.
\end{proof}

\subsection{Strictness at \texorpdfstring{\(k=3\)}{k=3}}

An \(\operatorname{srg}(v,d,\lambda,\mu)\) is \(d\)-regular on \(v\) vertices,
with \(\lambda\) and \(\mu\) common neighbors for adjacent and nonadjacent
vertex pairs, respectively.

On two strongly regular graphs with the same parameters, the initial pair
types and all simultaneous-replacement
counts depend only on equality and adjacency.  Induction leaves the
\(2\)-FWL colors parameter-determined.  The arity shift then makes
co-parametric non-isomorphic strongly regular graphs standard examples not
separated by \(3\)-WL \citep{feng2023fwl}.

\begin{corollary}[Universal four-clique counting]
\label{cor:universal-k4-counting}
For every finite simple graph \(G\), one global pair-refinement round of
ILG-\(3\)-WL determines the exact number \(\operatorname{cl}_4(G)\) of
four-vertex cliques.  More precisely,
\[
 \operatorname{cl}_4(G)=\frac{1}{6}
 \left|\left\{(e,f)\in E(G)^2:
 e\cap f=\varnothing,\quad
 \bigl|N_{L(G)}(e)\cap N_{L(G)}(f)\bigr|=4
 \right\}\right|.
\]
Consequently, ILG-\(3\)-WL-equivalent finite simple graphs have the same
four-clique count.
\end{corollary}

\begin{proof}
The first pair color records initial nonadjacency and the common-neighbor
multiplicity.  The resulting factor-six double count is proved in
Appendix~\ref{app:k4-counting-proof}.
\end{proof}

\begin{proposition}[Shrikhande/rook separation at \(k=3\)]
\label{prop:ilg-gains}
On the Whitney-general class, ILG-\(3\)-WL can separate graph pairs that \(3\)-WL on the root graph cannot.
\end{proposition}

\begin{proof}
The Shrikhande and \(4\times4\) rook graphs are Whitney-general,
co-parametric \(\operatorname{srg}(16,6,2,2)\) graphs, so their root
\(3\)-WL signatures agree.  The Shrikhande graph has no \(K_4\), whereas the
rook graph has eight.  By Corollary~\ref{cor:universal-k4-counting}, the
corresponding ordered-pair color therefore has multiplicities \(0\) and
\(48\), so pair refinement separates their line graphs, as does \(3\)-WL by
the arity shift.
Proposition~\ref{prop:ilg-kwl} identifies the direct line-graph test with
ILG-\(3\)-WL.
\end{proof}

\subsection{Backward Transfer on Whitney-General Graphs at \texorpdfstring{\(k=3\)}{k=3}}
\label{sec:backward-transfer}

The preceding witness shows that ILG-\(3\)-WL can be stronger, but not that it
preserves every distinction made by root \(3\)-WL.  The next theorem proves
that it does on Whitney-general graphs.

\begin{theorem}[Backward transfer at \(3\)-WL]
\label{thm:backward-transfer}
For Whitney-general graphs \(G,H\),
\[
  L(G)\equiv_{3\text{-WL}} L(H)
  \;\Longrightarrow\;
  G\equiv_{3\text{-WL}} H.
\]
\end{theorem}

\begin{proof}[Proof idea]
Let \(B_G\) be the unsigned vertex--edge incidence matrix.  The identity
\(B_G^\top B_G=2I+A_{L(G)}\) shows that \(2I+A_{L(G)}\) is the Gram matrix
of the root incidence vectors.  At \(k=3\), stable line-pair colors span a
matrix algebra closed under multiplication.  The Gram identity therefore
lets us recover endpoint degrees and local endpoint profiles without choosing
the shared endpoints globally.  Stable \(1\)-WL records vertex color classes
and their neighbor counts but provides no such coherent pair algebra, so the
same reconstruction is unavailable below \(k=3\).
Appendix~\ref{app:backward-transfer-proof} gives the full argument in five
stages.
\begin{enumerate}
\item \emph{Line algebra (\S A.1).} The indicator matrices of stable
\(2\)-FWL colors span a coherent matrix algebra.  Line-graph equivalence
induces an isomorphism between the two such algebras.

\item \emph{Local root data (\S A.2).} The line algebra recovers endpoint
degrees, star profiles (incident-edge counts by stable line color), and dart
profiles (relation counts at a specified endpoint).  Small roots \(K_1,K_2\),
the paw, and the diamond are treated separately, while \(K_3,K_{1,3}\) are
excluded.

\item \emph{Unequal endpoint profiles (\S A.3).} A common ordering of the
transported profile labels orients these edges in both roots.  Their signed
incidence matrices then define mixed Gram matrices that are mapped by the
line-algebra isomorphism.

\item \emph{Equal endpoint profiles (\S A.4).} No canonical orientation is
available.  A cancellation identity shows that either endpoint choice acts
identically on matrices in the recovered root space.  This recovers the cells
for equal and adjacent root-vertex pairs.

\item \emph{All distances (\S A.5--A.6).} Induction on root distance recovers
the remaining pair cells, whose disjoint \(0\)-\(1\) supports give Hadamard
closure and hence the stable root pair algebra.
\end{enumerate}
The arity shift \citep{feng2023fwl} converts this
\(2\)-FWL statement into Theorem~\ref{thm:backward-transfer}.
\end{proof}

\begin{corollary}[Connected graphs and the Whitney exception]
\label{cor:connected-completion}
Let \(G,H\) be connected finite simple graphs.  If
\(L(G)\equiv_{3\text{-WL}}L(H)\), then either
\(G\equiv_{3\text{-WL}}H\) or
\(\{G,H\}=\{K_3,K_{1,3}\}\).
\end{corollary}

\begin{proof}
If neither root belongs to the Whitney pair, both are Whitney-general and
Theorem~\ref{thm:backward-transfer} applies.  If one root belongs to that
pair, its line graph is \(K_3\).  Stable \(3\)-WL equivalence preserves the
order and degree multiset of the line graph, so the other line graph is also
\(K_3\); Whitney's theorem leaves exactly the two stated roots.
\end{proof}

\begin{theorem}[Backward transfer for disconnected graphs]
\label{thm:componentwise-backward}
For componentwise Whitney-general graphs \(G,H\),
\[
  L(G)\equiv_{3\text{-WL}}L(H)
  \;\Longrightarrow\;
  G\equiv_{3\text{-WL}}H.
\]
\end{theorem}

\begin{proof}[Proof idea]
Stable pair refinement matches the components of \(L(G)\) and \(L(H)\).
Since the roots have no isolated vertices, these are exactly the line graphs
of their root components, so Theorem~\ref{thm:backward-transfer} applies to
each matched pair.  Combining the component equivalences proves the claim;
see Appendix~\ref{app:componentwise-backward-extension}.
\end{proof}

\begin{corollary}[Strict containment at \(3\)-WL]
\label{cor:strict-three}
On both the Whitney-general and componentwise Whitney-general classes,
\[
  3\text{-WL}\ \prec\ \text{ILG-}3\text{-WL}.
\]
\end{corollary}

\begin{proof}
Theorems~\ref{thm:backward-transfer} and
\ref{thm:componentwise-backward} give containment, and
Proposition~\ref{prop:ilg-gains} makes it strict.
\end{proof}

\begin{remark}[Comparison with root \(4\)-WL]
\label{rem:four-wl-incomparable}
Root \(4\)-WL and ILG-\(3\)-WL are incomparable, even on connected
Whitney-general simple graphs.  Root \(4\)-WL distinguishes
the incidence graphs of two nonisomorphic
Steiner triple systems of order \(13\), whereas ILG-\(3\)-WL does not.
Conversely, ILG-\(3\)-WL distinguishes a compressed parity pair over \(K_5\)
that root \(4\)-WL does not.  Appendix~\ref{app:four-wl-incomparability}
gives both witnesses and their exact traces.
\end{remark}

\section{Higher-Order Neural Architecture}
\label{sec:global-edge-kgnn}

The global-edge recurrence defines a neural model on the same tuple domain.
For \(k\geq3\), one hidden state is stored for each ordered
\((k{-}1)\)-tuple of root edges.  A message associated with replacement edge
\(g\) jointly encodes the \(k-1\) states obtained by placing that same \(g\)
in each coordinate.  We display the first nontrivial case, \(k=3\), whose
states are ordered edge pairs.

Let \(q_e\in\mathbb{R}^{d_e}\) be an edge feature.  All \(q_e\) are the same
constant on an unlabeled graph.  Finite vertex labels can instead be included
in \(q_e\) as the unordered pair of endpoint labels; isolated vertices are
outside the line-domain representation.  This uses labels, not unique vertex
or edge identifiers.  Let
\(\relEmb:\{0,1,2\}\to\mathbb R^{d_r}\) embed the line-graph relation code;
an injective one-hot map suffices for exact WL simulation.  Let \(d_h\) be the
hidden-state width, so \(h_{ef}^{(\ell)}\in\mathbb R^{d_h}\).  Pair states
start from
\[
    h_{ef}^{(0)}
    =
    \phi_0\left(q_e, q_f, \relEmb(r_G(e,f))\right),
    \qquad e,f\in E(G).
\]
One layer computes
\[
    M_{ef}^{(\ell)}
    =
    \sum_{g\in E(G)}
    \psi_\ell\left(
        h_{gf}^{(\ell)},
        h_{eg}^{(\ell)},
        q_g,
        \relEmb(r_G(g,e)),
        \relEmb(r_G(g,f))
    \right),
\]
then updates
\[
    h_{ef}^{(\ell+1)}
    =
    \eta_\ell\left(h_{ef}^{(\ell)}, M_{ef}^{(\ell)}\right).
\]
After \(T_{\mathrm{net}}\) layers, a graph readout may combine pair and edge
summaries:
\[
    h_G
    =
    \operatorname{READOUT}\left(
        \sum_{e,f\in E(G)}\gamma_{\mathrm{read}}(h_{ef}^{(T_{\mathrm{net}})}),
        \sum_{e\in E(G)}\xi(q_e)
    \right).
\]
Our implementations use multilayer perceptrons (MLPs) for the maps above:
\(\phi_0\) initializes pair states, \(\psi_\ell\) produces messages,
\(\eta_\ell\) updates states, and \(\gamma_{\mathrm{read}},\xi\) feed the
invariant readout.  The corresponding WL color already contains \(q_g\) and
the relation codes, so the neural message receives no extra structural
information.

Vertex relabeling leaves the graph representation unchanged, and reordering
the internal edge list only permutes tuple states
(Appendix~\ref{app:equivariance}).  Appendix~\ref{app:algo-spec} gives the
general-\(k\) forward pass.

\begin{proposition}[Expressivity of global-edge ILG-\(k\)-GNN]
\label{prop:expressivity}
Fix \(k\geq3\), a fixed number of rounds, and a graph class with bounded edge
count and a finite alphabet of edge features.  Finite vertex labels may be
included when they are encoded invariantly in the edge features as above.
There are choices of
\(\relEmb,\phi_0,\psi_\ell,\eta_\ell,\gamma_{\mathrm{read}},\xi\), and
\(\operatorname{READOUT}\) for which
the summed multiset encodings and final readout are injective on every finite
domain encountered.  With these choices, the GNN simulates the corresponding
feature-labeled global-edge ILG-\(k\)-WL refinement.  Conversely, if two
feature-labeled graphs have identical refinement histories, no instance whose
messages use only the stated tuple states, edge features, and relation codes
can separate them through an invariant readout.
\end{proposition}

\begin{proof}[Proof sketch]
Bounded edge count and a finite feature alphabet allow only finitely many
replacement entries.  The multiset-encoding lemma of \citet{xu2019gin}
therefore gives injective summed encodings, and induction over layers assigns
each refinement color a distinct hidden state.  Conversely, every hidden
state is a function of its refinement color, so equal color histograms force
equal invariant readouts.  Appendix~\ref{app:neural-simulation} gives the full
construction.
\end{proof}

\section{Expressivity Evaluation}
\label{sec:deterministic-verification}

\paragraph{Evaluation objective.}
We report graph-pair separations by deterministic refinements and neural
realizations, together with the relevant WL upper bounds.
Table~\ref{tab:separation} covers the three substructure-counting witness
pairs, all \(105\) SR25 pairs, and all \(400\) BREC pairs.

\subsection{Deterministic Refinement}
\label{sec:exact-refinement}

For each deterministic method computed in this study, two graphs are separated
when their graph-level signatures differ.  A pair is reported as not separated
only after refinement stabilizes.  The mathematical recurrence uses exact
canonical signatures.  Regression tests on a finite collection of graphs
compare ILG with WL on explicitly constructed line graphs, checking the
implementation against Proposition~\ref{prop:ilg-kwl}.

SR25 contains \(15\) graphs with parameters
\(\operatorname{srg}(25,12,5,6)\), giving \(105\) unordered pairs
\citep{balcilar2021breaking}.  BREC contains \(400\) specified
non-isomorphic pairs whose main graphs are constructed to be
\(1\)-WL-indistinguishable
\citep{wang2024brec}.  Its CFI category is based on the
Cai--F\"urer--Immerman construction
\citep{wang2024brec,cai1992cfi}.  The three substructure-counting witness
pairs are \(C_6\) versus
\(C_3\sqcup C_3\), \(C_8\) versus \(C_4\sqcup C_4\), and Shrikhande versus
the \(4\times4\) rook graph.  The pairs differ in their numbers of triangles,
four-cycles, and four-cliques, respectively.  If a method does not separate
one of these pairs, it cannot count that pattern on all graphs; separation
alone does not prove that it can.  For four-cliques,
Corollary~\ref{cor:universal-k4-counting} supplies the stronger universal
counting result.  In the neural block,
``set \(3\)-WL'' denotes the subset-based refinement associated with Morris
\(3\)-GNN, not the ordered-tuple \(3\)-WL used elsewhere in the paper.
The table's \emph{WL bound} column lists the deterministic refinement whose
distinguishing power upper-bounds that of each neural architecture.  A dash
marks a deterministic method.

\begin{table}[t]
    \centering
    \caption{Deterministic refinement results and neural separation counts on
    the evaluated datasets and graph pairs.  BREC is split into four aggregate categories:
    Regular combines the Simple Regular,
    Strongly Regular, 4-Vertex Condition, and Distance Regular families;
    Extension and CFI combine their respective BREC subfamilies.  Benchmark
    entries count separated pairs; the last three columns give one verdict for
    each substructure-counting witness pair.  The PPGN BREC entry
    \({}^{\dagger}\)
    is sourced from
    \citet{wang2024brec}.  The ILG-\(3\)-GNN entry \({}^{\ddagger}\) reports an
    untrained dense evaluation on the three substructure-counting witness
    pairs, all \(105\) SR25 pairs, and all \(400\) BREC main pairs.  Neural
    rows use the protocols in
    Appendix~\ref{app:neural-evaluation-protocols}.}
    \label{tab:separation}
    \footnotesize
    \setlength{\tabcolsep}{0pt}
    \renewcommand{\arraystretch}{1.05}
    \begin{tabular*}{\textwidth}{@{\extracolsep{\fill}}ll*{6}{c}
        *{3}{>{\centering\arraybackslash}p{0.065\textwidth}}@{}}
        \toprule
        & & & \multicolumn{5}{c}{BREC pairs}
        & \multicolumn{3}{c}{\shortstack{Substructure-counting\\witness pairs}} \\
        \cmidrule(lr){4-8}\cmidrule(lr){9-11}
        Method & WL bound & SR25
        & Basic & Regular & Extension & CFI & All
        & \(C_3\) & \(C_4\) & \(K_4\) \\
        \midrule
        \multicolumn{11}{@{}l}{\emph{Deterministic refinement}} \\
        ILG-\(1\)-WL & ---
        & 0/105 & 0/60 & 0/140 & 0/100 & 0/100 & 0/400
        & \xmark & \xmark & \xmark \\
        \(3\)-WL & ---
        & 0/105 & 60/60 & 50/140 & 100/100 & 60/100 & 270/400
        & \cmark & \cmark & \xmark \\
        ILG-\(3\)-WL & ---
        & 105/105 & 60/60 & 139/140 & 100/100 & 60/100 & 359/400
        & \cmark & \cmark & \cmark \\
        \midrule
        \multicolumn{11}{@{}l}{\emph{Neural architectures}} \\
        ILG-\(1\)-GNN & ILG-\(1\)-WL
        & 0/105 & 0/60 & 0/140 & 0/100 & 0/100 & 0/400
        & \xmark & \xmark & \xmark \\
        Morris \(3\)-GNN & set \(3\)-WL
        & 0/105 & 60/60 & 50/140 & 97/100 & 0/100 & 207/400
        & \cmark & \cmark & \xmark \\
        PPGN & \(3\)-WL
        & 0/105 & 60/60 & 50/140 & 100/100 & 23/100
        & 233/400\({}^{\dagger}\)
        & \cmark & \cmark & \xmark \\
        ILG-\(3\)-GNN\({}^{\ddagger}\) & ILG-\(3\)-WL
        & 105/105 & 60/60 & 139/140 & 100/100 & 60/100
        & 359/400
        & \cmark & \cmark & \cmark \\
        \bottomrule
    \end{tabular*}
\end{table}

On the substructure-counting witness pairs, \(3\)-WL separates the triangle and four-cycle pairs
but not Shrikhande/rook.  Since the latter roots have different four-clique
counts, their shared signature proves that \(3\)-WL cannot count
four-cliques in general.  ILG-\(3\)-WL separates all three pairs.  The same
difference is universal for four-clique counts by
Corollary~\ref{cor:universal-k4-counting}.  On SR25, root-domain \(3\)-WL
assigns one signature to all \(15\) graphs, whereas ILG-\(3\)-WL separates
every pair.

On BREC, ILG-\(3\)-WL improves over root-domain \(3\)-WL only in the Regular
aggregate.  It separates \(139/140\) pairs there, broken down as \(50/50\)
Simple Regular, \(50/50\) Strongly Regular, \(20/20\) 4-Vertex Condition,
and \(19/20\) Distance Regular; root-domain \(3\)-WL separates \(50/140\).
\(3\)-WL already separates every pair in Basic and Extension, and both
refinements separate \(60/100\) CFI pairs.

The remaining \(41\) pairs comprise \(20\) CFI-\(3\)-WL pairs, \(20\)
CFI-\(4\)-WL pairs, and one Distance Regular pair.  On the tested CFI blocks,
ILG-\(3\)-WL and root-domain \(3\)-WL have identical pairwise outcomes.  We
make no claim about other CFI families.  The only Distance Regular pair not
separated is BREC pair \(380\) (zero-based); both refinements stabilize
without distinguishing it.

\subsection{Neural Evaluation}
\label{sec:neural-verification}

The ILG-\(1\)-GNN and Morris \(3\)-GNN rows use pair-specific training and the
evaluation protocols detailed in
Appendix~\ref{app:neural-evaluation-protocols}.  The PPGN category counts,
totaling \(233\) pairs, are taken from \citet[Table~2]{wang2024brec}.

For ILG-\(3\)-GNN, we evaluate an untrained dense model separately on each of
the three substructure-counting witness pairs, the \(105\) SR25 pairs, and the
\(400\) BREC main pairs,
using the same fixed parameters for both graphs in a pair.  The model has
constant uncolored edge features, invariant
equality/incident/disjoint relation codes, width \(32\), five layers, sine
activations, and a multiscale readout.  Each graph is evaluated in its stored
ordering and under three relabelings.  A pair is counted as separated only
when its minimum cross-graph distance exceeds both within-graph diameters by
the fixed numerical tolerance.  Every separation uses all four copies; a
failed pair may stop after two.  The resulting
verdicts are \(3/3\) on the substructure-counting witness pairs, \(105/105\) on SR25, and
\(359/400\) on BREC.  Across all three suites, the neural pairwise verdicts coincide with
deterministic ILG-\(3\)-WL.

All SR25 pairs and all three substructure-counting witness pairs are
\(1\)-WL-indistinguishable.
ILG-\(1\)-GNN separates none of them, consistent with
Lemma~\ref{lem:ilg1gnn-linegraph-mpnn} and
Proposition~\ref{prop:forward-transfer}.
Morris \(3\)-GNN uses set-based refinement rather than the ordered-tuple
\(3\)-WL studied here.  It separates \(0/105\) SR25 pairs and \(207/400\)
BREC pairs, separating the triangle and four-cycle pairs but not the
four-clique pair.  Appendix~\ref{app:neural-evaluation-protocols}
specifies the trained and untrained protocols and their invariance checks.

\section{Discussion and Conclusion}

The measurements in Appendix~\ref{app:empirical-complexity} follow the
complexity analysis: explicit and implicit line-domain refinement have the
same pair-state and per-round update costs, while endpoint access avoids
storing the line-graph structure.  On SR25, root \(4\)-WL is faster than
ILG-\(3\)-WL in the shared implementation.  This is a cost comparison, not an
expressivity ordering: the two methods are incomparable, and ILG-\(3\)-WL
distinguishes pairs that root \(4\)-WL misses
(Remark~\ref{rem:four-wl-incomparable}).

On Whitney-general graphs, ILG-\(k\)-WL is strictly weaker than root-domain
\(k\)-WL for \(k=1,2\), but strictly stronger for \(k=3\); the latter strict
containment also holds on componentwise Whitney-general graphs.  At \(k=3\),
ILG-\(3\)-WL counts four-cliques on every finite simple root graph
(Corollary~\ref{cor:universal-k4-counting}) and separates more pairs than
root-domain \(3\)-WL on SR25 and in BREC's Regular aggregate.  On the tested CFI blocks, the root and line-domain
refinements have the same pairwise outcomes.  The backward-transfer theorems
and Shrikhande/rook pair prove
\(3\text{-WL}\prec\mathrm{ILG}\text{-}3\text{-WL}\) on both classes.

The untrained dense ILG-\(3\)-GNN reproduces the deterministic ILG-\(3\)-WL
verdicts on the substructure-counting witness pairs, SR25, and BREC.  This
comparison tests graph-isomorphism expressivity rather than supervised
prediction.  Our
\(k=1\) result shows that ILG-\(1\)-WL is strictly weaker than root \(1\)-WL on
Whitney-general graphs.  Other line-graph networks report improved task
performance.  This may arise because the line-graph representation makes
useful structural information easier for the model to use, together with
task-specific edge features, training, or model bias.
Many standard graph-classification datasets contain few hard \(1\)-WL
collisions \citep{zopf2022wl}.
Since edges are the native line-graph objects, link-level tasks provide a
natural setting for downstream evaluation \citep{lachi2025link}.

Backward transfer for \(k\geq4\) remains open.  The \(k=3\) proof reconstructs
the coherent matrix algebra of stable pair refinement; extending it would
require a corresponding reconstruction of higher-arity root tuple colors and
simultaneous endpoint choices.  Repeated line-graph
transformation may also simplify isomorphism testing with \(3\)-WL: at most two
transformations take every connected graph other than \(C_3,C_4,C_5\) outside
the strongly regular class \citep{yang2024theoretical}.  Determining whether
this allows \(3\)-WL to distinguish additional graph pairs is an interesting
direction for future work.

\bibliographystyle{plainnat}
\bibliography{reference}

@inproceedings{maron2019ppgn,
  title = {Provably Powerful Graph Networks},
  author = {Maron, Haggai and Ben-Hamu, Heli and Serviansky, Hadar and Lipman, Yaron},
  booktitle = {Advances in Neural Information Processing Systems},
  pages = {2156--2167},
  publisher = {Curran Associates, Inc.},
  volume = {32},
  year = {2019},
  url = {https://proceedings.neurips.cc/paper/2019/hash/bb04af0f7ecaee4aae62035497da1387-Abstract.html}
}

@inproceedings{morris2019weisfeiler,
  title = {{Weisfeiler} and {Leman} Go Neural: Higher-Order Graph Neural Networks},
  author = {Morris, Christopher and Ritzert, Martin and Fey, Matthias and Hamilton, William L. and Lenssen, Jan Eric and Rattan, Gaurav and Grohe, Martin},
  booktitle = {Proceedings of the AAAI Conference on Artificial Intelligence},
  volume = {33},
  pages = {4602--4609},
  year = {2019},
  doi = {10.1609/aaai.v33i01.33014602},
  url = {https://ojs.aaai.org/index.php/AAAI/article/view/4384}
}

@inproceedings{bevilacqua2022esan,
  title = {Equivariant Subgraph Aggregation Networks},
  author = {Bevilacqua, Beatrice and Frasca, Fabrizio and Lim, Derek and Srinivasan, Balasubramaniam and Cai, Chen and Balamurugan, Gopinath and Bronstein, Michael M. and Maron, Haggai},
  booktitle = {International Conference on Learning Representations},
  year = {2022},
  url = {https://openreview.net/forum?id=dFbKQaRk15w}
}

@inproceedings{morris2020sparse,
  title = {{Weisfeiler} and {Leman} Go Sparse: Towards Scalable Higher-Order Graph Embeddings},
  author = {Morris, Christopher and Rattan, Gaurav and Mutzel, Petra},
  booktitle = {Advances in Neural Information Processing Systems},
  pages = {21824--21840},
  publisher = {Curran Associates, Inc.},
  volume = {33},
  year = {2020},
  url = {https://proceedings.neurips.cc/paper_files/paper/2020/hash/f81dee42585b3814de199b2e88757f5c-Abstract.html}
}

@inproceedings{bodnar2021topological,
  title = {{Weisfeiler} and {Lehman} Go Topological: Message Passing Simplicial Networks},
  author = {Bodnar, Cristian and Frasca, Fabrizio and Wang, Yuguang and Otter, Nina and Mont{\'u}far, Guido F. and Li{\`o}, Pietro and Bronstein, Michael M.},
  booktitle = {Proceedings of the 38th International Conference on Machine Learning},
  pages = {1026--1037},
  volume = {139},
  series = {Proceedings of Machine Learning Research},
  publisher = {PMLR},
  year = {2021},
  url = {https://proceedings.mlr.press/v139/bodnar21a.html}
}

@inproceedings{bodnar2021cellular,
  title = {{Weisfeiler} and {Lehman} Go Cellular: {CW} Networks},
  author = {Bodnar, Cristian and Frasca, Fabrizio and Otter, Nina and Wang, Yuguang and Li{\`o}, Pietro and Mont{\'u}far, Guido and Bronstein, Michael},
  booktitle = {Advances in Neural Information Processing Systems},
  pages = {2625--2640},
  publisher = {Curran Associates, Inc.},
  volume = {34},
  year = {2021},
  url = {https://proceedings.neurips.cc/paper_files/paper/2021/hash/157792e4abb490f99dbd738483e0d2d4-Abstract.html}
}

@inproceedings{wang2024brec,
  title = {An Empirical Study of Realized {GNN} Expressiveness},
  author = {Wang, Yanbo and Zhang, Muhan},
  booktitle = {Proceedings of the 41st International Conference on Machine Learning},
  pages = {52134--52155},
  year = {2024},
  editor = {Salakhutdinov, Ruslan and Kolter, Zico and Heller, Katherine and Weller, Adrian and Oliver, Nuria and Scarlett, Jonathan and Berkenkamp, Felix},
  volume = {235},
  series = {Proceedings of Machine Learning Research},
  month = {21--27 Jul},
  publisher = {PMLR},
  url = {https://proceedings.mlr.press/v235/wang24cl.html}
}

@inproceedings{balcilar2021breaking,
  title = {Breaking the Limits of Message Passing Graph Neural Networks},
  author = {Balcilar, Muhammet and H{\'e}roux, Pierre and Ga{\"u}z{\`e}re, Benoit and Vasseur, Pascal and Adam, S{\'e}bastien and Honeine, Paul},
  booktitle = {Proceedings of the 38th International Conference on Machine Learning},
  pages = {599--608},
  volume = {139},
  series = {Proceedings of Machine Learning Research},
  publisher = {PMLR},
  year = {2021},
  url = {https://proceedings.mlr.press/v139/balcilar21a.html}
}

@inproceedings{yang2024theoretical,
  title = {Theoretical Insights into Line Graph Transformation on Graph Learning},
  author = {Yang, Fan and Huang, Xingyue},
  booktitle = {NeurIPS 2024 Workshop on Symmetry and Geometry in Neural Representations},
  year = {2024},
  note = {Poster},
  url = {https://neurips.cc/virtual/2024/101444}
}

@inproceedings{chen2019supervised,
  title = {Supervised Community Detection with Line Graph Neural Networks},
  author = {Chen, Zhengdao and Li, Lisha and Bruna, Joan},
  booktitle = {International Conference on Learning Representations},
  year = {2019},
  url = {https://openreview.net/forum?id=H1g0Z3A9Fm}
}

@article{jiang2023coembedding,
  title = {Co-Embedding of Nodes and Edges with Graph Neural Networks},
  author = {Jiang, Xiaodong and Zhu, Ronghang and Ji, Pengsheng and Li, Sheng},
  journal = {IEEE Transactions on Pattern Analysis and Machine Intelligence},
  volume = {45},
  number = {6},
  pages = {7075--7086},
  year = {2023},
  doi = {10.1109/TPAMI.2020.3029762},
}

@article{cai2022linegraph,
  title = {Line Graph Neural Networks for Link Prediction},
  author = {Cai, Lei and Li, Jundong and Wang, Jie and Ji, Shuiwang},
  journal = {IEEE Transactions on Pattern Analysis and Machine Intelligence},
  volume = {44},
  number = {9},
  pages = {5103--5113},
  year = {2022},
  doi = {10.1109/TPAMI.2021.3080635}
}

@article{choudhary2021alignn,
  title = {Atomistic Line Graph Neural Network for Improved Materials Property Predictions},
  author = {Choudhary, Kamal and DeCost, Brian},
  journal = {npj Computational Materials},
  volume = {7},
  pages = {185},
  year = {2021},
  doi = {10.1038/s41524-021-00650-1}
}

@inproceedings{zopf2022wl,
  title = {{1-WL} Expressiveness Is (Almost) All You Need},
  author = {Zopf, Markus},
  booktitle = {2022 International Joint Conference on Neural Networks},
  year = {2022},
  eprint = {2202.10156},
  archivePrefix = {arXiv},
  primaryClass = {cs.LG}
}

@inproceedings{lachi2025link,
  title = {Bridging Theory and Practice in Link Representation with Graph Neural Networks},
  author = {Lachi, Veronica and Ferrini, Francesco and Longa, Antonio and Lepri, Bruno and Passerini, Andrea and Jaeger, Manfred},
  booktitle = {Advances in Neural Information Processing Systems},
  volume = {38},
  year = {2025},
  doi = {10.52202/085713-4090},
  url = {https://proceedings.neurips.cc/paper_files/paper/2025/hash/b1cd72276feff8173462e3f733ac66f8-Abstract-Conference.html}
}

@article{cai1992cfi,
  title = {An Optimal Lower Bound on the Number of Variables for Graph Identification},
  author = {Cai, Jin-Yi and F{\"u}rer, Martin and Immerman, Neil},
  journal = {Combinatorica},
  volume = {12},
  number = {4},
  pages = {389--410},
  year = {1992},
  doi = {10.1007/BF01305232}
}

@inproceedings{xu2019gin,
  title = {How Powerful are Graph Neural Networks?},
  author = {Xu, Keyulu and Hu, Weihua and Leskovec, Jure and Jegelka, Stefanie},
  booktitle = {International Conference on Learning Representations},
  year = {2019},
  url = {https://openreview.net/forum?id=ryGs6iA5Km}
}

@inproceedings{feng2023fwl,
  title = {Extending the Design Space of Graph Neural Networks by Rethinking Folklore {Weisfeiler-Lehman}},
  author = {Feng, Jiarui and Kong, Lecheng and Liu, Hao and Tao, Dacheng and Li, Fuhai and Zhang, Muhan and Chen, Yixin},
  booktitle = {Advances in Neural Information Processing Systems},
  pages = {9029--9064},
  publisher = {Curran Associates, Inc.},
  volume = {36},
  year = {2023},
  doi = {10.52202/075280-0397},
  url = {https://proceedings.neurips.cc/paper_files/paper/2023/hash/1cac8326ce3fbe79171db9754211530c-Abstract-Conference.html}
}

@article{whitney1932congruent,
  title = {Congruent Graphs and the Connectivity of Graphs},
  author = {Whitney, Hassler},
  journal = {American Journal of Mathematics},
  volume = {54},
  number = {1},
  pages = {150--168},
  year = {1932},
  doi = {10.2307/2371086},
  publisher = {Johns Hopkins University Press}
}

@article{ramana1994fractional,
  title = {Fractional Isomorphism of Graphs},
  author = {Ramana, Motakuri V. and Scheinerman, Edward R. and Ullman, Daniel},
  journal = {Discrete Mathematics},
  volume = {132},
  number = {1--3},
  pages = {247--265},
  year = {1994},
  doi = {10.1016/0012-365X(94)90241-0}
}

@book{grohe2017descriptive,
  title = {Descriptive Complexity, Canonisation, and Definable Graph Structure Theory},
  author = {Grohe, Martin},
  year = {2017},
  publisher = {Cambridge University Press},
  series = {Lecture Notes in Logic},
  volume = {47},
  doi = {10.1017/9781139028868}
}

@article{du2026lgan,
  title = {{LGAN}: An Efficient High-Order Graph Neural Network via the Line Graph Aggregation},
  author = {Du, Lin and Bai, Lu and Li, Jincheng and Cui, Lixin and Du, Hangyuan and Zhang, Lichi and Chen, Yuting and Li, Zhao},
  journal = {Proceedings of the AAAI Conference on Artificial Intelligence},
  volume = {40},
  number = {25},
  pages = {20914--20922},
  year = {2026},
  doi = {10.1609/aaai.v40i25.39232},
}

@article{fuhlbruck2021identifiability,
  author = {Fuhlbr{\"u}ck, Frank and K{\"o}bler, Johannes and Verbitsky, Oleg},
  title = {Identifiability of Graphs with Small Color Classes by the {Weisfeiler--Leman} Algorithm},
  journal = {SIAM Journal on Discrete Mathematics},
  volume = {35},
  number = {3},
  pages = {1792--1853},
  year = {2021},
  doi = {10.1137/20M1327550}
}

@article{ponomarenko2023wl,
  author = {Ponomarenko, Ilia},
  title = {On the {WL}-dimension of Circulant Graphs of Prime Power Order},
  journal = {Algebraic Combinatorics},
  volume = {6},
  number = {6},
  pages = {1469--1490},
  year = {2023},
  doi = {10.5802/alco.315},
}

@misc{chen2019coherent,
  author = {Chen, Gang and Ponomarenko, Ilia},
  title = {Lectures on Coherent Configurations},
  year = {2019},
  howpublished = {Online lecture notes},
  note = {Updated edition},
  url = {https://www.pdmi.ras.ru/~inp/ccNOTES.pdf},
}

@inproceedings{gilmer2017mpnn,
  title = {Neural Message Passing for Quantum Chemistry},
  author = {Gilmer, Justin and Schoenholz, Samuel S. and Riley, Patrick F. and Vinyals, Oriol and Dahl, George E.},
  booktitle = {Proceedings of the 34th International Conference on Machine Learning},
  series = {Proceedings of Machine Learning Research},
  volume = {70},
  pages = {1263--1272},
  year = {2017},
  publisher = {PMLR},
  url = {https://proceedings.mlr.press/v70/gilmer17a.html},
}

@inproceedings{hamilton2017graphsage,
  title = {Inductive Representation Learning on Large Graphs},
  author = {Hamilton, William L. and Ying, Zhitao and Leskovec, Jure},
  booktitle = {Advances in Neural Information Processing Systems},
  volume = {30},
  pages = {1024--1034},
  year = {2017},
  publisher = {Curran Associates, Inc.},
  url = {https://proceedings.neurips.cc/paper/2017/hash/5dd9db5e033da9c6fb5ba83c7a7ebea9-Abstract.html}
}

@inproceedings{khalil2017combinatorial,
  title = {Learning Combinatorial Optimization Algorithms over Graphs},
  author = {Khalil, Elias and Dai, Hanjun and Zhang, Yuyu and Dilkina, Bistra and Song, Le},
  booktitle = {Advances in Neural Information Processing Systems},
  volume = {30},
  pages = {6348--6358},
  year = {2017},
  publisher = {Curran Associates, Inc.},
  url = {https://proceedings.neurips.cc/paper/2017/hash/d9896106ca98d3d05b8cbdf4fd8b13a1-Abstract.html}
}

@inproceedings{allamanis2018programs,
  title = {Learning to Represent Programs with Graphs},
  author = {Allamanis, Miltiadis and Brockschmidt, Marc and Khademi, Mahmoud},
  booktitle = {International Conference on Learning Representations},
  year = {2018},
  url = {https://openreview.net/forum?id=BJOFETxR-}
}

\clearpage

\appendix

\section{Proof of the \texorpdfstring{\(3\)}{3}-WL Backward-Transfer Theorem}
\label{app:backward-transfer-proof}
This appendix first proves the connected Whitney-general core,
Theorem~\ref{thm:backward-transfer}, using joint pair refinement, and then
derives the componentwise extension in
Theorem~\ref{thm:componentwise-backward}.  Throughout the core proof,
\(G,H\) denote the two root
graphs and \(L(G),L(H)\) their line graphs; we do not rename either line
graph. ``Pair refinement'' means the joint ordered-pair update below, i.e.,
2-FWL. By the standard arity-shift theorem, its stable graph-distinguishing
power equals that of the paper's separate-coordinate \(3\)-WL
\citep{feng2023fwl}. The symbol \(\circ\) denotes the
Hadamard (entrywise) product.

\paragraph{Proof outline.}
The proof proceeds in five steps.  First, we characterize stable pair
refinement by its coherent matrix algebra.  Second, we determine root endpoint
degrees and star and dart profiles from the stable pair algebra of \(L(G)\).
Third, we use these profiles to construct root-side algebras and an isomorphism
\(\Omega\) induced by the line-algebra isomorphism.  Fourth, we construct the root
profile cells at distances zero and one, treating unequal and equal profile
fibers separately.  Fifth, we construct the remaining cells by induction on
root distance.  The resulting algebras are Hadamard closed, and \(\Omega\) then
restricts to an isomorphism of the stable root pair algebras.

\subsection{Coherent-algebra characterization of pair refinement}

We first record the standard characterization of joint pair refinement by
coherent configurations.  Proposition~2.3 of
\citet{fuhlbruck2021identifiability} identifies the stable joint-pair
partition with the coherent closure of the graph relations, and their
Lemma~2.4 proves that an algebraic isomorphism preserves every refinement
round; see also
\citet[Sections~2.2, 2.5, and 3.1]{ponomarenko2023wl}. Those references call
the joint update on ordered pairs \(2\)-WL. In the graph-learning terminology
used here it is 2-FWL; its equivalence in stable graph-distinguishing power
to separate-coordinate \(3\)-WL is the separate arity-shift result
\citep{feng2023fwl}. We include the short derivation
needed here so that the matrix identities and their index order are explicit.

For a graph \(G\), pair refinement assigns colors to ordered pairs by the
joint-replacement update
\[
 c_{t+1}^G(x,y)
 =\hash\left(
   c_t^G(x,y),
   \multiset{\bigl(c_t^G(z,y),c_t^G(x,z)\bigr):z\in V(G)}
 \right).
\]
The order of the two entry components agrees with the global-replacement
formula in Section~\ref{sec:global-edge}. Some references list the two entry
coordinates in the reverse order; the fixed coordinate swap is a bijection on
entries and does not change the induced refinement partition. The
initial color \(c_0^G(x,y)\)
records whether \(x=y\), whether \(x\) and
\(y\) are adjacent, or whether they are distinct and nonadjacent.  Write
\(c_\infty^G\) for the \emph{stable} coloring, attained at the first round
after which one further update no longer refines the induced partition of
ordered pairs.  Call two graphs \emph{pair-refinement equivalent} when the
stable graph-level signatures of their joint pair refinements, computed
with the shared injective encoding of Section~\ref{sec:background}, are
equal. By the standard arity-shift theorem, this is equivalent in stable
graph-distinguishing power to
\(G\equiv_{3\text{-WL}}H\)
\citep{feng2023fwl}.
For a stable color \(\rho\), define its \(0\)-\(1\) relation matrix by
\[
 (R_\rho^G)_{xy}=\one[c_\infty^G(x,y)=\rho].
\]
We use the same symbol \(R_\rho^G\) for the matrix and for its support
relation \(\{(x,y):c_\infty^G(x,y)=\rho\}\), writing
\((x,y)\in R_\rho^G\) accordingly.  At stabilization, one more
refinement round does not split a color class.  Hence, if both \((x,y)\) and
\((x',y')\) have stable color \(\tau\), their stable multisets of replacement
entries are equal.  The multiplicity of the entry \((\sigma,\rho)\) is
therefore independent of the chosen pair of color \(\tau\).  This defines the
\emph{intersection number}
\[
 p_{\rho\sigma}^{\tau}
 =\left|\{z:(x,z)\in R_\rho^G,\ (z,y)\in R_\sigma^G\}\right|
 \qquad ((x,y)\in R_\tau^G).
\]
This representative-independence is the standard intersection-number axiom
for a coherent configuration
\citep[Section~2.2, condition~(C3)]{ponomarenko2023wl}; the argument above
shows directly why stable pair refinement satisfies it.
The reversed order \((\sigma,\rho)\) is exactly the entry order in the update:
\((z,y)\) has color \(\sigma\) and \((x,z)\) has color \(\rho\).  Entrywise,
for \((x,y)\in R_\tau^G\),
\[
 \begin{aligned}
 (R_\rho^G R_\sigma^G)_{xy}
 &=\sum_{z\in V(G)}(R_\rho^G)_{xz}(R_\sigma^G)_{zy}\\
 &=\left|\{z:(x,z)\in R_\rho^G,\ (z,y)\in R_\sigma^G\}\right|
 =p_{\rho\sigma}^{\tau}.
 \end{aligned}
\]
Because the stable relations partition \(V(G)^2\), this entrywise identity is
equivalent to
\[
 R_\rho^G R_\sigma^G
 =\sum_\tau p_{\rho\sigma}^{\tau}R_\tau^G.
\]

The matrices \(R_\rho^G\) have disjoint supports and sum to \(J_{V(G)}\).
The equality and adjacency relations are unions of stable classes because
they are distinguished initially, so their real span \(\mathcal A_G\)
contains \(I_{V(G)}\) and \(A_G\), as well as \(J_{V(G)}\).  The displayed
product formula proves closure under ordinary multiplication.  Coordinate
swap permutes the stable classes: this follows by induction from the
symmetric initial types.  If \(a^\top,b^\top\) denote the reversed colors
at one round, then an entry \((a,b)\) becomes
\((b^\top,a^\top)\) after reversing the outer pair.  Thus
\((R_\rho^G)^\top\) is another basis relation; write \(\rho^\top\) for its
color, so \((R_\rho^G)^\top=R_{\rho^\top}^G\), and write \(|R_\rho^G|\) for
the number of ones in \(R_\rho^G\).  Disjointness gives
\[
 R_\rho^G\circ R_\sigma^G
 =\one[\rho=\sigma]R_\rho^G.
\]
Thus \(\mathcal A_G\) is closed under transpose and Hadamard multiplication
as well.  Call a real linear space of square matrices that contains the
identity and all-ones matrices and is closed under transpose, ordinary
multiplication, and Hadamard multiplication a \emph{coherent matrix
algebra}; \(\mathcal A_G\) is the \emph{stable pair coherent algebra} of
\(G\).

\begin{lemma}[Stable-pair algebra characterization]
\label{lem:pair-algebra-bridge}
Graphs \(G,H\) are pair-refinement equivalent if and only if there is a bijection
\(\pi\) between their stable pair colors such that the linear map
\[
 \varphi(R_\rho^G)=R_{\pi(\rho)}^H
\]
maps \(I_{V(G)},J_{V(G)},A_G\) to
\(I_{V(H)},J_{V(H)},A_H\) and preserves transpose, ordinary multiplication,
Hadamard multiplication, and matrix trace.
\end{lemma}

\begin{proof}
Run pair refinement on \(G\) and \(H\) in parallel with the shared injective
color encoding fixed in Section~\ref{sec:background}.  If the stable
signatures agree, matching
color names gives a bijection \(\pi\) and preserves every relation
multiplicity \(|R_\rho^G|=|R_{\pi(\rho)}^H|\).  A stable color \(\tau\)
records its initial equality/adjacency type and the multiplicity of every
replacement entry.  The observation above therefore shows that \(\pi\)
preserves every intersection number \(p_{\rho\sigma}^{\tau}\), with the
same statement for the transposed colors.  In particular,
\(\pi(\rho^\top)=\pi(\rho)^\top\).  Consequently the displayed linear map
preserves ordinary multiplication and transpose.  It preserves Hadamard
multiplication by the disjoint-support identity above.
Initial types give preservation of \(I,A\), summing all basis matrices gives
preservation of \(J\), and matching relation sizes gives trace preservation:
\(\operatorname{tr}(R_\rho^G)=|R_\rho^G|\) for an equality-type color and is
zero otherwise.  Hence trace is preserved on the whole span.

Conversely, suppose such a basis bijection and linear map exist.  Preservation
of \(I,A\), transpose, and multiplication preserves the initial types and
all intersection numbers.  More explicitly, uniqueness of expansion in the
disjoint basis shows from
\(\varphi(I_{V(G)})=I_{V(H)}\) and \(\varphi(A_G)=A_H\) that
\(\pi\) preserves equality and adjacency type; uniqueness of the expansion
of \(R_\rho R_\sigma\) then shows that it preserves every
\(p_{\rho\sigma}^{\tau}\).  To see explicitly why this recovers the
refinement, observe inductively that each round-\(t\) color is a union of
stable relations.  For a stable relation \(R_\tau^G\), the multiplicity of a
round-\(t\) entry \((b,a)\) is the sum of \(p_{\rho\sigma}^{\tau}\) over the
stable relations \(R_\rho^G\) contained in \(a\) and \(R_\sigma^G\) contained
in \(b\).  Thus the transported unions have identical entry multisets at the
next round.  This is the same induction as in
\citet[Lemma~2.4]{fuhlbruck2021identifiability}, with the two entry
components reversed to match our convention.

It remains to match graph-level color multiplicities.  Trace preservation
applied to
\(R_\rho^G(R_\rho^G)^\top\) gives
\[
 |R_\rho^G|
 =\operatorname{tr}\!\left(R_\rho^G(R_\rho^G)^\top\right)
 =|R_{\pi(\rho)}^H|.
\]
Thus the matched stable colors have the same refinement data and the same
multiplicities, so their stable graph-level signatures agree.  This is the
standard coherent-configuration characterization quoted above.
\end{proof}

Equivalently, \(\mathcal A_G\) is the smallest coherent matrix algebra
containing \(I_{V(G)},J_{V(G)},A_G\)
\citep[Propositions~2.1 and~2.3]{fuhlbruck2021identifiability}.  We also use that a
joint level set of finitely many matrices in \(\mathcal A_G\) is a union of
basis relations.
Its indicator belongs to \(\mathcal A_G\) by finite Hadamard interpolation.
Explicitly, for a transported pair of matrices, take the union of the finite
sets of entries attained on the two sides.  For a requested value \(a\) in
that union, choose the Lagrange polynomial that is one at \(a\) and zero at
every other value in the union, and evaluate it entrywise using Hadamard
products.  If \(a\) is not attained on one side, the resulting equality
indicator on that side is the zero matrix.  A requested value outside the
union is assigned the zero indicator on both sides.  The Hadamard product of
these single-matrix indicators is the desired joint-level-set indicator.
Because the same polynomial is used on both sides, the entire construction is
transported.  This convention applies to every entrywise selector below.

The \emph{paw} is a triangle with one pendant edge, and the \emph{diamond}
is \(K_4\) with one edge removed.  These names are used only for the two
small local configurations isolated below.

\begin{lemma}[Exceptional roots handled directly]
\label{lem:bt-direct-roots}
Assume \(G,H\) are Whitney-general and
\(L(G)\equiv_{3\text{-WL}}L(H)\).  If either root is \(K_1,K_2\), the paw,
or the diamond, then \(G\cong H\).
\end{lemma}

\begin{proof}
On a diagonal pair \((e,e)\), the multiplicity of adjacency-colored
replacement entries records \(d_{L(G)}(e)\); hence pair refinement preserves
line-graph order and degree multiplicities.  Equivalently, the line degree
on a diagonal fiber is the sum of the out-valencies of its outgoing adjacent
basis relations, while the fiber size is its trace; all these quantities are
transported.  A connected root with an empty
line graph is \(K_1\), and a connected root with
a one-vertex line graph is \(K_2\).  The line graph of the paw has degree
multiset \((3,3,2,2)\): the two degree-three vertices are universal, so this
uniquely determines \(K_4\) minus one edge.
The line graph of the diamond has degree multiset \((4,3,3,3,3)\): the
degree-four vertex is universal and the other four vertices induce \(C_4\),
so the line graph is uniquely \(K_1\vee C_4\), where \(\vee\) denotes graph
join.  In either case \(L(G)\) and
\(L(H)\) are isomorphic.  Neither line graph is \(K_3\), so Whitney's theorem
\citep{whitney1932congruent} gives \(G\cong H\).
\end{proof}

\subsection{Stable line data and root profiles}

For the remainder, let \(G,H\) be Whitney-general roots satisfying
\(L(G)\equiv_{3\text{-WL}}L(H)\), and assume by
Lemma~\ref{lem:bt-direct-roots} that both roots have order at least three and
that neither is the paw or the diamond.

These cases are separated from the general construction for the following
reasons.  For roots of orders one and two, the line algebra contains no
adjacent basis relation.  The paw and diamond are the additional
Whitney-general roots for which, respectively, the degree collision (A.3) of
Lemma~\ref{lem:bt-line-degree} and the wedge ambiguity (A.4) of
Lemma~\ref{lem:bt-wedge-type} occur.  Lemma~\ref{lem:bt-direct-roots} treats
all four cases directly and proves \(G\cong H\).  Thus only the construction
below excludes them.

Let \(\mathcal A_{L(G)}\) have basis
relations \(R_\rho^{L(G)}\).  Write \(D_\alpha^{L(G)}\) for those basis
relations supported on the coordinate diagonal; each is a diagonal
\(0\)-\(1\) projector.  Define the
associated edge fiber
\[
 \mathcal E_\alpha^G
 =\{e\in E(G):(D_\alpha^{L(G)})_{ee}=1\}.
\]
Define \(R_\rho^{L(H)},D_\alpha^{L(H)},\mathcal E_\alpha^H\) analogously.
We call \(R_\rho^{L(G)}\) \emph{adjacent} when its support is contained in
\(A_{L(G)}\), and \emph{disjoint} when its support consists of ordered pairs
of disjoint root edges.  If
\(R_\rho^{L(G)}
 =D_\alpha^{L(G)}R_\rho^{L(G)}D_\beta^{L(G)}\), write
\(\mathrm{src}(\rho)=\alpha\) and \(\mathrm{tgt}(\rho)=\beta\); its
out-valency is
\[
 \nu_\rho
 =\left|\{f\in\mathcal E_\beta^G:(e,f)\in R_\rho^{L(G)}\}\right|,
 \qquad e\in\mathcal E_\alpha^G.
\]
Every basis relation lies in a unique block
\(\mathcal E_\alpha^G\times\mathcal E_\beta^G\), so its source and target
fibers are well defined.  The intersection-number identities make
\(\nu_\rho\) independent of \(e\), and it is positive for a nonempty basis
relation.  Summing over the adjacent relations from \(\mathcal E_\alpha\)
to \(\mathcal E_\delta\) gives the \emph{fiber valency}
\(q_{\alpha\delta}\): every \(e\in\mathcal E_\alpha\) has exactly
\(q_{\alpha\delta}\) line neighbors in \(\mathcal E_\delta\).  We refer to
this constancy as \emph{equitability} of the diagonal fibers; the numbers
\(q_{\alpha\delta}\) are sums of intersection numbers.
Lemma~\ref{lem:pair-algebra-bridge} supplies
\[
 \varphi:\mathcal A_{L(G)}\longrightarrow\mathcal A_{L(H)}
\]
matching all stated algebra operations and trace, with
\(\varphi(R_\rho^{L(G)})=R_\rho^{L(H)}\) and
\(\varphi(D_\alpha^{L(G)})=D_\alpha^{L(H)}\) after synchronizing the labels.  From
now on, unadorned \(R_\rho,D_\alpha,\mathcal E_\alpha\) denote the
\(L(G)\)-relation, \(L(G)\)-projector, and \(G\)-edge fiber, respectively.
Their transported mates retain the explicit notation
\(R_\rho^{L(H)},D_\alpha^{L(H)},\mathcal E_\alpha^H\).  We call any such
matched objects \emph{transported}.  Coordinate unit vectors
are denoted by \(\mathbf e_\alpha\) or \(\mathbf e_\rho\), according to the
index set.  For a matrix \(M\), \(\one[M=a]\) denotes the entrywise
\(0\)-\(1\) matrix with \((e,f)\)-entry \(\one[M_{ef}=a]\); entrywise
inequalities such as \(\one[M>0]\) are interpreted analogously.  Finally,
every automorphism of a root induces an automorphism of its line graph and,
by induction on rounds, preserves each refinement color; stable colors,
fibers, and all profiles defined below are therefore automorphism
invariant.

For a root vertex \(u\), define its \emph{star profile}
\[
 \mathbf s(u)=(s_\alpha(u))_\alpha,
 \qquad
s_\alpha(u)=\left|\{e\in\mathcal E_\alpha:u\in e\}\right|.
\]
Thus \(\mathbf s(u)\) counts incident root edges by their stable diagonal
color in \(L(G)\).  A \emph{dart} is an incident pair \((u,e)\) with \(u\in e\).  For
\(e\in\mathcal E_\alpha\), define its
\emph{dart profile}
\[
 \mathbf d(u,e)_\rho
 =\left|\{f\ne e:u\in f,\ (e,f)\in R_\rho\}\right|,
\]
indexed by all basis relations; it is supported on the adjacent relations
leaving \(\mathcal E_\alpha\), and only those coordinates appear below.  It refines
the \(u\)-wing of \(e\), namely the other edges incident with \(u\), by the
stable relation from \(e\) to each such edge.  Summing coordinates by target
fiber recovers the star profile:
\[
 s_\beta(u)
 =\one[\alpha=\beta]
  +\sum_{\substack{\rho\text{ adjacent}:\,
                    \mathrm{src}(\rho)=\alpha,\;
                    \mathrm{tgt}(\rho)=\beta}}
        \mathbf d(u,e)_\rho.
 \tag{A.1}
\]
The indicator accounts for \(e\) itself, which is included in
\(\mathbf s(u)\) but excluded from \(\mathbf d(u,e)\).

For adjacent root edges \(e=uv,f=uw\), call \(u\) the center and \(v,w\)
the outer vertices.  The wedge is
\emph{closed} if \(vw\) is an edge, and \emph{open} otherwise; in the closed
case \(vw\) is its closing edge.

For a root quantity \(q\), call a transported relation label \emph{mixed} if
\(q\) is nonconstant on its support in either root or has different constant
values on the two transported supports.  Thus proving that a label is not
mixed proves both within-root constancy and cross-root transport.

For a root edge \(e=uv\), put
\[
 \delta_u=d_G(u)-1,\qquad \delta_v=d_G(v)-1,\qquad
 t_e=|N_G(u)\cap N_G(v)|;
\]
in comparisons between two representatives, \(d\) and \(N\) refer to the
root containing the displayed representative.  For \(f\in N_{L(G)}(e)\),
let \(\lambda_e(f)=|N_{L(G)}(e)\cap N_{L(G)}(f)|\), and define
\[
 \operatorname{LD}_G(e)
 =\multiset{\lambda_e(f):f\in N_{L(G)}(e)}.
\]
For integers \(0\le\delta_u\le\delta_v\) and \(0\le t\le\delta_u\), define
the multiset
\[
 F(\delta_u,\delta_v,t)
 =[\delta_u]^{t}\uplus[\delta_u-1]^{\delta_u-t}
  \uplus[\delta_v]^{t}\uplus[\delta_v-1]^{\delta_v-t},
 \tag{A.2}
\]
where \(\uplus\) is multiset union and \([a]^b\) denotes \(b\) copies of
\(a\).

\begin{lemma}[Line-degree multisets and their collisions]
\label{lem:bt-line-degree}
The diagonal pair color of a root edge \(e\) determines
\(\operatorname{LD}_G(e)\), and the multiset is transported.  Moreover,
\(\operatorname{LD}_G(e)=F(\delta_u,\delta_v,t_e)\) with
\(\delta_u\le\delta_v\), and \(F\) determines its three arguments except
for the single collision
\[
 F(0,2,0)=F(1,1,1)=[1,1].
 \tag{A.3}
\]
\end{lemma}

\begin{proof}
The diagonal pair color of \(e\) determines the multiset.  Indeed,
\(A_{L(G)}^2\) is in the coherent algebra, so it has a constant value
\(\ell_\rho\) on every adjacent basis relation \(R_\rho\).  If
\(e\in\mathcal E_\alpha\), then
\[
 \operatorname{LD}_G(e)
 =\mathop{\biguplus}_{\substack{\rho\ {\rm adjacent}\\
                    \mathrm{src}(\rho)=\alpha}}
   [\ell_\rho]^{\nu_\rho}.
\]
Both \(\ell_\rho\) and the out-valency \(\nu_\rho\) are intersection-number
data, so this multiset is also transported.  Listing the two endpoint wings
gives \(\operatorname{LD}_G(e)=F(\delta_u,\delta_v,t_e)\): each wing at an
endpoint of degree \(d\) contributes one half of (A.2), with \(t_e\) of its
edges closed by the opposite endpoint.
For the uniqueness claim, \(|F|=\delta_u+\delta_v\); for nonempty \(F\),
\(\max F=\delta_v-1\) when \(t_e=0\) and \(\delta_v\) when \(t_e>0\).
If two candidates have the
same \(t_e=0\) status, these data recover \(\delta_u,\delta_v\); for
\(t_e>0\), the identity
\[
 \sum F=\delta_u(\delta_u-1)+\delta_v(\delta_v-1)+2t_e
\]
recovers \(t_e\).  A collision across the two statuses must compare a
positive-\(t\) candidate \((a,b,t)\) with \((a-1,b+1,0)\): equality of
cardinalities and maxima gives these two shifts.  The zero-\(t\) multiset
contains \(b+1\) copies of its maximum \(b\), which forces \(a=b\);
otherwise the positive candidate has at most \(a<b\) such copies.  With
\(a=b\), the zero-\(t\) candidate also contains \(a-1\) copies of
\(a-2\), absent from the positive candidate.  Hence \(a=1\), and equality
of the remaining multiplicities forces \(t=1\).  This is exactly (A.3).
The empty case is immediate.
\end{proof}

\begin{lemma}[Endpoint-degree recovery]
\label{lem:bt-endpoint-degrees}
Every diagonal fiber of \(\mathcal A_{L(G)}\) fixes its unordered endpoint
degrees, and these labels are transported by \(\varphi\).
\end{lemma}

\begin{proof}
By Lemma~\ref{lem:bt-line-degree}, only the collision (A.3) could make a
transported diagonal fiber mixed in its unordered endpoint degrees.
In root language, (A.3) confuses a pendant \((1,3)\)-edge with a triangular
\((2,2)\)-edge.  Suppose two representatives of a transported diagonal
fiber \(\mathcal E_\alpha\), either in one root or one in each root, realize
the two cases.  Write the pendant edge as \(e=pc\), with
the other edges at the degree-three vertex \(c\) equal to
\(f_1=ca,f_2=cb\).
Write the triangular edge as \(e'=uv\) in \(uvw\), with
\(f'_1=uw,f'_2=vw\).  The degree-two vertices \(u,v\) are adjacent true
twins, so swapping them is a root automorphism; it shows that
\(f'_1,f'_2\) have one diagonal color \(\beta\).  The \(\alpha\)-edge \(e'\)
has exactly these two line neighbors.  Equitability of the diagonal
\(\alpha\)-fiber therefore makes the two line neighbors \(f_1,f_2\) of \(e\)
both \(\beta\)-colored.

If \(d(w)=2\), then \(f'_1\) is again triangular of type \((2,2)\),
whereas \(f_1\), having a degree-three endpoint, must be pendant by (A.3).
Thus \(d(a)=d(b)=1\), and connectedness gives \(K_{1,3}\) on the
pendant side and \(K_3\) on the triangular side.  If \(d(w)>2\),
uniqueness of \(F\) for the matched \(\beta\)-edges transports both their
endpoint degrees and their value \(t_e=1\).  Explicitly,
\(\{2,d(w)\}=\{3,d(a)\}\), so
\(d(w)=3,d(a)=2\), and the common-neighbor count gives one common neighbor
of \(c,a\), necessarily \(b\).  Thus \(ab\) is an edge, symmetrically
\(d(b)=2\), and \(p,c,a,b\) form a paw and exhaust their degrees.  For
\(g=ab\), the line vertices
\(e,g\) have the same open neighborhood, namely \(\{f_1,f_2\}\).  Swapping
these line-graph twins is an automorphism, so \(g\) also has color
\(\alpha\).  Thus \(f_1\) has two \(\alpha\)-neighbors.  Equitability at the
matched \(\beta\)-vertices \(f_1,f'_1\) forces \(f'_1\) to have two as well.
Because the \(\beta\)-fiber has endpoint type \(\{2,3\}\), it differs from
\(\alpha\); hence \(f'_2\) is not the second one, and the third edge \(h=wr\)
at \(w\) has color \(\alpha\).  Since \(d(w)=3\), (A.3) makes \(h\) pendant,
and the
other root is also a paw.  In the within-root comparison, the same exhausted
configurations are components.  Thus the only propagated collisions are the
Whitney pair and the paw, and both are excluded by the standing
hypotheses.
\end{proof}

\begin{lemma}[Central degree and wedge type]
\label{lem:bt-wedge-type}
Every adjacent basis relation fixes its central degree and whether its
wedges are open or closed, and these labels are transported by \(\varphi\).
\end{lemma}

\begin{proof}
Take an adjacent pair
\(e=uv,f=uw\).  Its coherent common-neighbor number is
\[
 (A_{L(G)}^2)_{ef}=d(u)-2+\one[vw\in E(G)].
 \tag{A.4}
\]
If a relation were mixed in its central degree, (A.4) would force a closed
wedge of degree \(c\) and an open wedge of degree \(c+1\).
Lemma~\ref{lem:bt-endpoint-degrees} fixes the unordered endpoint degrees of
each source and target fiber.  For either leg, equality of those unordered
pairs reads
\(\{c,d(\text{closed outer})\}
 =\{c+1,d(\text{open outer})\}\).  Thus the closed outer endpoint has degree
\(c+1\) and the open outer endpoint degree \(c\); all four wedge legs have
endpoint type \(\{c,c+1\}\).

Let \(\Pi_{\mathrm{hh}}\) be the sum of the diagonal projectors of endpoint type
\(\{c+1,c+1\}\), and set
\[
 M=A_{L(G)}\Pi_{\mathrm{hh}},\qquad Q=MM^\top.
\]
Here \(Q_{ef}\) counts the high--high edges adjacent in \(L(G)\) to both
\(e\) and \(f\).  In the closed realization the only such common neighbor is
the closing high--high edge, so \(Q_{ef}=1\).  In the open realization the
center is the common high endpoint of \(e,f\), so every high--high edge
adjacent to \(e\) is also adjacent to \(f\), and \(Q_{ef}=Q_{ee}\).
Constancy on the common basis relation first equates this open value to the
closed value one; constancy and transport on the source and target diagonal
fibers then equate the corresponding diagonal values.  Since \(Q_{ee}\)
counts high--high edges at the high endpoint of \(e\), every involved
degree-\((c+1)\) endpoint consequently has exactly one high-degree neighbor.
Let \(U=\one[Q=1]\) and put
\(T=A_{L(G)}\circ U\).  The two realizations now have the following data:
\[
\begin{array}{c|c|c|c}
 &\text{center degree}&\text{outer degrees}&(T^2)_{ef}\\ \hline
 \text{open}&c+1&c,c&c-2\\
 \text{closed}&c&c+1,c+1&0
\end{array}
\]
Indeed, in the open realization the unique high--high edge at the center is
not a \(T\)-neighbor of itself, whereas each of the other \(c-2\) center
edges is a common \(T\)-neighbor of \(e,f\).  In the closed realization the
center is low, and each high outer endpoint has only the closing high--high
edge; the closing edge itself is not a \(T\)-neighbor, so there is no common
\(T\)-neighbor.
Constancy of \(T^2\) on the common relation gives \(c-2=0\), hence \(c=2\).
A source edge in the closed realization then has
two \(T\)-neighbors: the other wedge leg and the remaining edge at its
degree-three outer endpoint.  In the open realization it has the other
wedge leg, plus a second \(T\)-neighbor exactly when its degree-two outer
endpoint is adjacent to the center's unique high-degree neighbor.  Since
\((T^2)_{ee}\) is the \(T\)-degree of \(e\), equality of the corresponding
source diagonal entries forces this edge, and the target calculation
forces its analogue at the other outer vertex.  The four vertices form a
diamond and exhaust their degrees, so connectedness makes the root the
diamond, contrary to the hypothesis.
\end{proof}

For an edge fiber \(\mathcal E_\alpha\), collect the fiber valencies into
\[
 \mathbf q_\alpha=(q_{\alpha\beta})_\beta+2\mathbf e_\alpha.
\]
The two endpoint wings of every \(\alpha\)-edge \(uv\), together with the
edge itself once at each endpoint, give
\[
 \mathbf s(u)+\mathbf s(v)=\mathbf q_\alpha.
 \tag{A.5}
\]
For an adjacent basis relation \(R_\gamma\) with source \(\mathcal E_\alpha\),
target \(\mathcal E_\beta\), and representative \(e=uv\), \(f=uw\), define
the vectors
\[
 h^\star_\gamma(\delta)=(A_{L(G)}D_\delta A_{L(G)})_{ef},
 \qquad
 h^{\mathrm d}_\gamma(\rho)=(R_\rho A_{L(G)})_{ef},
\]
indexed by the diagonal fibers \(\delta\) and by the adjacent basis
relations \(\rho\), respectively.  The first counts the common line
neighbors of \(e,f\) lying in \(\mathcal E_\delta\); the second counts the
common line neighbors \(g\) with \((e,g)\in R_\rho\).  Both are sums of
intersection numbers, hence constant on \(R_\gamma\) and transported.

\paragraph{Profile-difference observation.}
We use the following observation twice.  For a closed adjacent relation
\(R_\gamma\), suppose each representative has a profile of the form
\(v_\gamma-\mathbf e_{\ell}\), where
the vector \(v_\gamma\) is transported and depends only on \(\gamma\), and
the label \(\ell\) is carried by the closing edge of the representative.
If such a relation is mixed between profiles \(\mathbf x\ne\mathbf y\)
whose representatives carry closing labels \(\ell_{\mathbf x}\) and
\(\ell_{\mathbf y}\), subtraction gives
\[
 \mathbf y-\mathbf x
 =\mathbf e_{\ell_{\mathbf x}}-\mathbf e_{\ell_{\mathbf y}},
\]
so \(\ell_{\mathbf x}\ne\ell_{\mathbf y}\), and the two profiles differ, by
one, in exactly the two coordinates \(\ell_{\mathbf x},\ell_{\mathbf y}\).
If moreover the fixed unordered endpoint-profile pair of the source fiber
is transported and contains the central profile of every representative,
then that pair equals \(\{\mathbf x,\mathbf y\}\) in each root, and
\(\mathbf x,\mathbf y\) are the only profiles realized on \(R_\gamma\).

\begin{lemma}[Endpoint star profiles]
\label{lem:bt-star-endpoint}
Every diagonal fiber of \(\mathcal A_{L(G)}\) fixes its unordered endpoint
star profiles, and every open adjacent basis relation fixes its central
star profile.  These labels are transported by \(\varphi\).
\end{lemma}

\begin{proof}
Let \(R_\gamma\) be an open adjacent relation with source fiber
\(\mathcal E_\alpha\) and target fiber \(\mathcal E_\beta\).  The common
line neighbors of an open wedge \(e=uv,f=uw\) are exactly the other edges
at \(u\), so
\[
 \mathbf s(u)=\mathbf e_\alpha+\mathbf e_\beta+h^\star_\gamma.
 \tag{A.6}
\]
Every nonempty relation has positive constant out-valency on its source
fiber.  Thus, if an open relation leaves \(\mathcal E_\alpha\), every
\(\alpha\)-edge is the source of a pair in that relation; (A.6) supplies one
endpoint profile and (A.5) supplies the other.  If
none leaves \(\mathcal E_\alpha\), then for each \(\alpha\)-edge \(uv\), every
neighbor of either endpoint other than the other endpoint is adjacent to
both.  Here write \(N[x]=N(x)\cup\{x\}\) for the closed neighborhood.
Swapping the adjacent true twins \(u,v\), whose closed neighborhoods
coincide, is a root automorphism; its induced
line-graph automorphism fixes \(e\) and interchanges its two wings.  Stable
colors are automorphism invariant, so the profiles agree coordinatewise and
equal \(\mathbf q_\alpha/2\).  This proves the endpoint-profile
assertion, while (A.6) proves the open central-profile assertion.
\end{proof}

\begin{lemma}[Central star profiles and class sizes]
\label{lem:bt-star-profiles}
Every adjacent basis relation fixes its central star profile, and every
star-profile class size is determined and transported.
\end{lemma}

\begin{proof}
By Lemma~\ref{lem:bt-star-endpoint}, only a closed relation
\(R_\gamma\), with source \(\mathcal E_\alpha\) and target
\(\mathcal E_\beta\), could be mixed between central profiles
\(\mathbf x\ne\mathbf y\).  If its \(\mathbf x\)-centered representative
has closing-edge fiber \(\chi\), direct listing gives
\[
 h^\star_\gamma
 =\mathbf x-\mathbf e_\alpha-\mathbf e_\beta+\mathbf e_\chi,
\]
so the central profile is
\((\mathbf e_\alpha+\mathbf e_\beta+h^\star_\gamma)-\mathbf e_\chi\), of
the form required by the profile-difference observation with the closing-edge fiber
as label.  With \(\xi\) the closing fiber of a \(\mathbf y\)-centered
representative, the principle gives
\[
 \mathbf y-\mathbf x=\mathbf e_\chi-\mathbf e_\xi,
 \tag{A.7}
\]
and, since Lemma~\ref{lem:bt-star-endpoint} fixes and transports the
unordered endpoint pair of the source and target fibers, that pair equals
\(\{\mathbf x,\mathbf y\}\) and no other central profile occurs on
\(R_\gamma\).
For \(K\in\{G,H\}\), let
\[
 C_{\mathbf x}^K=\{u\in V(K):\mathbf s_K(u)=\mathbf x\},
 \qquad
 C_{\mathbf y}^K=\{u\in V(K):\mathbf s_K(u)=\mathbf y\},
\]
where \(\mathbf s_K\) denotes the star profile in \(K\).  We suppress the
root superscript inside a one-root calculation.  Since both central labels
occur in the source and target fibers, Lemma~\ref{lem:bt-star-endpoint}
makes every source or target edge of \(\gamma\) join \(C_{\mathbf x}\) to
\(C_{\mathbf y}\).  In the \(\mathbf x\)-centered realization the
\(\chi\)-edge
joins the two outer \(C_{\mathbf y}\)-vertices, so the fixed endpoint pair of
fiber \(\chi\) is \(\{\mathbf y,\mathbf y\}\).  Since
\(\mathbf x\ne\mathbf y\), (A.7) gives \(\xi\ne\chi\); moreover
\(x_\chi=0\), and its \(\chi\)-coordinate gives \(y_\chi=1\).
Consequently \((x_\chi,y_\chi)=(0,1)\).  Symmetrically the \(\xi\)-fiber
joins \(C_{\mathbf x}\) to itself and
\((x_\xi,y_\xi)=(1,0)\).  Hence the \(\xi\)- and
\(\chi\)-edges form perfect matchings on \(C_{\mathbf x}\) and
\(C_{\mathbf y}\), respectively.

Take a \(\mathbf x\)-centered representative
\(e=uv\in\mathcal E_\alpha\),
\(f=uw\in\mathcal E_\beta\), closed by
\(vw\in\mathcal E_\chi\), and let
\(u\bar u\in\mathcal E_\xi\) be the matching edge at \(u\).  For an
\(\mathbf x\)-centered \(\gamma\)-pair, this matching edge is its unique
common \(\xi\)-edge; for a \(\mathbf y\)-centered pair, the closing edge is
the unique common \(\xi\)-edge.  In the latter realization,
let \(R_\eta,R_\theta\) be the two relations through that witness.  Their
middle fiber is the \(\xi\)-fiber, so every
\((\eta,\theta)\)-intermediate is a common \(\xi\)-edge.  Therefore
\(p_{\eta\theta}^{\gamma}=1\), while every other decomposition through the
\(\xi\)-fiber has multiplicity zero.
Transport makes \(u\bar u\) the unique witness with
\((e,u\bar u)\in R_\eta\) and
\((u\bar u,f)\in R_\theta\).  In the \(\mathbf y\)-centered realization the
witness is the closing edge, so both leg relations are closed.  Wedge-type
transport (Lemma~\ref{lem:bt-wedge-type}) now supplies
\(\bar u v,\bar u w\), producing a \(K_4\) with
opposite matching-edge fibers \(\xi,\chi\).

For every edge fiber \(\delta\notin\{\xi,\chi\}\), (A.7) gives
\(x_\delta=y_\delta\).  If this common coordinate is positive, the fixed
endpoint pair of \(\delta\) contains both profiles and therefore is
\(\{\mathbf x,\mathbf y\}\); if it is zero, no \(\delta\)-edge meets either
class.  Together
with the matchings, this shows that
\(C_{\mathbf x}\cup C_{\mathbf y}\) has no edge to a third profile class.
Connectedness yields
\(V(G)=C_{\mathbf x}^G\mathbin{\dot\cup}C_{\mathbf y}^G\), and analogously
for \(H\).  After deleting the two internal matchings, the remaining bipartite
graph is \(\kappa\)-regular, where
\(\kappa=\sum_{\delta\notin\{\xi,\chi\}}x_\delta
   =\sum_{\delta\notin\{\xi,\chi\}}y_\delta\).

Let \(D_\xi\) be the diagonal projector onto the \(\xi\)-matching edges and
define
\[
 M_\xi=A_{L(G)}D_\xi A_{L(G)},
 \qquad
 W_\xi=A_{L(G)}\circ M_\xi.
\]
Since \(D_\xi\) and \(A_{L(G)}\) are transported, so are
\(M_\xi,W_\xi\), and \(W_\xi^2\), and all are constant on basis relations.
For an adjacent pair, a common \(\xi\)-edge is either the unique matching
edge at an \(\mathbf x\)-center or the closing matching edge at an
\(\mathbf y\)-center; simplicity prevents both alternatives from occurring
at once.  Hence \(W_\xi\) is \(0\)-\(1\).  The \(W_\xi\)-subgraph indexed by
a fixed \(\xi\)-edge \(z\bar z\) has as its vertices the cross edges incident
with \(z\) or \(\bar z\).  Two are \(W_\xi\)-adjacent exactly when they share
the same \(\mathbf x\)-endpoint, or when they share a
\(\mathbf y\)-endpoint across the \(\xi\)-matching.  Thus it consists of two
\(\kappa\)-cliques, together with the possibly partial cross-matching
\(zy\leftrightarrow\bar zy\).  Two distinct vertices in one clique have
exactly the other \(\kappa-2\) clique vertices as common \(W_\xi\)-neighbors,
whereas a cross-match pair has no common \(W_\xi\)-neighbor.  Therefore
\[
 (W_\xi^2)_{ef}=
 \begin{cases}
 \kappa-2,&e\ne f\text{ lie in the same clique},\\
 0,&e,f\text{ form a cross-match pair}.
 \end{cases}
 \tag{A.8}
\]
The preceding construction makes the \(\mathbf x\)-centered pair a
same-clique pair and the \(\mathbf y\)-centered pair a cross match.  Here
\(\kappa\ge2\).  If \(\kappa=2\), every vertex of the resulting \(K_4\) has no
additional incident edge, so connectedness implies that the root is \(K_4\).
Edge transitivity of \(K_4\) then puts all its edges in one
diagonal line-algebra fiber, so every vertex has the same star profile and
\(\mathbf x=\mathbf y\).  Thus \(\kappa\ge3\), and
the values in
(A.8) differ, contradicting constancy or transport of \(W_\xi^2\) on
\(R_\gamma\).  This proves the central-profile assertion; it remains to
compute the class sizes.

On the \(G\)-side, for an incident fiber \(\mathcal E_\alpha\),
\[
 |C_{\mathbf x}^G|x_\alpha=
 \begin{cases}
 |\mathcal E_\alpha|,&\mathcal E_\alpha\text{ joins }\mathbf x
   \text{ to a distinct profile},\\
 2|\mathcal E_\alpha|,&\mathcal E_\alpha\text{ joins }\mathbf x
   \text{ to itself}.
 \end{cases}
\]
Every star-profile class has such a fiber with \(x_\alpha>0\), because the
connected nontrivial roots under consideration have no isolated vertices.
Since
\(|\mathcal E_\alpha|=\operatorname{tr}(D_\alpha)\) is transported, these formulas
recover and transport every profile-class size.
\end{proof}

\begin{lemma}[Endpoint dart profiles]
\label{lem:bt-dart-endpoint}
Every diagonal fiber of \(\mathcal A_{L(G)}\) fixes its unordered pair of
endpoint dart profiles, and every open adjacent relation \(R_\gamma\)
fixes its source dart, namely
\(\mathbf d(e\cap f,e)=\mathbf e_\gamma+h^{\mathrm d}_\gamma\).  These
labels are transported by \(\varphi\).
\end{lemma}

\begin{proof}
For a source fiber \(\mathcal E_\alpha\), every line neighbor
of \(e=uv\) shares exactly one endpoint with \(e\), so the two wings split
the out-valencies:
\[
 \mathbf d(u,e)+\mathbf d(v,e)=(\nu_\rho)_\rho
 \qquad(e=uv\in\mathcal E_\alpha).
 \tag{A.9}
\]
For an open relation \(\gamma\), the common line neighbors \(g\) of the
wedge \(e=uv,f=uw\) with \((e,g)\in R_\rho\) are exactly the
\(\rho\)-colored edges of the \(u\)-wing other than \(f\), so the open-star
listing gives
\[
 \mathbf d(e\cap f,e)=\mathbf e_\gamma+h^{\mathrm d}_\gamma.
 \tag{A.10}
\]
If an open relation leaves \(\mathcal E_\alpha\), (A.10)
fixes one endpoint dart, and (A.9) fixes the other.  If no open relation
leaves, the true-twin automorphism of
Lemma~\ref{lem:bt-star-endpoint} swaps the two
wings and makes the profiles equal.  This proves the diagonal assertion,
and (A.10) is the open source-dart assertion.
\end{proof}

\begin{lemma}[Source and target dart profiles]
\label{lem:bt-dart-profiles}
For a representative \(e=uv,f=uw\) of an adjacent basis relation,
the relation fixes the source dart \(\mathbf d(u,e)\) and the target dart
\(\mathbf d(u,f)\), and these labels are transported by \(\varphi\).
\end{lemma}

\begin{proof}
Consider the source dart; the target statement
follows by transpose.  If the two endpoint star profiles of its source
fiber differ, Lemmas~\ref{lem:bt-star-endpoint}
and~\ref{lem:bt-star-profiles} identify the central wing,
and the fixed unordered dart pair, together with (A.1), identifies its dart
profile: under (A.1), the two dart profiles recover the two distinct endpoint
star profiles, so the fixed central star profile selects the correct dart.
For equal endpoint star profiles, no adjacent relation can be centered at a
degree-one endpoint.  At degree two, each dart has one other incident edge,
so relation \(\gamma\) fixes the unit dart profile \(\mathbf e_\gamma\).
An open relation at any degree is fixed by (A.10).

Suppose that a closed relation \(\gamma\), centered at a root vertex of
degree \(d\ge3\), is mixed between source-dart profiles
\(\mathbf x\ne\mathbf y\).  Interchanging \(G,H\) if necessary, take the
\(\mathbf x\)-centered representative in \(G\); the
\(\mathbf y\)-centered representative may lie in \(G\) or in \(H\).  In the
\(\mathbf x\)-centered representative, let \(\rho\) be the relation of the first
closing leg \((e,g)\), where \(g\) is the closing edge.  Direct listing gives
\[
 \mathbf d(u,e)=\mathbf e_\gamma+h^{\mathrm d}_\gamma-\mathbf e_\rho,
\]
of the form required by the profile-difference observation, with the closing-leg
relation as label.  With \(\sigma\) the closing-leg relation of a
\(\mathbf y\)-centered representative, the principle gives
\[
 \mathbf y-\mathbf x=\mathbf e_\rho-\mathbf e_\sigma,
 \tag{A.11}
\]
and, since Lemma~\ref{lem:bt-dart-endpoint} fixes and transports the
unordered dart pair of the source fiber, that pair equals
\(\{\mathbf x,\mathbf y\}\) and these are the only source-dart profiles on
\(R_\gamma\); in particular \(\rho\ne\sigma\).  Write
\[
 (x_\rho,y_\rho)=(a,a+1),\qquad
 (x_\sigma,y_\sigma)=(b+1,b).
 \tag{A.12}
\]

We first prove \(a,b>0\).  If \(a=0\), then (A.9) gives the out-valency
\(\nu_\rho=x_\rho+y_\rho=1\).  Thus \(R_\rho\) has exactly one target from
every source edge.  Write \(f_\rho(e)\) for this target and form
\[
 Z_\rho^{L(G)}=((R_\rho A_{L(G)})\circ A_{L(G)})A_{L(G)},
 \qquad
 Z_\rho^{L(H)}=\varphi(Z_\rho^{L(G)}).
\]
Functionality also forces \(\rho\ne\gamma\): in the
\(\mathbf x\)-centered closed realization, otherwise the other wedge leg
and the closing edge would be two distinct \(R_\rho\)-targets of \(e\).
Hence in the \(\mathbf y\)-centered realization the unique
\(R_\rho\)-target is a third edge at the center.
On either side \(K\in\{G,H\}\), functionality makes
\((Z_\rho^{L(K)})_{ef}\) count line vertices
adjacent simultaneously to \(e,f,f_\rho(e)\).  It is zero when
\(f_\rho(e)\) is the closing edge and \(d-3\) when \(f_\rho(e)\) is a third
edge at the center.  Indeed, functionality gives
\((R_\rho^{L(K)}A_{L(K)})_{eg}
 =A_{L(K)}(f_\rho(e),g)\); the Hadamard product
restricts \(g\) to a common neighbor of \(e\) and \(f_\rho(e)\), and the
final multiplication by \(A_{L(K)}\) additionally requires \(g\sim f\).
For \(d\ge4\) this contradicts constancy or transport on \(R_\gamma\).  If
\(d=3\), distinguish the two representatives as follows.  In the star
realization write \(e=uv,f=uw,f_\rho(e)=uz\).  Write the closing
representative as
\(\hat e=\hat u\hat v,\hat f=\hat u\hat w\), with closing edge
\(\hat g=\hat v\hat w\); this hatted realization is the
\(\mathbf x\)-centered one in \(G\), while the star realization lies on the
side containing the \(\mathbf y\)-centered representative.
Then \((\hat e,\hat g)\in R_\rho\); define \(R_\theta\) by
\((\hat g,\hat f)\in R_\theta\).  Functionality gives
\(p_{\rho\theta}^{\gamma}=1\), and transport forces
\((uz,f)\in R_\theta\) in the star realization.  Both pairs through the
witness are closed in the closing realization.  Wedge-type transport
(Lemma~\ref{lem:bt-wedge-type}) therefore makes the corresponding wedges
in the star realization closed:
they supply \(vz,wz\), while the original closed \(\gamma\)-wedge supplies
\(vw\).

In the star realization, \(R_\gamma,R_\rho,R_\theta\) are centered at the
cubic vertex \(u\).  In the closing realization these three relations are
centered at \(\hat u,\hat v,\hat w\), respectively.  Central-degree
transport (Lemma~\ref{lem:bt-wedge-type}) therefore makes all three
vertices cubic.  The target edge \(\hat g\) of
\(R_\rho\) and the target edge \(\hat f\) of \(R_\theta\) consequently have
endpoint-degree pair \(\{3,3\}\); transporting their diagonal-fiber labels
back gives
\(d(z)=d(w)=3\).  Finally, the equal endpoint star profiles of the source
fiber give \(d(v)=d(u)=3\), since degree is the sum of the star-profile
coordinates.  Hence the displayed \(K_4\) is cubic and, by connectedness, is
the whole root.  Its darts are automorphic, contradicting
\(\mathbf x\ne\mathbf y\).  Therefore
\(a>0\); the symmetric argument gives \(b>0\).

Every mixed relation is closed by (A.10).  By (A.9), the source fiber's
fixed dart pair \(\{\mathbf x,\mathbf y\}\) has total valency vector
\(\mathbf x+\mathbf y\).  Thus the closing \(\rho\)-leg opposite an
\(\mathbf x\)-dart is centered at the \(\mathbf y\)-dart, while
\(x_\rho=a>0\) gives a \(\rho\)-relation centered at an \(\mathbf x\)-dart;
hence \(\rho\) is mixed, and
symmetrically so is \(\sigma\).  For
\(\lambda\in\{\rho,\sigma\}\), let \(r_{\mathbf x}\) and
\(r_{\mathbf y}\) be its closing-leg labels in its
\(\mathbf x\)- and \(\mathbf y\)-centered realizations.  Applying the listing
that produced (A.11) to \(R_\lambda\) gives
\[
 \mathbf y-\mathbf x
 =\mathbf e_{r_{\mathbf x}}-\mathbf e_{r_{\mathbf y}}.
\]
Comparison with (A.11) forces
\(r_{\mathbf x}=\rho\) and \(r_{\mathbf y}=\sigma\).

Fix \(e=uv\), with dart profiles \(\mathbf x\) at \(u\) and \(\mathbf y\)
at \(v\).  If a
\(\rho\)- or \(\sigma\)-neighbor in the \(u\)-wing is \(f=uz\), the closing
rule supplies the edge \(vz\), which is a \(\rho\)-neighbor in the
\(v\)-wing.  The two source sets are disjoint because
\(R_\rho\ne R_\sigma\), and simplicity makes the map \(uz\mapsto vz\)
injective.  Therefore
\[
 x_\rho+x_\sigma\le y_\rho.
\]
By (A.12), this is \(a+(b+1)\le a+1\), contradicting \(b>0\).
\end{proof}

\subsection{The profile-oriented ordinary algebra}

Let \(\mathfrak C_G\) and \(\mathfrak C_H\) be the recovered star-profile
classes of \(G\) and \(H\), with the transport bijection written
\(C\mapsto C^H\).  Let \(B_G\in\{0,1\}^{V(G)\times E(G)}\) and
\(B_H\in\{0,1\}^{V(H)\times E(H)}\) be the unsigned incidence matrices, so
\(B_G(u,e)=\one[u\in e]\), and similarly for \(H\).
For \(C\in\mathfrak C_G\), define the diagonal
\(0\)-\(1\) projectors
\[
 (P_C^G)_{uu'}=\one[u=u'\in C],
 \qquad
 (P_{C^H}^H)_{u'v'}=\one[u'=v'\in C^H].
\]
For a diagonal
edge fiber \(\mathcal E_\alpha\), let
\(\{\mathbf d_\alpha^+,\mathbf d_\alpha^-\}\) be its unordered endpoint
dart-profile pair.  If the two profiles differ, choose the same ordering of
their transported profile labels on \(G\) and \(H\), and orient every
\(\alpha\)-edge toward its \(+\)-dart.  Let
\(\widetilde B_{G,\alpha}\in
\{-1,0,1\}^{V(G)\times E(G)}\) be the signed incidence matrix with entry
\(+1\) at the \(+\)-dart endpoint, entry \(-1\) at the other endpoint, and
zero off that fiber; define \(\widetilde B_{H,\alpha}\) analogously.  This
ordering serves only to synchronize the transported labels;
equal-dart-profile fibers are not oriented.  Put
\[
 \mathcal I_G=\{B_G\}\cup
 \{\widetilde B_{G,\alpha}:\mathbf d_\alpha^+\ne\mathbf d_\alpha^-\},
 \qquad
 \mathcal I_H=\{B_H\}\cup
 \{\widetilde B_{H,\alpha}:\mathbf d_\alpha^+\ne\mathbf d_\alpha^-\}.
\]
For each \(S_G\in\mathcal I_G\), write \(S_H\in\mathcal I_H\) for the
incidence matrix with the same transported label.  The next lemma relates the
root profile projectors to the stable line algebra.  Placing a profile-class
projector between two incidence matrices produces an edge-indexed Gram matrix
that belongs to the stable line algebra and is transported.

\begin{lemma}[Mixed Grams]
\label{lem:bt-mixed-grams}
For \(S_G,T_G\in\mathcal I_G\) and \(C\in\mathfrak C_G\),
\[
 S_G^\top P_C^G T_G\in\mathcal A_{L(G)},
\]
and
\[
 \varphi(S_G^\top P_C^G T_G)=S_H^\top P_{C^H}^H T_H.
\]
\end{lemma}

\begin{proof}
The \((e,f)\)-entry is
\[
 \sum_{u\in e\cap f\cap C}
   S_G(u,e)\,T_G(u,f).
 \tag{A.13}
\]
It is zero for disjoint edges; the source and target fibers also determine
whether either signed incidence is zero.  On an adjacent basis relation,
Lemmas~\ref{lem:bt-star-profiles} and~\ref{lem:bt-dart-profiles} fix the
central star profile and, whenever a signed incidence occurs, its source or
target sign.  On a diagonal relation, the unordered endpoint dart profiles
(Lemma~\ref{lem:bt-dart-endpoint}) fix both star memberships through
(A.1).  A signed factor can then be
nonzero only for the fiber containing \(e=f\).  If both factors are unsigned
or both are signed, (A.13) counts the endpoints in \(C\).  With exactly one
signed factor it is the signed \(C\)-incidence; this is zero when both
endpoints have the same star profile, and otherwise its sign is fixed by the
synchronized \(+/-\) dart labels.  Thus (A.13) is constant on every basis
relation of \(\mathcal A_{L(G)}\) and depends only on transported labels.
Consequently
\[
 S_G^\top P_C^G T_G=\sum_\tau m_\tau R_\tau
\]
with the same coefficients \(m_\tau\) after transport to \(H\), proving both
claims.
\end{proof}

Define the root-side space
\[
 \mathcal W_G=\operatorname{span}\left(
   \{P_C^G:C\in\mathfrak C_G\}\cup
   \{P_C^G S_G M T_G^\top P_D^G:
     C,D\in\mathfrak C_G,\ S_G,T_G\in\mathcal I_G,\;
     M\in\mathcal A_{L(G)}\}
 \right).
 \tag{A.14}
\]
Define \(\mathcal W_H\) by the same formula with the transported \(H\)-side
objects.  Here \(M\) is indexed by \(E(G)\times E(G)\), so each incidence
sandwich is a \(V(G)\times V(G)\) root matrix.  The next lemma makes
\(\mathcal W_G\) an ordinary matrix algebra; Hadamard closure is proved only
after its value cells are constructed.
We call the matrices \(P_C^G\) \emph{projector atoms} and the displayed
incidence sandwiches \emph{incidence atoms}.

\begin{lemma}[Root-algebra transport and compression]
\label{lem:bt-trace-transport}
The space \(\mathcal W_G\) is a unital real algebra---that is, it contains
the identity matrix and is closed under real linear combinations and ordinary
matrix multiplication---and is closed under transpose.
The atomwise prescriptions
\[
 \begin{aligned}
 P_C^G&\longmapsto P_{C^H}^H,\\
 P_C^G S_G M T_G^\top P_D^G
 &\longmapsto
 P_{C^H}^H S_H\varphi(M)T_H^\top P_{D^H}^H
 \end{aligned}
 \tag{A.15}
\]
extend to a well-defined, unital, transpose- and trace-preserving algebra
isomorphism
\[
 \Omega:\mathcal W_G\longrightarrow\mathcal W_H.
\]
For every \(Z\in\mathcal W_G\) and \(U_G,V_G\in\mathcal I_G\),
\[
 U_G^\top ZV_G\in\mathcal A_{L(G)}
 \tag{A.16}
\]
compatibly with \(\Omega,\varphi\).  Moreover,
\(I_{V(G)},J_{V(G)},A_G\) belong to \(\mathcal W_G\) and are transported to
\(I_{V(H)},J_{V(H)},A_H\).
\end{lemma}

\begin{proof}
\emph{Algebra structure.}
The identity is \(\sum_{C\in\mathfrak C_G}P_C^G\), and transposing an atom
exchanges its two outer profile blocks and incidence matrices and replaces
\(M\) by \(M^\top\).  Multiplication by a projector atom either annihilates
a mismatched outer block or leaves the matching block unchanged, so only
products of two incidence atoms require calculation.  Products with
unequal middle projectors vanish; the remaining products reduce as
\[
 (P_C^G S_G M T_G^\top P_D^G)
 (P_D^G U_G Q V_G^\top P_F^G)
 =P_C^G S_G\bigl(M(T_G^\top P_D^G U_G)Q\bigr)V_G^\top P_F^G.
\]
The middle factor belongs to \(\mathcal A_{L(G)}\) by
Lemma~\ref{lem:bt-mixed-grams}.  Thus \(\mathcal W_G\) is a unital
transpose-closed algebra, and (A.15) is formally multiplicative and
transpose preserving.

\emph{Well-definedness.}
To check independence of the chosen expression, first note that a product
whose profile-block indices do not close has trace zero on both sides.
For a cyclic product of incidence atoms, with \(C_r=C_0\), cyclicity of trace
gives the explicit
contraction
\[
\begin{aligned}
 &\operatorname{tr}\!\left(
   \prod_{i=1}^{r}
   P_{C_{i-1}}^G S_i M_i T_i^\top P_{C_i}^G
 \right)\\
 &\qquad=
 \operatorname{tr}\!\left(
   M_1(T_1^\top P_{C_1}^G S_2)
   M_2\cdots
   M_r(T_r^\top P_{C_0}^G S_1)
 \right),
\end{aligned}
\]
where indices are cyclic.  Each parenthesized factor is a transported mixed
Gram, so the right-hand side is a transported trace in
\(\mathcal A_{L(G)}\).  Projector-only words have trace zero or a transported
class size, and any other projector factors are absorbed into adjacent
profile blocks.  By linearity, all trace words in the spanning atoms are
therefore preserved.  If a formal linear
combination \(Z\) represents zero and \(\widetilde Z\) is its formal image
under (A.15), then
\[
 \operatorname{tr}(\widetilde Z^\top\widetilde Z)
 =\operatorname{tr}(Z^\top Z)=0.
\]
The left side is the squared real Frobenius norm of \(\widetilde Z\), so
\(\widetilde Z=0\).  The prescription is therefore independent of the chosen
representation.  The same construction with \(\varphi^{-1}\) is inverse on
every spanning atom, so \(\Omega\) is bijective.

\emph{Compression.}
Compressing an incidence atom gives
\[
 U_G^\top(P_C^G S_G M T_G^\top P_D^G)V_G
 =(U_G^\top P_C^G S_G)M(T_G^\top P_D^G V_G)
 \in\mathcal A_{L(G)},
\]
while for a projector atom this is the mixed-Gram lemma.  This proves
(A.16) and its transport compatibility.

\emph{The matrices \(I\), \(J\), and \(A\).}
Let \(D_G=\sum_{C\in\mathfrak C_G}d(C)P_C^G\), where \(d(C)\) is
the common root degree on profile class \(C\), and define \(D_H\)
analogously.  Since
\(d(C)=\sum_\alpha s_\alpha(C)\), star-profile transport gives
\(\Omega(D_G)=D_H\).  The roots here are connected and have order at least
three, so all degrees are positive; hence
\(D_G^{-1}=\sum_C d(C)^{-1}P_C^G\) lies in \(\mathcal W_G\) and is
transported to \(D_H^{-1}\).  The line-side
all-ones matrix \(J_{E(G)}\) belongs to \(\mathcal A_{L(G)}\), and entry
inspection gives
\[
 B_GB_G^\top=D_G+A_G,
 \qquad
 B_GJ_{E(G)}B_G^\top=D_GJ_{V(G)}D_G.
 \tag{A.17}
\]
For any \(M\in\mathcal A_{L(G)}\), expanding the root identity on both sides
shows
\[
 B_GMB_G^\top=\sum_{C,D}P_C^GB_GMB_G^\top P_D^G,
\]
so (A.15) transports this unblocked sandwich to
\(B_H\varphi(M)B_H^\top\).  In particular,
\(B_GB_G^\top=B_G I_{E(G)}B_G^\top\), and both sandwiches in (A.17) are sums
of the profile-blocked atoms in (A.14).  Therefore the matrices
\[
 I_{V(G)}=\sum_CP_C^G,\qquad
 A_G=B_GB_G^\top-D_G,\qquad
 J_{V(G)}=D_G^{-1}B_GJ_{E(G)}B_G^\top D_G^{-1}.
\]
belong to \(\mathcal W_G\) and are transported.
\end{proof}

\subsection{Diagonal and root-edge profile cells}

For any finite index set \(S\) and square matrix \(M\) indexed by \(S\), write
\[
 \operatorname{Diag}_S(M)=I_S\circ M
\]
for coordinate-diagonal extraction.  We use \(S=V(G),V(H)\) for root
matrices and \(S=E(G),E(H)\) for line-indexed matrices.  Call a diagonal
edge fiber
\emph{unequal} or
\emph{equal} according as its two endpoint dart profiles differ or agree.
For a fiber \(\mathcal E_\alpha\), set
\[
 \Delta_\alpha=\sum_{C\in\mathfrak C_G}s_\alpha(C)P_C^G,
\]
where \(s_\alpha(C)\) is the common coordinate on \(C\); this is the
diagonal matrix with entry \(s_\alpha(u)\) at \(u\).  Thus
\(\Delta_\alpha\in\mathcal W_G\) directly, without using diagonal
extraction.  If the fiber is unequal, also set
\[
 \Sigma_\alpha=
 \frac12\left(B_GD_\alpha\widetilde B_{G,\alpha}^\top+
               \widetilde B_{G,\alpha}D_\alpha B_G^\top\right).
\]
This incidence formula lies in \(\mathcal W_G\).  Edge by edge, its two
rank-one terms cancel at the two off-diagonal endpoint positions and
contribute \(+1\) and \(-1\) at the two diagonal positions.  Thus its
diagonal entry is the number of \(+\)-darts minus the number of \(-\)-darts
of that fiber.  Both matrices are transported.

\begin{lemma}[Diagonal extraction]
\label{lem:bt-diagonal-extraction}
For every \(Z\in\mathcal W_G\),
\[
 \operatorname{Diag}_{V(G)}(Z)\in\mathcal W_G,
 \qquad
 \Omega(\operatorname{Diag}_{V(G)}(Z))
 =\operatorname{Diag}_{V(H)}(\Omega Z).
 \tag{A.18}
\]
\end{lemma}

\begin{proof}
Projector atoms are already diagonal, so it suffices to treat an incidence atom
\[
 F=P_C^G S_G M T_G^\top P_C^G.
\]
Terms with different outer projectors have zero diagonal.  Explicitly,
\[
 F_{uu}
 =\one[u\in C]\sum_{e,f\ni u}S_G(u,e)M_{ef}T_G(u,f).
\]
For fixed \((u,e)\),
separate \(f=e\) and group the remaining \(f\)'s by adjacent basis relation
\(R_\rho\).  The matrix \(M\) is constant on each \(R_\rho\), the coordinate
\(\mathbf d(u,e)_\rho\) gives its multiplicity, and the recovered target dart profile
fixes the corresponding entry of \(T_G\).  Hence the inner sum depends only on
the source fiber together with the dart profile and its \(S_G\)-entry.  The
separated \(f=e\) contribution is likewise fixed by the source diagonal
fiber and the two incidence labels.  More explicitly, for an incident dart
\((u,e)\) whose edge lies in \(\mathcal E_\alpha\) and whose dart profile is
\(\delta\), set
\[
 \kappa_{\alpha,\delta}^{S,T,M}
 :=S_G(u,e)\sum_f M_{ef}T_G(u,f).
\]
The preceding grouping proves that this number is independent of the chosen
incident dart with fiber \(\alpha\) and profile \(\delta\).  Thus the full
diagonal sum groups by dart type with a concrete, well-defined coefficient.

Write \(\delta_\alpha\) for the common endpoint dart profile of an equal
fiber and \(\delta_\alpha^+,\delta_\alpha^-\) for the ordered profiles of an
unequal fiber.  An equal fiber contributes
\(\kappa_{\alpha,\delta_\alpha}^{S,T,M}
  (\Delta_\alpha)_{uu}\).  For an unequal fiber, the two dart counts satisfy
\[
 n_\pm(u)=\frac{(\Delta_\alpha)_{uu}\pm(\Sigma_\alpha)_{uu}}2.
\]
Consequently the exact matrix identity is
\[
 \operatorname{Diag}_{V(G)}(F)
 =\sum_\alpha P_C^G
 \begin{cases}
  \kappa_{\alpha,\delta_\alpha}^{S,T,M}\Delta_\alpha,
    &\mathcal E_\alpha\text{ equal},\\[2mm]
  \displaystyle
  \frac{\kappa_{\alpha,\delta_\alpha^+}^{S,T,M}
        +\kappa_{\alpha,\delta_\alpha^-}^{S,T,M}}2\,\Delta_\alpha
  +\frac{\kappa_{\alpha,\delta_\alpha^+}^{S,T,M}
        -\kappa_{\alpha,\delta_\alpha^-}^{S,T,M}}2\,\Sigma_\alpha,
    &\mathcal E_\alpha\text{ unequal}.
 \end{cases}
\]
Every summand lies in \(\mathcal W_G\).  The same coefficients occur on
\(H\), while \(P_C^G,\Delta_\alpha,\Sigma_\alpha\) are carried to their
corresponding \(H\)-matrices.  This proves both assertions in (A.18).
\end{proof}

Equal-dart-profile fibers require no canonical orientation.  Fix one such
fiber \(\mathcal E_\alpha\), choose an arbitrary orientation only for the
following proof, and let \(\widehat B_\alpha\) have entry \(+1\) at the
chosen positive endpoint, \(-1\) at the other endpoint, and zero elsewhere,
including on columns outside \(\mathcal E_\alpha\).  Put
\(K_\alpha=\widehat B_\alpha^\top B_G\).
For an adjacent \(f\), the value \(K_\alpha(e,f)\) is \(+1\) or \(-1\) according
to which oriented endpoint of \(e\) is shared.  Equality of the two endpoint
dart profiles therefore says, for every adjacent basis relation \(R_\rho\)
and every \(e\in\mathcal E_\alpha\),
\[
 \sum_{f:(e,f)\in R_\rho}K_\alpha(e,f)=0.
 \tag{A.19}
\]
Indeed, the summands are \(+1\) for the \(R_\rho\)-neighbors in the positive
wing of \(e\) and \(-1\) for those in the negative wing, and equality of the
two endpoint dart profiles equates the two counts.

\begin{lemma}[Equal-fiber cancellation]
\label{lem:bt-equal-cancellation}
For every \(Z\in\mathcal W_G\),
\[
 \operatorname{Diag}_{E(G)}(\widehat B_\alpha^\top ZB_G)=0,
 \qquad
 \operatorname{Diag}_{E(G)}(B_G^\top Z\widehat B_\alpha)=0.
 \tag{A.20}
\]
Consequently, for every edge \(e=uv\in\mathcal E_\alpha\),
\[
 Z_{uu}=Z_{vv},
 \qquad
 Z_{uv}=Z_{vu}.
 \tag{A.21}
\]
\end{lemma}

\begin{proof}
For \(S_G\in\mathcal I_G\) and a profile class \(C\),
\[
 \widehat B_\alpha^\top P_C^G S_G
 =K_\alpha\circ(B_G^\top P_C^G S_G).
\]
For disjoint edge pairs both sides vanish.  For an adjacent off-diagonal
pair this is the unique shared-endpoint contribution; on the
diagonal both sides vanish because the two endpoints have the same star
profile and no \(\widetilde B_{G,\beta}\in\mathcal I_G\) is supported on
this equal fiber.
For an incidence atom, put
\(Q=B_G^\top P_C^G S_G\in\mathcal A_{L(G)}\) and
\(Q'=M T_G^\top P_D^G B_G\in\mathcal A_{L(G)}\).  The \(e\)-diagonal entry of
\(\widehat B_\alpha^\top(P_C^G S_G M T_G^\top P_D^G)B_G\) is
\[
 \sum_f K_\alpha(e,f)Q(e,f)Q'(f,e).
 \]
For fixed \(e\), \(Q(e,f)\) is constant on \(R_\rho\), while \(Q'(f,e)\) is
constant on \(R_\rho^\top\).  Diagonal and disjoint relations contribute
zero, and (A.19) annihilates each adjacent relation.  Projector atoms vanish
similarly.  This proves the first identity in (A.20).  Apply it to
\(Z^\top\in\mathcal W_G\) and transpose the resulting diagonal to obtain the
second.

Orient \(e=uv\) positively at \(u\).  The two equations at \(e\) are
\[
 Z_{uu}+Z_{uv}-Z_{vu}-Z_{vv}=0,
 \qquad
 Z_{uu}+Z_{vu}-Z_{uv}-Z_{vv}=0.
\]
Their sum and difference give (A.21).
The conclusion no longer contains \(\widehat B_\alpha\), so it is independent
of the arbitrary orientation used in the proof.
\end{proof}

Fix a basis \(Z_1,\ldots,Z_t\) of \(\mathcal W_G\) and the transported
basis \(\Omega Z_1,\ldots,\Omega Z_t\) of \(\mathcal W_H\).  For an ordered
root pair on each side define its \emph{complete \(\mathcal W\)-profile}
\[
 \begin{aligned}
 \zeta_G(u,v)&=((Z_1)_{uv},\ldots,(Z_t)_{uv}),\\
 \zeta_H(u',v')&=((\Omega Z_1)_{u'v'},\ldots,
                    (\Omega Z_t)_{u'v'}).
 \end{aligned}
\]
Replacing the chosen basis applies an invertible linear transformation to
every profile vector and therefore leaves its joint level sets unchanged.
Thus the profile cells defined below are basis independent.
Use the finite common index set
\[
 \Lambda
 =\{\zeta_G(u,v):u,v\in V(G)\}
  \cup\{\zeta_H(u',v'):u',v'\in V(H)\}.
\]
Thus \(\Lambda\) includes values attained on only one root; the induction will
show that no such one-sided nonempty cell remains at a fixed distance.  For
every nonnegative integer \(d\) and every \(\lambda\in\Lambda\), define
\[
 P^G_{d,\lambda}(u,v)
 =\one[\operatorname{dist}_G(u,v)=d,\ \zeta_G(u,v)=\lambda].
\]
Define \(P^H_{d,\lambda}\) by the same formula using \(\zeta_H\) and
\(\operatorname{dist}_H\).  These matrices are defined even when their
support is empty; we call a nonzero one a \emph{profile cell}.  We suppress a
superscript only in a one-root calculation.  Every matrix in
\(\mathcal W_G\) is constant on each \(G\)-profile cell by definition.

\begin{lemma}[Distance-zero and distance-one cells]
\label{lem:bt-base-cells}
For every \(\lambda\in\Lambda\) and \(d\in\{0,1\}\), the matrix
\(P^G_{d,\lambda}\) lies in \(\mathcal W_G\), and
\(\Omega(P^G_{d,\lambda})=P^H_{d,\lambda}\), including when either displayed
matrix is zero.
\end{lemma}

\begin{proof}
\emph{Diagonal cells.}
Diagonal extraction and polynomial interpolation of commuting diagonal
matrices construct every profile cell on \(I_{V(G)}\).  Namely, apply finite
ordinary polynomial interpolation jointly to
\[
 \operatorname{Diag}_{V(G)}(Z_1),\ldots,
 \operatorname{Diag}_{V(G)}(Z_t).
\]
Their ordinary products evaluate coordinatewise because the matrices are
diagonal, so this step does not assume Hadamard closure of \(\mathcal W_G\).
For each coordinate use the union of the values attained by the corresponding
diagonal matrices on \(G\) and \(H\), as in the interpolation convention
above.  Equation (A.18) therefore transports the resulting diagonal selector,
including a zero selector for a requested value unattained on one side.

\emph{Unequal fibers.}
Let \(\mathcal E_\alpha\) be unequal and define its endpoint-selector
incidences
\[
 B_{G,\alpha}^\pm
 =\frac12(B_GD_\alpha\pm\widetilde B_{G,\alpha}).
\]
Each \(\alpha\)-column selects its \(\pm\)-endpoint.  Since
\(B_{G,\alpha}^\pm\) is a linear combination of
\(B_GD_\alpha,\widetilde B_{G,\alpha}\), multiplying \(D_\alpha\) onto
(A.16) shows that
\((B_{G,\alpha}^s)^\top ZB_{G,\alpha}^t\in\mathcal A_{L(G)}\) for
\(s,t\in\{+,-\}\).  Its
\(e\)-diagonal entry is \(Z_{u_su_t}\) for \(e=u_+u_-\).  Joint level sets
over the basis \(Z_i\), intersected with the diagonal projector \(D_\alpha\),
therefore give transported diagonal projectors
\(\Pi_{\alpha,\lambda}^{+-}\) and
\(\Pi_{\alpha,\lambda}^{-+}\).  Explicitly,
\(\Pi_{\alpha,\lambda}^{+-}\) selects the edge
\(e=u_+u_-\) when \(\zeta_G(u_+,u_-)=\lambda\), while
\(\Pi_{\alpha,\lambda}^{-+}\) selects it when
\(\zeta_G(u_-,u_+)=\lambda\).  They are transported because all the
compressed basis matrices used in their joint level sets are transported.
The ordered
root-edge cell on this fiber is
\[
 B_{G,\alpha}^+\Pi_{\alpha,\lambda}^{+-}(B_{G,\alpha}^-)^\top
 +B_{G,\alpha}^-\Pi_{\alpha,\lambda}^{-+}(B_{G,\alpha}^+)^\top.
 \tag{A.22}
\]
After expanding \(B_{G,\alpha}^\pm\), absorb every \(D_\alpha\) and
\(\Pi\)-factor into the line-algebra middle matrix.  Inserting
\(I_{V(G)}=\sum_C P_C^G\) on both root sides then expresses every resulting
term as an incidence atom in (A.14).  Thus (A.22) lies in \(\mathcal W_G\) and
is transported.

\emph{Equal fibers.}
If \(\mathcal E_\alpha\) is equal,
Lemmas~\ref{lem:bt-diagonal-extraction} and
\ref{lem:bt-equal-cancellation} apply as follows.  Equation (A.16) places
both compressed terms below in \(\mathcal A_{L(G)}\); (A.18) transports the
extracted root diagonal, and Hadamard multiplication by the transported
\(D_\alpha\) preserves the line algebra.  Hence the diagonal line matrix
\[
 b_Z^\alpha=
 \frac12D_\alpha\circ
 \left(B_G^\top ZB_G
       -B_G^\top\operatorname{Diag}_{V(G)}(Z)B_G\right)
\]
has value \(Z_{uv}=Z_{vu}\) at \(e=uv\).  Indeed,
\[
 (B_G^\top ZB_G)_{ee}
 -(B_G^\top\operatorname{Diag}_{V(G)}(Z)B_G)_{ee}
 =Z_{uv}+Z_{vu}=2Z_{uv},
\]
where the last equality is (A.21).  Interpolate the joint values of
the \(b_{Z_i}^\alpha\) to obtain a diagonal line projector
\(\Pi_{\alpha,\lambda}\), per the interpolation convention; then
\[
 B_G\Pi_{\alpha,\lambda}B_G^\top
 -\operatorname{Diag}_{V(G)}
    (B_G\Pi_{\alpha,\lambda}B_G^\top)
\]
is the corresponding symmetric ordered root-edge cell.  Because the root is
simple, its off-diagonal support consists exactly of the two orientations of
the selected root edges; diagonal extraction removes only the endpoint
counts.  Summing the fiberwise selectors with the same \(\lambda\) proves
the lemma.  The fiberwise selectors use the same common value index
\(\Lambda\), so an unattained requested value contributes the zero
projector.  Applying the same construction to \(\Omega^{-1}\) shows that
nonzero distance-zero and distance-one cells correspond bijectively and,
more strongly, establishes the displayed equality for every
\(\lambda\in\Lambda\).
\end{proof}

\subsection{Induction on root distance}

We extend Lemma~\ref{lem:bt-base-cells} by strong induction on root
distance.  The induction uses ordered edge pairs whose four endpoint distances
identify a unique endpoint pair at the current distance.

\begin{lemma}[Profile cells at all root distances]
\label{lem:bt-distance-induction}
For every nonnegative integer \(d\) and every \(\lambda\in\Lambda\),
\(P^G_{d,\lambda}\) belongs to \(\mathcal W_G\) and
\(\Omega(P^G_{d,\lambda})=P^H_{d,\lambda}\), including zero cells.
\end{lemma}

\begin{proof}
Let
\[
 (\Gamma_s^G)_{uv}=\one[\operatorname{dist}_G(u,v)=s],
 \qquad
 (\Gamma_s^H)_{u'v'}=\one[\operatorname{dist}_H(u',v')=s]
\]
be the \(0\)-\(1\) distance matrices.  The assertion holds for
\(s=0,1\) by Lemma~\ref{lem:bt-base-cells}.  Suppose all profile cells at
distances below \(d\ge2\) have been constructed and transported.  Then
\(\Gamma_s^G=\sum_{\mu\in\Lambda} P^G_{s,\mu}\) and
\(\Gamma_s^H=\sum_{\mu\in\Lambda} P^H_{s,\mu}=\Omega(\Gamma_s^G)\) for \(s<d\).  Hence
\[
 \Gamma_{\ge d}^G=J_{V(G)}-\sum_{s<d}\Gamma_s^G\in\mathcal W_G,
 \qquad \Gamma_{d-2}^G\in\mathcal W_G.
\]
Define the two transported line-algebra matrices
\[
 N_d^{L(G)}=B_G^\top \Gamma_{d-2}^G B_G,
 \qquad
 T_d^{L(G)}=B_G^\top \Gamma_{\ge d}^G B_G.
\]
Their \((e,f)\)-entries count, respectively, the ordered endpoint pairs
\((x,y)\in e\times f\) at distance exactly \(d-2\) and at distance at
least \(d\).  Define the corresponding matrices for \(H\).  Equation
(A.16) gives
\(\varphi(N_d^{L(G)})=N_d^{L(H)}\) and
\(\varphi(T_d^{L(G)})=T_d^{L(H)}\).

\emph{Witness edge-pair identification.}
Finite level-set interpolation in the stable line algebra gives the
transported \(0\)-\(1\) matrix
\[
 Q_d^{L(G)}
 =\one[N_d^{L(G)}>0]\circ\one[T_d^{L(G)}=1]
 \in\mathcal A_{L(G)},
 \qquad
 Q_d^{L(H)}=\varphi(Q_d^{L(G)}).
\]
Thus \(Q_d^{L(G)}\) is a union of stable line-graph basis relations, rather than
necessarily a single basis relation.

If \((Q_d^{L(G)})_{ef}=1\), choose \(x\in e,y\in f\) with
\(\operatorname{dist}_G(x,y)=d-2\), which is possible because
\((N_d^{L(G)})_{ef}\ne0\).  Moving across either root edge changes distance
by at most one, so every pair in \(e\times f\) has distance at most \(d\).
The equality \((T_d^{L(G)})_{ef}=1\) now says that exactly one ordered
endpoint pair has distance \(d\).  We call \((e,f)\) a \emph{witness edge
pair} for this unique distance-\(d\) endpoint pair.  The same argument applies
on \(H\).

Conversely, every ordered pair \((u,v)\) at distance \(d\) has a witness edge
pair.  Take a geodesic \(u=x_0,x_1,\ldots,x_d=v\) and put
\(e=x_0x_1\), \(f=x_{d-1}x_d\).  The four endpoint distances are
\(d,d-1,d-1,d-2\): the displayed geodesic gives the upper bounds, and the
triangle inequality with \(\operatorname{dist}_G(u,v)=d\) gives the matching
lower bounds.  Hence \((N_d^{L(G)})_{ef}\ne0\) and
\((T_d^{L(G)})_{ef}=1\).  This also covers \(d=2\), where
\(x_1=x_{d-1}\).

\emph{Profile readout.}
For \(Z=\sum_i a_iZ_i\in\mathcal W_G\), put
\(z_\mu=\sum_i a_i\mu_i\) for \(\mu\in\Lambda\).  For every \(s<d\), the
already constructed and transported matrix
\[
 L_{s,Z}^G=\sum_{\mu\in\Lambda}z_\mu P^G_{s,\mu}
\]
is entrywise equal to \(Z\circ \Gamma_s^G\), but its membership in
\(\mathcal W_G\) follows directly from the induction hypothesis, without
assuming Hadamard closure.  Define \(L_{s,\Omega Z}^H\) by the same
coefficient formula; then
\(L_{s,\Omega Z}^H=\Omega(L_{s,Z}^G)\).  Now form
\[
 M_{d,Z}^{L(G)}
 =B_G^\top ZB_G-\sum_{s<d}B_G^\top L_{s,Z}^G B_G.
 \tag{A.23}
\]
Hence
\((M_{d,Z}^{L(G)})_{ef}=Z_{uv}\) when \((u,v)\) is the unique distance-\(d\)
endpoint pair associated with \((e,f)\).  Compression compatibility shows
that the analogous matrix on
\(H\) is
\(M_{d,\Omega Z}^{L(H)}=\varphi(M_{d,Z}^{L(G)})\).
Each term of (A.23) belongs to \(\mathcal A_{L(G)}\) by (A.16), so
\(M_{d,Z}^{L(G)}\) is a transported line-algebra matrix.

Fix a profile value
\(\lambda=(\lambda_1,\ldots,\lambda_t)\), and let
\[
 (K_{d,\lambda}^{L(G)})_{ef}
 =(Q_d^{L(G)})_{ef}\prod_{i=1}^t
   \one[(M_{d,Z_i}^{L(G)})_{ef}=\lambda_i].
\]
This is again a union of stable basis relations.  Its transported mate
\[
 K_{d,\lambda}^{L(H)}
 =\varphi(K_{d,\lambda}^{L(G)})
\]
selects exactly the witness edge pairs whose unique distance-\(d\) endpoint
pair has complete profile \(\lambda\).  By the interpolation convention, an unattained requested
coordinate gives the zero selector on that side.

\emph{Multiplicity correction.}
Put
\[
 F_{d,\lambda}^G=B_GK_{d,\lambda}^{L(G)}B_G^\top.
\]
This matrix belongs to \(\mathcal W_G\), and its \(H\)-side counterpart is
its image under \(\Omega\).  The matrix
\(F_{d,\lambda}^G\) is supported on root distances at most \(d\), because
every endpoint pair of a selected witness edge pair has distance at most
\(d\).  If
\(F_{d,\lambda}^G=\sum_i a_iZ_i\), put
\(f_\mu=\sum_i a_i\mu_i\) for \(\mu\in\Lambda\).  The induction hypothesis
constructs and transports the lower-layer matrix
\[
 L_{<d,\lambda}^G
 =\sum_{s<d}\sum_{\mu\in\Lambda}f_\mu P^G_{s,\mu}.
\]
Entrywise this matrix equals
\(\sum_{s<d}(F_{d,\lambda}^G\circ \Gamma_s^G)\); this equality is an observation
about entries, not an appeal to Hadamard closure.  For a root pair at
distance below \(d\), the lower-distance sum removes the corresponding
entry of \(F_{d,\lambda}^G\); above distance \(d\), that entry is zero by the
support observation.  At distance \(d\), every selected witness edge pair
contributing to the \((u,v)\)-entry has \((u,v)\) as its unique
distance-\(d\) endpoint pair.  The entry is then zero unless its profile is
\(\lambda\), and otherwise counts those witness edge pairs.  Since
\(F_{d,\lambda}^G\in\mathcal W_G\), its value on the complete profile
\(\lambda=(\lambda_1,\ldots,\lambda_t)\) is the explicit scalar
\[
 w_{d,\lambda}:=\sum_i a_i\lambda_i.
\]
The transported matrix is \(\sum_i a_i\Omega Z_i\), so it has the same scalar
on the transported \(H\)-profile.  We therefore have the
entrywise identities
\[
 \widehat F_{d,\lambda}^G
 :=F_{d,\lambda}^G-L_{<d,\lambda}^G
 =w_{d,\lambda}P^G_{d,\lambda},
 \qquad
 \Omega(\widehat F_{d,\lambda}^G)
 =w_{d,\lambda}P^H_{d,\lambda}.
 \tag{A.24}
\]
This argument does not use Hadamard closure of \(\mathcal W_G\): each
lower-distance restriction in the definition of \(\widehat F\) was defined
first as the already constructed profile-cell sum displayed above.  If
\(P^G_{d,\lambda}\ne0\), the geodesic construction above gives a selected
witness edge pair, so \(w_{d,\lambda}>0\).  If instead
\(P^H_{d,\lambda}\ne0\), the identical argument on \(H\) gives the same
conclusion for the common scalar.  Dividing (A.24) then constructs both cells
at once:
\[
 P^G_{d,\lambda}
 =\frac1{w_{d,\lambda}}\widehat F_{d,\lambda}^G,
 \qquad
 P^H_{d,\lambda}
 =\frac1{w_{d,\lambda}}\Omega(\widehat F_{d,\lambda}^G).
\]
If the cell occurs on neither root, both cell matrices are zero.  This proves
the induction step for every \(\lambda\in\Lambda\), including empty cells, and
strong induction establishes the lemma for all root distances.
\end{proof}

\subsection{Completion of the proof}

\begin{proof}[Proof of Theorem~\ref{thm:backward-transfer}]
If either root is \(K_1,K_2\), the paw, or the diamond,
Lemma~\ref{lem:bt-direct-roots} proves the result.  We may therefore use the
standing hypotheses of the general construction above.

The nonempty profile cells from
Lemma~\ref{lem:bt-distance-induction} are pairwise disjoint \(0\)-\(1\)
matrices and sum to \(J_{V(G)}\).  Every matrix in \(\mathcal W_G\) is its
profile-value-weighted sum of these cells, while every cell belongs to
\(\mathcal W_G\).  Disjoint nonzero \(0\)-\(1\) cells are linearly
independent, so these two containments show that they form a basis of
\(\mathcal W_G\), and
\[
 P^G_{d,\lambda}\circ P^G_{d',\mu}
 =\one[(d,\lambda)=(d',\mu)]P^G_{d,\lambda}.
\]
Thus \(\mathcal W_G\) is Hadamard closed, and the same holds for
\(\mathcal W_H\).  Since \(\Omega\) maps the profile cells and maps
\(I_{V(G)},J_{V(G)},A_G\) to \(I_{V(H)},J_{V(H)},A_H\), it preserves
Hadamard products as well as the operations in
Lemma~\ref{lem:bt-trace-transport}.

We now restrict \(\Omega\) to the stable root algebras by proving the two
required inclusions separately.  Let
\(\mathcal A_G\) and \(\mathcal A_H\) denote the stable root pair algebras
from the first subsection.  Since \(\mathcal W_G\) is now a coherent matrix
algebra containing \(I_{V(G)},J_{V(G)},A_G\), minimality gives
\(\mathcal A_G\subseteq\mathcal W_G\); similarly
\(\mathcal A_H\subseteq\mathcal W_H\).  Consider the preimage
\[
 \mathcal B_G
 =\{Z\in\mathcal W_G:\Omega(Z)\in\mathcal A_H\}.
\]
Preservation of ordinary and Hadamard multiplication, transpose, and
\(I,J,A\) makes \(\mathcal B_G\) a coherent matrix algebra containing
\(I_{V(G)},J_{V(G)},A_G\).  Minimality therefore gives
\(\mathcal A_G\subseteq\mathcal B_G\), or
\(\Omega(\mathcal A_G)\subseteq\mathcal A_H\).  Applying the same argument
to
\[
 \mathcal B_H
 =\{Y\in\mathcal W_H:\Omega^{-1}(Y)\in\mathcal A_G\}
\]
gives \(\Omega^{-1}(\mathcal A_H)\subseteq\mathcal A_G\), equivalently
\(\mathcal A_H\subseteq\Omega(\mathcal A_G)\).  Hence \(\Omega\) restricts
to an isomorphism
\[
 \Psi:\mathcal A_G\longrightarrow\mathcal A_H
\]
preserving ordinary and Hadamard multiplication, transpose, trace, and
\(I,J,A\).

It remains to verify the basis-bijection hypothesis of
Lemma~\ref{lem:pair-algebra-bridge}.  Let \(\mathcal C_G^{\rm st},
\mathcal C_H^{\rm st}\) be the
stable root-color sets and expand
\[
 \Psi(R_\rho^G)=\sum_{\sigma\in\mathcal C_H^{\rm st}}
   a_{\rho\sigma}R_\sigma^H.
\]
Because \(R_\rho^G\circ R_\rho^G=R_\rho^G\), Hadamard preservation gives
\(a_{\rho\sigma}^2=a_{\rho\sigma}\), so every coefficient is \(0\) or
\(1\).  Images of distinct source relations are Hadamard-orthogonal, and
\[
 \sum_{\rho\in\mathcal C_G^{\rm st}}\Psi(R_\rho^G)
 =\Psi(J_{V(G)})=J_{V(H)}.
\]
Consequently each target relation \(R_\sigma^H\) occurs in exactly one of
the displayed images.  Injectivity of \(\Psi\) makes every source image
nonzero.  Finally, the relation matrices are bases and \(\Psi\) is a linear
isomorphism, so
\[
 |\mathcal C_G^{\rm st}|=\dim\mathcal A_G
 =\dim\mathcal A_H=|\mathcal C_H^{\rm st}|.
\]
Thus the partition of \(\mathcal C_H^{\rm st}\) into the nonempty supports of the
source images has singleton parts.  There is a bijection
\(\pi:\mathcal C_G^{\rm st}\to\mathcal C_H^{\rm st}\) with
\(\Psi(R_\rho^G)=R_{\pi(\rho)}^H\).  All the remaining hypotheses of
Lemma~\ref{lem:pair-algebra-bridge} are inherited from \(\Psi\); in
particular relation sizes agree because
\(|R_\rho^G|=\operatorname{tr}(R_\rho^G(R_\rho^G)^\top)\).
The lemma now gives \(G\equiv_{3\text{-WL}}H\).
\end{proof}

\subsection{Componentwise extension to disconnected roots}
\label{app:componentwise-backward-extension}

We now remove connectedness without changing the connected reconstruction
above.  For a graph \(X\), let \(\operatorname{Comp}(X)\) denote its set of
connected components, including isolated vertices.

\begin{lemma}[Component matching for stable pair refinement]
\label{lem:component-parabolic-matching}
If finite graphs \(X,Y\) satisfy \(X\equiv_{3\text{-WL}}Y\), then their
components admit a bijection
\(\beta:\operatorname{Comp}(X)\to\operatorname{Comp}(Y)\) such that
\[
  X[C]\equiv_{3\text{-WL}}Y[\beta(C)]
  \qquad(C\in\operatorname{Comp}(X)).
\]
Conversely, any such componentwise bijection implies
\(X\equiv_{3\text{-WL}}Y\).
\end{lemma}

\begin{proof}
Assume first that the graphs are nonempty.  Let \(\mathfrak X,\mathfrak Y\)
be their stable-pair coherent configurations and let
\(\varphi:\mathfrak X\to\mathfrak Y\) be the trace-preserving algebraic
isomorphism induced by joint refinement, with
\(\varphi(A_X)=A_Y\).  The least equivalence relation
\[
 e_X=\langle I_X\cup A_X\rangle
\]
has exactly the connected components of \(X\) as its classes.  It is a
parabolic of \(\mathfrak X\), and algebraic isomorphisms transport this
equivalence closure, so \(\varphi(e_X)=e_Y\)
\citep[Section~2.7, equation~(5)]{ponomarenko2023wl}.

Use the unique decomposition of a partial parabolic into indecomposable
partial parabolics.  The isomorphism \(\varphi\) bijects these summands and
preserves their numbers of classes; after pairing their classes, it restricts
on any paired classes \(C,D\) to an algebraic isomorphism of the corresponding
restricted coherent configurations
\citep[Definition~2.1.20, Lemma~2.1.21, Proposition~2.3.25, and
Exercise~2.7.31]{chen2019coherent}.  This restriction maps
\(A_X|_{C\times C}\) to \(A_Y|_{D\times D}\).  It therefore maps the smallest
coherent algebra generated by identity, all-ones, and adjacency on \(X[C]\)
into the corresponding stable pair algebra on \(Y[D]\); applying the inverse
restriction gives the reverse inclusion, hence equality.  The restricted
isomorphism also preserves relation multiplicities.  The coherent-algebra
characterization from
Lemma~\ref{lem:pair-algebra-bridge} gives
\(X[C]\equiv_{3\text{-WL}}Y[D]\).  Repeating this for every class yields
\(\beta\).

For the converse, use the bijective-pebble characterization of counting
logic.  In each round, take the disjoint union of Duplicator's bijections for
the matched component games.  Pebbled vertices in different components stay
in different matched components, so equality and adjacency are preserved.
The empty-graph case is immediate.
\end{proof}

\begin{proof}[Proof of Theorem~\ref{thm:componentwise-backward}]
Apply Lemma~\ref{lem:component-parabolic-matching} to the assumed equivalence
of \(L(G)\) and \(L(H)\).  The line graph of a connected component containing
an edge is connected, and each connected line-graph component arises this way.
Because \(G\) and \(H\) have no isolated vertices, this gives a bijection
between their root components.  For each matched pair \(C,D\),
\[
  L(C)\equiv_{3\text{-WL}}L(D).
\]
Both \(C\) and \(D\) are Whitney-general, so
Theorem~\ref{thm:backward-transfer} gives
\(C\equiv_{3\text{-WL}}D\).  The converse direction of
Lemma~\ref{lem:component-parabolic-matching} now yields
\(G\equiv_{3\text{-WL}}H\).
\end{proof}

\section{Four-Clique Counting and Incomparability Witnesses}
\label{app:finite-certificates}

\subsection{Universal four-clique count}
\label{app:k4-counting-proof}

The initial pair color distinguishes nonadjacent line vertices, which for a
simple root are distinct disjoint edges.  In the first global update of an
ordered pair \((e,f)\), the multiplicity of the replacement entry
\((\mathrm{adj},\mathrm{adj})\) is exactly
\(\lvert N_{L(G)}(e)\cap N_{L(G)}(f)\rvert\); hence the graph-level histogram
determines the set in Corollary~\ref{cor:universal-k4-counting}.

Write \(e=ab\) and \(f=cd\).  Their common line neighbors are precisely the
present edges among \(ac,ad,bc,bd\), so the count is four exactly when the
four endpoints induce a \(K_4\).  Conversely, every \(K_4\) has three
unordered pairs of opposite edges, hence six ordered pairs, and the endpoint
union identifies the clique uniquely.  This proves the formula.  The Lean
development separately checks the local equivalence, the six-to-one fiber
count, the first-round histogram implication, and the global-edge ILG-\(3\)
corollary.

\subsection{Exact witnesses for comparison with root
\texorpdfstring{\(4\)-WL}{4-WL}}
\label{app:four-wl-incomparability}

We verify the two directions in Remark~\ref{rem:four-wl-incomparable} using
exact joint \(3\)-FWL on the roots and exact joint \(2\)-FWL on their line
graphs.  Refinement is collision-safe, uses one shared color namespace for
each compared pair, and runs until separation or stability.  The color traces,
including the atomic round, are
\[
\begin{array}{@{}lcc@{}}
\toprule
\text{witness roots} & \text{root }4\text{-WL} & \text{ILG-}3\text{-WL} \\
\midrule
\operatorname{STS}(13)\text{ incidence pair}
  & 14\to67\to101\;(\text{separates})
  & 3\to5\to7\to7\;(\text{stable equal}) \\
K_5\text{ compressed parity pair}
  & 15\to43\to43\;(\text{stable equal})
  & 3\to8\to46\;(\text{separates}) \\
\bottomrule
\end{array}
\]

For the first pair, translate the base blocks \(\{0,1,4\}\) and
\(\{0,2,7\}\) modulo \(13\), then apply the Pasch trade
\[
014,027,179,249\ \longleftrightarrow\ 017,024,149,279.
\]
The two systems have respectively \(13\) and \(8\) Pasch configurations, so
their connected incidence roots, each with \((n,m)=(39,78)\), are
nonisomorphic.  Figure~\ref{fig:incomparability-sts} shows the two roots under
a shared cyclic layout.

\begin{figure}[H]
  \centering
  \includegraphics[width=.78\linewidth]{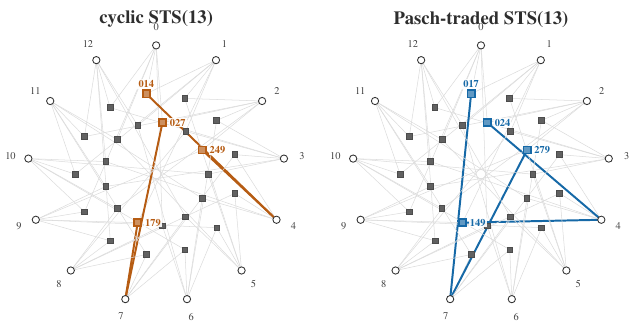}
  \caption{The \(\operatorname{STS}(13)\) incidence witness.  Each panel is a
  root graph: hollow circles are points, squares are blocks, and an edge means
  that the point belongs to the block.  Under the displayed alignment, \(74\)
  light-gray edges are shared; the Pasch trade removes the four orange edges
  and adds the four blue edges.  Root \(4\)-WL separates the pair, whereas
  ILG-\(3\)-WL does not.}
  \label{fig:incomparability-sts}
\end{figure}

For the reverse pair, replace every vertex of \(K_5\) by the
eight parity-constrained bitstrings on its four incident base edges, joining
states in adjacent gadgets when their shared-edge bits agree.  One root uses
even parity in every gadget; the other uses odd parity at one base vertex.
Both are connected \(16\)-regular graphs with \((n,m)=(40,320)\); the first
contains \(64\) copies of \(K_5\) and the second contains none.
Figure~\ref{fig:incomparability-parity} shows these roots under a shared
five-gadget layout.

\begin{figure}[H]
  \centering
  \includegraphics[width=.78\linewidth]{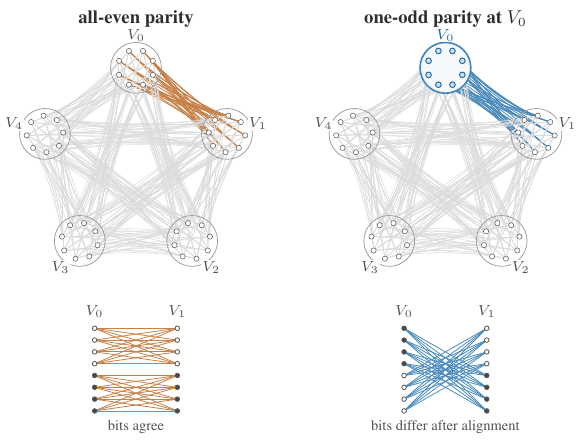}
  \caption{The compressed parity witness over \(K_5\).  Each \(V_i\) contains
  eight bitstring states, shown as graph vertices; states in two gadgets are
  adjacent when their bits for the corresponding base edge agree.  The left
  root uses even parity in every gadget, while the right root uses odd parity
  in \(V_0\).  Under the displayed alignment, \(288\) gray edges are shared;
  the \(32\) orange edges are removed and the \(32\) blue edges are added.
  Each main panel shows all \(40\) vertices and \(320\) edges, and the inset
  enlarges the changed \(V_0\)--\(V_1\) bundle.  ILG-\(3\)-WL separates the
  pair, whereas root \(4\)-WL does not.}
  \label{fig:incomparability-parity}
\end{figure}

All four
roots are therefore nonisomorphic, connected, simple, and Whitney-general.
The reproduction script
\path{experiments/deterministic/run_incomparability_certificate.py}
reconstructs both witness pairs and reproduces these refinement traces.

\section{Machine Verification}
\label{app:machine-verification}

The accompanying Lean~4 development formalizes finite simple undirected
graphs, line graphs, collision-free WL and folklore-WL refinements, the direct
and global-edge ILG definitions, and the message-passing model used in the
neural expressivity theorem.  Top-level import files organize the
formalization into line-graph correspondences, first-order results,
third-order results, and neural expressivity.

The formalized results include the ILG/line-graph refinement and GNN
correspondences, forward transfer, the equal-order first-order counterexample,
the universal factor-six four-clique identity and its global-edge ILG-\(3\)
consequence, the Shrikhande--rook separation, the connected completion with the
Whitney exception, and the connected and componentwise backward-transfer
theorems.  The formalization of backward transfer follows the
reconstruction argument in Appendix~\ref{app:backward-transfer-proof} and
proves the relation between pair refinement and the separate-coordinate
\(3\)-WL formulation used in the paper.  It also formalizes the lower- and
upper-bound directions for global-edge ILG-GNN expressivity.
Theorem~\ref{thm:backward-transfer} is stated in the stable-signature form used
in the paper and proved equivalent to the all-round collision-free formulation
used internally.  The Lean development covers the \(k=3\) arity-shift instance
used by the backward theorems; the general arity shift and the cited \(k=2\)
consequence are taken from the literature.

Lean checks all modules containing the stated theorems.  The development
contains no \texttt{sorry}, \texttt{admit}, or custom axiom declarations.  A
separate Lean file records the axioms used by the main results.  The top-level
imports, build instructions, and theorem index mapping each numbered paper
result and Appendix~A lemma to its Lean declaration are available in the
accompanying repository:
\url{https://github.com/lukeyf/ilg-wl}.

\section{Proofs of ILG Equivalence and Expressivity Results}

\subsection{ILG-\texorpdfstring{\(1\)}{1}-GNN Equals MPNN on \texorpdfstring{\(L(G)\)}{L(G)}}
\label{app:ilg1gnn-linegraph-mpnn}

We give the detailed proof of
Lemma~\ref{lem:ilg1gnn-linegraph-mpnn}. Let \(G\) be a finite simple
undirected graph. Under the canonical identification
\(E(G)=V(L(G))\), define the ILG neighborhood of \(e\in E(G)\) by
\[
    N_{\mathrm{ILG}}(e)
    =
    \{f\in E(G): f\neq e,\ e\cap f\neq\emptyset\}.
\]
By the definition of the line graph,
\[
    N_{\mathrm{ILG}}(e)=N_{L(G)}(e).
    \tag{D.1}
\]

Consider a \(T_{\mathrm{net}}\)-layer MPNN on \(L(G)\) of the form
\[
    h_e^{(\ell+1)}
    =
    U_\ell\left(
        h_e^{(\ell)},
        \operatorname{AGG}_\ell
        \left\{
            M_\ell\left(h_e^{(\ell)},h_f^{(\ell)},
                         \omega^{L(G)}_{ef}\right)
            : f\in N_{L(G)}(e)
        \right\}
    \right),
\]
where \(M_\ell\) is the message map, \(U_\ell\) is the update map,
\(\operatorname{AGG}_\ell\) is permutation invariant, and
\(\omega^{L(G)}_{ef}\) is an optional edge feature on the line graph. In the
untyped case it is constant. If typed line-graph adjacencies are used, assume
the ILG layer uses the corresponding feature
\[
    \omega^{\mathrm{ILG}}_{ef}=\omega^{L(G)}_{ef}
    \quad
    \text{for all } f\in N_{\mathrm{ILG}}(e).
    \tag{D.2}
\]
The matching ILG-\(1\)-GNN on \(G\) is
\[
    \tilde h_e^{(\ell+1)}
    =
    U_\ell\left(
        \tilde h_e^{(\ell)},
        \operatorname{AGG}_\ell
        \left\{
            M_\ell\left(\tilde h_e^{(\ell)},\tilde h_f^{(\ell)},\omega^{\mathrm{ILG}}_{ef}\right)
            : f\in N_{\mathrm{ILG}}(e)
        \right\}
    \right).
\]
Both models are initialized by the same edge features:
\[
    h_e^{(0)}=\tilde h_e^{(0)}=\phi_0(q_e)
    \quad
    \text{for all } e\in E(G).
\]

We prove by induction that
\[
    h_e^{(\ell)}=\tilde h_e^{(\ell)}
    \quad
    \text{for all } e\in E(G)
    \tag{D.3}
\]
at every layer \(\ell\). The base case \(\ell=0\) is the shared initialization above. Suppose (D.3) holds at layer \(\ell\). For a fixed edge \(e\), the MPNN message multiset at the line-graph vertex \(e\) is
\[
    \mathcal{M}^{L(G)}_e
    =
    \left\{\!\left\{
        M_\ell\left(
            h_e^{(\ell)},
            h_f^{(\ell)},
            \omega^{L(G)}_{ef}
        \right)
        : f\in N_{L(G)}(e)
    \right\}\!\right\}.
\]
Using the neighborhood identity (D.1), the feature identity (D.2), and the induction hypothesis, this multiset is equal to
\[
    \mathcal{M}^{\mathrm{ILG}}_e
    =
    \left\{\!\left\{
        M_\ell\left(
            \tilde h_e^{(\ell)},
            \tilde h_f^{(\ell)},
            \omega^{\mathrm{ILG}}_{ef}
        \right)
        : f\in N_{\mathrm{ILG}}(e)
    \right\}\!\right\}.
\]
Because \(\operatorname{AGG}_\ell\) is a function of the multiset rather than the ordering of its elements,
\[
    \operatorname{AGG}_\ell(\mathcal{M}^{L(G)}_e)
    =
    \operatorname{AGG}_\ell(\mathcal{M}^{\mathrm{ILG}}_e).
\]
Applying the same update map \(U_\ell\) gives
\[
    h_e^{(\ell+1)}
    =
    \tilde h_e^{(\ell+1)}.
\]
Thus (D.3) holds at layer \(\ell+1\), and induction proves the layerwise equivalence.

If the implementation stores line-graph vertices in an ordering different
from the edge ordering used by the ILG layer, let \(\mathsf P_E\) be the
corresponding permutation matrix. The same argument gives
\[
    \mathsf Z_{L(G)}^{(\ell)}=\mathsf P_E\widetilde{\mathsf Z}^{(\ell)}
    \quad
    \text{for all }\ell.
\]
Any graph readout that is invariant to the order of vertices or edges
therefore satisfies
\[
    \operatorname{READOUT}\left(\{h_e^{(T_{\mathrm{net}})}:
                                   e\in V(L(G))\}\right)
    =
    \operatorname{READOUT}\left(\{\tilde h_e^{(T_{\mathrm{net}})}:
                                   e\in E(G)\}\right).
\]
The equivalence is therefore exact for edge-only ILG-\(1\)-GNN.

\subsection{Atomic Types on the Line Graph}
\label{app:atomic-types}

Let \(G\) be simple and undirected. A vertex of \(L(G)\) is a root edge
\(e\in E(G)\). For ordered edge tuples
\(\mathbf{e}=(e_1,\ldots,e_k)\in E(G)^k\), the \emph{atomic type} of
\(\mathbf{e}\) as a tuple of vertices in \(L(G)\) is determined by equality
and adjacency among tuple coordinates:
\[
    \operatorname{atp}_{L(G)}(\mathbf{e})
    =
    \left(
        \one[e_i=e_j],
        \one[e_i\sim_{L(G)} e_j]
    \right)_{1\leq i,j\leq k}.
\]
Because \(e_i\sim_{L(G)} e_j\) if and only if \(e_i\neq e_j\) and
\(e_i\cap e_j\neq\emptyset\), this atomic type can be computed from the
root-edge endpoints without building \(L(G)\). For \(k=2\), equality and
adjacency collapse to the three-valued relation code
\[
    r_G(e,f)=0 \text{ if } e=f,\qquad
    r_G(e,f)=1 \text{ if } e\sim_{L(G)} f,\qquad
    r_G(e,f)=2 \text{ otherwise}.
\]
Thus every initial color used by \(k\)-WL on \(L(G)\) is an endpoint-incidence function on \(G\).

\begin{remark}[Completeness of the relation code for \(k=2\)]
\label{rem:relation-complete}
The three-valued code \(r_G(e,f)\) determines the atomic type of the ordered
pair \((e,f)\) in \(L(G)\). The only
line-graph relations between two vertices are equality, adjacency, and
non-adjacency. Since line-graph adjacency is symmetric for simple undirected
\(G\), the code is symmetric: \(r_G(e,f)=r_G(f,e)\). The initial atomic type
for an ordered pair is therefore
\((r_G(e,f),r_G(f,e))=(r_G(e,f),r_G(e,f))\), which is determined by
\(r_G(e,f)\) alone.
\end{remark}

\subsection{Forward Transfer: \texorpdfstring{\(1\)}{1}-WL Equivalence Implies Line-Graph \texorpdfstring{\(1\)}{1}-WL Equivalence}
\label{app:forward-transfer}

We prove Proposition~\ref{prop:forward-transfer}.

Suppose \(G\equiv_{1\text{-WL}} H\). We argue directly with the stable
\(1\)-WL (color-refinement) colorings; this is equivalent to the
fractional-isomorphism characterization
\citep{ramana1994fractional,grohe2017descriptive} and avoids constructing
an explicit doubly stochastic matrix on edges.

Let \(c_G\) and \(c_H\) be the stable \(1\)-WL vertex colorings of \(G\)
and \(H\). For stable colors \(a,b\), let \(q_G(a,b)\) be the number of
neighbors of color \(b\) at a vertex of color \(a\); define \(q_H\)
analogously. Define an induced coloring of the line graph by
\[
    \widehat c_G(e) = \multiset{c_G(u), c_G(v)}
    \qquad\text{for } e=\{u,v\}\in E(G),
\]
and likewise \(\widehat c_H\) on \(L(H)\).

\emph{\(\widehat c_G\) is equitable on \(L(G)\).} Fix
\(e=\{u,v\}\). Its neighbors in \(L(G)\) are the root edges sharing \(u\)
or \(v\), other than \(e\). For any target color \(\multiset{a,b}\), the
number of such edges \(f\) with
\(\widehat c_G(f)=\multiset{a,b}\) is determined by
\(c_G(u),c_G(v)\) and the quotient numbers \(q_G\), with the contribution of
\(e\) itself removed. It is therefore constant on each \(\widehat c_G\)-class.
Thus \(\widehat c_G\) is equitable; the same proof applies to
\(\widehat c_H\). We do not require these induced colorings to be the
coarsest equitable partitions.

\emph{Transfer.} Because \(G\equiv_{1\text{-WL}}H\), a stable-color
bijection identifies the class sizes and quotient numbers \(q_G,q_H\). It
therefore maps \(\widehat c_G\)-classes to \(\widehat c_H\)-classes while
preserving their sizes and their equitable quotients. Starting color
refinement from these matched equitable partitions gives identical stable
histograms; coarsening back to the uniform initialization cannot separate
the graphs either. Therefore
\(L(G)\equiv_{1\text{-WL}}L(H)\). \qed

\subsection{Equal-Order Failure of First-Order Backward Transfer}
\label{app:first-order-transfer}

For the roots in Proposition~\ref{prop:first-order-loss}, the different root degree histograms \(4^2 2^6\) and \(3^6 1^2\) prove \(1\)-WL inequivalence. To verify the line premise without relying on software, start line-graph color refinement from degrees. In both line graphs, two vertices have degree two and eight have degree four. The next round splits the degree-four vertices according to whether they have a degree-two neighbor, producing cells \(A,B,C\) of sizes \(2,4,4\). Their common quotient is
\[
 A:(0,2,0),\qquad B:(1,1,2),\qquad C:(0,2,2).
\]
This partition is equitable and therefore stable. Thus the two line graphs have identical stable color-class sizes and quotients, while the roots do not. This also explains why an argument based only on root order or inverse-line-graph uniqueness cannot prove first-order backward transfer.

\subsection{Equivalence of ILG-\texorpdfstring{\(k\)}{k}-WL and \texorpdfstring{\(k\)}{k}-WL on \texorpdfstring{\(L(G)\)}{L(G)}}
\label{app:ilg-kwl-proof}

We prove Proposition~\ref{prop:ilg-kwl}.

For \(k=1\), the ILG update uses exactly the line-neighbor set identified in
the proof of Lemma~\ref{lem:ilg1gnn-linegraph-mpnn}; hence its color
refinement is exactly \(1\)-WL on \(L(G)\). Assume for the remainder of
this subsection that \(k\geq2\).

Let
\(c_t^{L(G)},c_t^{\mathrm{ILG},G}:E(G)^k\to\mathcal C_t\) denote the two
colorings, run with the shared encoding convention of
Section~\ref{sec:background}. Both start from the same atomic type:
\[
 c_0^{L(G)}(\mathbf e)
 =c_0^{\mathrm{ILG},G}(\mathbf e)
 =\hash(\operatorname{atp}_{L(G)}(\mathbf e)).
\]

For \(\mathbf{e}=(e_1,\ldots,e_k)\) and \(g\in E(G)\), write
\[
    \mathbf{e}[i\leftarrow g]
    =
    (e_1,\ldots,e_{i-1},g,e_{i+1},\ldots,e_k).
\]
The separate-coordinate \(k\)-WL update on \(L(G)\) is
\[
    c_{t+1}^{L(G)}(\mathbf{e})
    =
    \hash\left(
        c_t^{L(G)}(\mathbf{e}),
        \left(
            \multiset{
                c_t^{L(G)}(\mathbf{e}[i\leftarrow g]) : g\in E(G)
            }
        \right)_{i=1}^{k}
    \right).
\]
The ILG-\(k\)-WL update is
\[
    c_{t+1}^{\mathrm{ILG},G}(\mathbf{e})
    =
    \hash\left(
        c_t^{\mathrm{ILG},G}(\mathbf{e}),
        \left(
            \multiset{
                c_t^{\mathrm{ILG},G}(\mathbf{e}[i\leftarrow g]) : g\in E(G)
            }
        \right)_{i=1}^{k}
    \right).
\]

The two displayed recurrences have the same domain, initialization,
replacement range, and update. Induction on \(t\) therefore gives
\[
 c_t^{L(G)}(\mathbf e)=c_t^{\mathrm{ILG},G}(\mathbf e)
 \qquad(\mathbf e\in E(G)^k)
\]
at every round. If two implementations use different injective color names,
the same induction instead supplies the canonical bijection between those
names.

\textbf{Graph-level consequence.} Since the tuple domains are identical and
the colors agree at every round, their graph-level signatures agree:
\[
    \multiset{c_t^{\mathrm{ILG},G}(\mathbf e):\mathbf e\in E(G)^k}
    =\multiset{c_t^{L(G)}(\mathbf e):\mathbf e\in E(G)^k}
    \quad\text{for all }t.
\]
Therefore any graph pair separated by \(k\)-WL on \(L(G)\) is separated by ILG-\(k\)-WL on \(G\), and conversely. \qed

\subsection{Equivalence of Global-Edge ILG-\texorpdfstring{\(k\)}{k}-WL and \texorpdfstring{\(k\)}{k}-WL on \texorpdfstring{\(L(G)\)}{L(G)}}
\label{app:global-edge-2fwl}

We prove Corollary~\ref{cor:global-edge-fwl} for \(k\geq3\), writing the update
explicitly for \(k=3\).

\paragraph{From direct ILG-\(k\)-WL to the global-replacement form.}
The direct ILG-\(k\)-WL state space is \(E(G)^k\). Its
separate-coordinate update contains \(k\) replacement multisets and enumerates
\(O(m^{k+1})\) replacement entries per round; canonical sorting gives the
\(O(m^{k+1}\log m)\) worst-case time in Table~\ref{tab:complexity}. The
global-edge construction instead uses the equivalent global-replacement form
on \((k{-}1)\)-tuples. At \(k=3\), the
pair coloring of \(L(G)\) has the update
\[
 \vartheta_t^{L(G)}(g;e,f)
 :=\bigl(c_t^{\mathrm{FWL},L(G)}(g,f),
          c_t^{\mathrm{FWL},L(G)}(e,g)\bigr),
\]
\[
 c_{t+1}^{\mathrm{FWL},L(G)}(e,f)
 =\hash\left(
   c_t^{\mathrm{FWL},L(G)}(e,f),
   \multiset{\vartheta_t^{L(G)}(g;e,f):g\in E(G)}
 \right).
\]
This pair-state process has the same stable graph-distinguishing power as the
separate-coordinate \(3\)-WL convention fixed in
Section~\ref{sec:background}
\citep{feng2023fwl}.

\paragraph{Restriction to \(k\geq 3\).}
For \(k=2\), global replacement has one-edge states.  Its message
\(\multiset{c_t(g):g\in E(G)}\) contains no relation between \(g\) and the
current edge, so it cannot recover line-graph adjacency; for example, it
fails on \(P_4\) versus \(K_{1,3}\).  Adding \(r_G(g,e)\) restores exactly
\(1\)-WL power.  The \(k=2\) comparison therefore follows from the
graph-level equivalence of \(1\)-WL and \(2\)-WL.

\paragraph{Specialization to the line graph.}
All atomic information needed to initialize this pair coloring is given by
\(r_G(\cdot,\cdot)\) (Remark~\ref{rem:relation-complete}). Define the
global-edge coloring \(c_t^{\mathrm{GE},G}\) by
\[
 c_0^{\mathrm{GE},G}(e,f)=\hash(r_G(e,f)),
\]
and, with
\[
 \vartheta_t^{\mathrm{GE},G}(g;e,f)
 :=\bigl(c_t^{\mathrm{GE},G}(g,f),c_t^{\mathrm{GE},G}(e,g),
          r_G(g,e),r_G(g,f)\bigr),
\]
\[
 c_{t+1}^{\mathrm{GE},G}(e,f)
 =\hash\left(
   c_t^{\mathrm{GE},G}(e,f),
   \multiset{\vartheta_t^{\mathrm{GE},G}(g;e,f):g\in E(G)}
 \right).
\]

\paragraph{Redundancy of the extra relation codes.}
At every round \(t\geq0\), the color \(c_t^{\mathrm{GE},G}(e,g)\)
refines its initial color and therefore determines \(r_G(g,e)\).
Symmetrically, \(c_t^{\mathrm{GE},G}(g,f)\) determines \(r_G(g,f)\).
Thus the two appended relation codes are functions of the first two entry
coordinates and cannot refine the multiset partition.

\paragraph{Partition equivalence by induction.}
At \(t=0\), both colorings are injective functions of \(r_G(e,f)\), so
they induce the same partition. Assume a color-name bijection \(\pi_t\)
identifies their partitions at round \(t\). The global-replacement multiset is
\[
 \mathcal S_t^{\mathrm{FWL},L(G)}(e,f)
 =\multiset{\vartheta_t^{L(G)}(g;e,f):g\in E(G)}.
\]
The global-edge multiset is
\[
 \mathcal S_t^{\mathrm{GE},G}(e,f)
 =\multiset{\vartheta_t^{\mathrm{GE},G}(g;e,f):g\in E(G)}.
\]
Apply \(\pi_t\) to the first two coordinates. The preceding paragraph shows
that the remaining two coordinates are fixed functions of them, so the two
multisets determine one another. Their full update signatures therefore
induce the same partition at round \(t+1\), giving \(\pi_{t+1}\). \qed

\subsection{Neural Simulation of the Tuple Refinement}
\label{app:neural-simulation}

We prove Proposition~\ref{prop:expressivity} for \(k=3\), where the states are
edge pairs.  The general construction replaces pairs by ordered
\((k{-}1)\)-tuples.

\paragraph{Setup.} Fix a feature-labeled graph class \(\mathcal{G}\) with
\(|E(G)|\leq m_{\max}\) for all \(G\in\mathcal{G}\), and assume that its
edge features range over a fixed finite alphabet. At \(k=3\), the
global-edge ILG-\(k\)-GNN layer is
\[
    M_{ef}^{(\ell)}
    =
    \sum_{g\in E(G)}
    \psi_\ell\left(
        h_{gf}^{(\ell)},h_{eg}^{(\ell)},q_g,
        \relEmb(r_G(g,e)),\relEmb(r_G(g,f))
    \right),
\]
\[
    h_{ef}^{(\ell+1)}
    =
    \eta_\ell(h_{ef}^{(\ell)},M_{ef}^{(\ell)}).
\]

\paragraph{Simulation direction (lower bound).}
We construct functions \(\psi_\ell,\eta_\ell\) such that
\(h_{ef}^{(\ell)}\) encodes the feature-labeled discrete color
\(c_\ell^{\mathrm{GE},G}(e,f)\).

\emph{Base case.} Choose \(\relEmb\) injective and choose \(\phi_0\) to be
injective on the finite domain
\(\{(q_e,q_f,\relEmb(r_G(e,f))):e,f\in E(G),\ G\in\mathcal G\}\). For the
feature-labeled refinement, set
\[
 c_0^{\mathrm{GE},G}(e,f)=\hash(q_e,q_f,r_G(e,f)).
\]
Then \(h_{ef}^{(0)}=\phi_0(q_e,q_f,\relEmb(r_G(e,f)))\) injectively encodes
this initial color.

\emph{Induction step.} Suppose
\(h_{ef}^{(\ell)}=\chi_\ell(c_\ell^{\mathrm{GE},G}(e,f))\) for an injective
code \(\chi_\ell:\mathcal D_\ell\to\mathbb R^{d_h}\). Define the
\emph{replacement entry}
\[
 \tau_\ell^G(g;e,f)
 =\left(c_\ell^{\mathrm{GE},G}(g,f),
         c_\ell^{\mathrm{GE},G}(e,g),
         r_G(g,e),r_G(g,f),q_g\right).
\]
The entry set is finite, with cardinality at most
\[
 |\mathcal D_\ell|^2\,3^2\,|\{q_g\}|.
\]
Consequently \(\multiset{\tau_\ell^G(g;e,f):g\in E(G)}\) ranges over a
finite family of multisets of size at most \(m_{\max}\).

By the multiset-encoding lemma \citep{xu2019gin}, for any countable domain and bounded multiset size, there exists a function \(\omega_\ell\) and a dimension \(d\) such that
\[
    \Phi_\ell\bigl(\multiset{\tau_\ell^G(g;e,f):g\in E(G)}\bigr)
    \;:=\;
    \sum_{g\in E(G)}\omega_\ell(\tau_\ell^G(g;e,f))
    \;\in\;\mathbb{R}^d
\]
is injective in the multiset argument. Since \(\chi_\ell\) and \(\relEmb\)
are injective, the neural message inputs determine \(\tau_\ell^G\). Choose
\(\psi_\ell\) to realize \(\omega_\ell\) on this finite input domain; an MLP
of sufficient width can do so. Then
\[
    M_{ef}^{(\ell)}=\sum_{g\in E(G)}\psi_\ell(\cdots)
    =\Phi_\ell\bigl(\multiset{\tau_\ell^G(g;e,f):g\in E(G)}\bigr),
\]
which is injective in the replacement multiset. Choose \(\eta_\ell\) to
map \((h_{ef}^{(\ell)},M_{ef}^{(\ell)})\) injectively to a new code for
\(c_{\ell+1}^{\mathrm{GE},G}(e,f)\). The discrete update shows that these two
inputs determine the new color.

At the final round, pair colors and edge features range over finite domains.
Since the final hidden states injectively encode pair colors and both domains
and their multiset sizes are finite, choose \(\gamma_{\mathrm{read}}\) and
\(\xi\) so that their sums injectively encode the respective multisets.
Choose \(\operatorname{READOUT}\) injective on the resulting finite set of
summary pairs.  The graph output then injectively encodes the complete
feature-labeled refinement signature.

\paragraph{Upper bound (converse).}
Run the feature-labeled refinement jointly on \(G\) and \(H\), so equal
color names have a graph-independent meaning. We prove by induction that,
for fixed network parameters, there is one function \(F_\ell\) such that,
for \(K\in\{G,H\}\),
\[
 h_{ef}^{(\ell)}=F_\ell(c_\ell^{\mathrm{GE},K}(e,f))
 \qquad(e,f\in E(K)).
\]
The base case follows from the initial color above. At the induction step,
the new color determines the old color and the complete multiset of
replacement entries \(\tau_\ell^K(g;e,f)\). The induction hypothesis therefore
determines every message input, its multiset, the aggregate, and the updated
state. This defines \(F_{\ell+1}\) independently of \(K\).

If the two graphs are not separated by the feature-labeled refinement, the multiplicity of every final pair color is the same on both sides; the same is true for every edge-feature color represented on the diagonal. Hence the graph readout
\[
 h_K=\operatorname{READOUT}\left(
   \sum_{e,f\in E(K)}\gamma_{\mathrm{read}}
      (h_{ef}^{(T_{\mathrm{net}})}),
   \sum_{e\in E(K)}\xi(q_e)\right),
 \qquad K\in\{G,H\},
\]
receives identical multisets of inputs and gives \(h_G=h_H\). No tuplewise
or edgewise bijection between the graphs is used. \qed

\subsection{Algorithmic Specification}
\label{app:algo-spec}

This section specifies the global-edge ILG-\(k\)-GNN forward pass for general
\(k\).  The procedure operates directly on \(G\) without constructing the
adjacency list of \(L(G)\).  Fix an edge ordering
\(\mathbf E_G=(e_1,\ldots,e_m)\) of \(E(G)\).

\paragraph{Step 1: edge features.}
Construct an edge feature matrix \(\mathsf Q_G\in\mathbb{R}^{m\times d_e}\). For unlabeled simple graphs, \((\mathsf Q_G)_i\) may be a constant vector. With vertex labels, \((\mathsf Q_G)_i\) concatenates a constant edge indicator with the one-hot labels of the two endpoints in sorted order.

\paragraph{Step 2: relation access.}
After locally indexing the incident root vertices, store the endpoint table
\(\mathsf P_G\in\mathbb N^{m\times2}\), whose \(i\)-th row lists the two
endpoint indices of \(e_i\). A relation code is computed in \(O(1)\) by
comparing two rows: \(r_G(e_i,e_j)=0\) when \(i=j\), \(1\) when the rows share
an entry, and \(2\) otherwise. Constructing and retaining \(\mathsf P_G\)
uses \(O(m)\) time and space. The layer evaluates \(r_G\) only for the index
pairs of the current tuple chunk rather than retaining a relation matrix.
In Steps~3--4, \((\mathsf R_G)_{ij}\) denotes a call to this accessor, not a
stored tensor. Implementations that require a dense relation tensor may
instead materialize \(\mathsf R_G\in\{0,1,2\}^{m\times m}\), using \(O(m^2)\)
space. This tensor is not required by the model definition or by an
implementation that computes relations on demand.

\paragraph{Step 3: tuple initialization.}
For every ordered \((k{-}1)\)-tuple \(\mathbf{i}=(i_1,\ldots,i_{k-1})\), compute
\[
    \mathsf Z_{\mathbf{i}}^{(0)}
    =
    \phi_0\left((\mathsf Q_G)_{i_1},\ldots,(\mathsf Q_G)_{i_{k-1}},\;
    \bigl((\mathsf R_G)_{i_p,i_q}\bigr)_{1\leq p<q\leq k-1}\right).
\]
This yields \(\mathsf Z^{(0)}\in\mathbb{R}^{m^{k-1}\times d_h}\).

\paragraph{Step 4: global-edge message aggregation.}
For each layer \(\ell\) and each tuple \(\mathbf{i}\), aggregate over the global edge index \(g\):
\[
    \mathsf M_{\mathbf{i}}^{(\ell)}
    =
    \sum_{g=1}^{m}
    \psi_\ell\left(
        \bigl(\mathsf Z_{\mathbf{i}[p\leftarrow g]}^{(\ell)}\bigr)_{p=1}^{k-1},\;
        (\mathsf Q_G)_g,\;
        \bigl((\mathsf R_G)_{g,i_p}\bigr)_{p=1}^{k-1}
    \right),
\]
where \(\mathbf{i}[p\leftarrow g]\) denotes the tuple with coordinate~\(p\) replaced by~\(g\). Then update
\(\mathsf Z_{\mathbf{i}}^{(\ell+1)}=
\eta_\ell(\mathsf Z_{\mathbf{i}}^{(\ell)},\mathsf M_{\mathbf{i}}^{(\ell)})\).

\paragraph{Step 5: invariant readout.}
After \(T_{\mathrm{net}}\) layers, produce graph logits by applying an
invariant readout, where \(\|\) denotes vector concatenation:
\[
    z_G =
    \mathrm{MLP}\left(
        \left[
        \sum_{\mathbf{i}}\gamma_{\mathrm{read}}(\mathsf Z_{\mathbf{i}}^{(T_{\mathrm{net}})})
        \;\middle\|\;
        \sum_i \xi((\mathsf Q_G)_i)
        \right]
    \right).
\]
Any fixed edge and tuple ordering may be used.  Under a vertex relabeling,
\(\mathsf R_G,\mathsf Q_G,\mathsf Z\) are permuted consistently, and the final
sums are unchanged.

\subsection{Permutation Equivariance and Invariance}
\label{app:equivariance}

We prove equivariance and invariance for the \(k=3\) instantiation (edge
pairs); the general \((k{-}1)\)-tuple case is identical. Let
\(\sigma:V(G)\to V(H)\) be a relabeling isomorphism and define its induced
edge bijection by
\[
 \sigma_E(\{u,v\})=\{\sigma(u),\sigma(v)\}.
\]
For \(e\in E(G)\), write \(\bar e=\sigma_E(e)\). Primes denote features
and hidden states on \(H=\sigma(G)\).

\paragraph{Relation preservation.}
For any \(e,f\in E(G)\), the edges share an endpoint if and only if their
images do. Therefore
\[
 r_H(\bar e,\bar f)=r_G(e,f).
\]

\paragraph{Initialization.}
If edge features satisfy \(q_{\bar e}'=q_e\) (i.e., features are defined
intrinsically), then the initialized pair tensor satisfies
\[
\begin{aligned}
 h_{\bar e,\bar f}^{\prime(0)}
 &=\phi_0(q_{\bar e}',q_{\bar f}',\relEmb(r_H(\bar e,\bar f)))\\
 &=\phi_0(q_e,q_f,\relEmb(r_G(e,f)))
 =h_{ef}^{(0)}.
\end{aligned}
\]

\paragraph{Induction.}
Assume \(h_{\bar e,\bar f}^{\prime(\ell)}=h_{ef}^{(\ell)}\) for
all \(e,f\in E(G)\). The message for the relabeled pair is
\[
\begin{aligned}
    M_{\bar e,\bar f}^{\prime(\ell)}
    &=
    \sum_{g'\in E(H)}
    \psi_\ell\left(
        h_{g',\bar f}^{\prime(\ell)},
        h_{\bar e,g'}^{\prime(\ell)},
        q_{g'}',
        \relEmb(r_H(g',\bar e)),
        \relEmb(r_H(g',\bar f))
    \right).
\end{aligned}
\]
Substituting \(g'=\sigma_E(g)\), using the bijection
\(E(G)\to E(H)\), and applying the induction hypothesis, feature identity,
and relation preservation gives
\[
    M_{\bar e,\bar f}^{\prime(\ell)}
    =
    \sum_{g\in E(G)}
    \psi_\ell\left(
        h_{gf}^{(\ell)},
        h_{eg}^{(\ell)},
        q_g,
        \relEmb(r_G(g,e)),
        \relEmb(r_G(g,f))
    \right)
    =
    M_{ef}^{(\ell)}.
\]
Thus
\(h_{\bar e,\bar f}^{\prime(\ell+1)}
=\eta_\ell(h_{ef}^{(\ell)},M_{ef}^{(\ell)})=h_{ef}^{(\ell+1)}\).

\paragraph{Invariance.}
The graph readout sums over all ordered edge pairs and all edges. Since
\(\sigma_E:E(G)\to E(H)\) is a bijection, the multisets
\[
 \multiset{h_{\bar e,\bar f}^{\prime(T_{\mathrm{net}})}:
             e,f\in E(G)}
 \quad\text{and}\quad
 \multiset{h_{ef}^{(T_{\mathrm{net}})}:e,f\in E(G)}
\]
are identical, and likewise for edge features. Therefore the invariant
readouts are equal: \(h_H=h_G\). \qed

\subsection{Complexity Details}
\label{app:complexity}

Let \(n=|V(G)|\), \(m=|E(G)|\),
\(m_L=|E(L(G))|=\sum_{v\in V(G)}\binom{d_G(v)}{2}\), and fix
\(k\geq 3\).  Table~\ref{tab:complexity} reports worst-case upper bounds for
fixed \(k\) and fixed neural feature widths.  Discrete update signatures are
streamed one target tuple at a time.  For the neural implementation, \(C\) is
the target-tuple chunk size.

\begin{table}[H]
    \centering
    \caption{Worst-case complexity of the root- and line-domain formulations.
    Initialization is performed once; iterative costs are per refinement
    round or neural layer.}
    \label{tab:complexity}
    \footnotesize
    \setlength{\tabcolsep}{4pt}
    \begin{tabular}{@{}
        >{\raggedright\arraybackslash}p{0.28\textwidth}
        >{\raggedright\arraybackslash}p{0.34\textwidth}
        >{\raggedright\arraybackslash}p{0.28\textwidth}@{}}
        \toprule
        Method & Time & Peak space \\
        \midrule
        Root \(G\) \((k{-}1)\)-FWL
        & \(O(n^{k-1}+n+m)\) initialization;
          \(O(n^k\log n)\) per round
        & \(O(n^{k-1}+n+m)\) \\
        \addlinespace[3pt]
        Direct ILG-\(k\)-WL
        & \(O(m^k)\) initialization;
          \(O(m^{k+1}\log m)\) per round
        & \(O(m^k+m)\) \\
        \addlinespace[3pt]
        Explicit \(L(G)\) \((k{-}1)\)-FWL
        & \(O(m^{k-1}+m+m_L)\) initialization;
          \(O(m^k\log m)\) per round
        & \(O(m^{k-1}+m+m_L)\) \\
        \addlinespace[3pt]
        Global-edge ILG-\(k\)-WL
        & \(O(m^{k-1}+m)\) initialization;
          \(O(m^k\log m)\) per round
        & \(O(m^{k-1}+m)\) \\
        \addlinespace[3pt]
        Global-edge ILG-\(k\)-GNN
        & \(O(m^{k-1}+m)\) initialization;
          \(O(m^k)\) per layer
        & \(O(m^{k-1}+Cm)\) per layer with checkpointing \\
        \bottomrule
    \end{tabular}
\end{table}

Direct ILG-\(k\)-WL stores \(k\)-tuples, whereas global-edge ILG-\(k\)-WL stores
\((k{-}1)\)-tuples; this accounts for the additional factor of \(m\) in the
direct formulation.  Explicit \((k{-}1)\)-FWL on \(L(G)\) and global-edge
ILG-\(k\)-WL have the same per-round complexity.  They differ in whether
line-graph adjacency is stored or computed from root-edge endpoints.  Since
\(m_L\leq\binom{m}{2}=O(m^2)\), both use \(O(m^2)\) total worst-case space at
\(k=3\).  The \(O(m)\) versus \(O(m+m_L)\) difference refers only to graph
storage, excluding tuple states.  If exact color identifiers occupy a dense
integer range of suitable size, array-based counting removes the sorting factor.  \(T\) rounds or
layers multiply iterative time by \(T\).  The reported ILG-\(3\)-GNN experiment
and the empirical cost experiment both execute the dense ordered edge-pair
state space, so their implementation does not alter these bounds.

The deterministic refinement runner represents each canonical signature by a
deterministic \(128\)-bit digest.  A digest collision is therefore possible but
highly unlikely.

\section{Empirical Cost and Scaling}
\label{app:empirical-complexity}

\paragraph{Objectives.}
At \(k=3\), we test four claims: line-domain cost scales with \(m\); computing
adjacency from root endpoints avoids storing \(E(L(G))\) but leaves the
\(m^2\) states and \(m^3\) updates unchanged; root \(4\)-WL can be cheaper on
dense roots; and explicit and implicit implementations return the same
signatures.

\paragraph{Experimental setup.}
We compare root \(2\)-FWL (the paper's root \(3\)-WL), root \(3\)-FWL (root
\(4\)-WL), \(2\)-FWL on an explicitly materialized \(L(G)\), and global-edge
ILG-\(3\)-WL.  The explicit and implicit methods use the same deterministic
pair-refinement implementation with \(128\)-bit digests.  The explicit method constructs and retains
\(E(L(G))\), including isolated line vertices, whereas the implicit method
queries the atomic relation from the root-edge endpoints.  We checked their
final signatures for equality on every measured call; all checks passed.
All methods use dense state spaces, one graph per call, and complete refinement
rounds.

The primary sweep uses connected \(d\)-regular roots on \(n=25\) vertices for
\(d\in\{2,4,6,8,12,16,20,24\}\), with three fixed seeds per degree (the
endpoints are \(C_{25}\) and \(K_{25}\)).  This fixes every root-domain state
and replacement count while varying \(m=25d/2\).  Each call performs atomic
initialization and exactly one complete refinement round, avoiding
graph-dependent stabilization depth.  A fresh single-threaded process first
performs one warm-up and then three timed calls; Table~\ref{tab:empirical-density}
reports the median over the resulting nine observations.  Peak resident set
size (RSS) is measured without allocation tracing in a separate fresh process
on the fixed first seed.  Root generation and interpreter start-up are
excluded from time, while explicit line-graph construction is included.

The reference run used an Apple M2 CPU, macOS 26.5.1, Python 3.12.4, and
NetworkX 3.6.1.  All algorithmic calls are single-process with library thread
counts fixed to one.  The benchmark runner records every repetition, seed,
IQR, pre-call RSS, graph count, software version, and exact output digest.
As with the other deterministic experiments, colors use \(128\)-bit digests,
so equality has negligible rather than mathematically zero collision risk.

\begin{table}[H]
    \centering
    \caption{Fixed-\(n=25\) one-round cost sweep.  Times are median seconds;
    four representative degrees are shown, while the exponent fit uses all
    six points with \(m\geq75\).  RSS columns report absolute fresh-process
    peaks in MiB for one fixed graph.  At the displayed degrees,
    \((m,m_L)\) equals \((25,25)\), \((75,375)\), \((150,1650)\), and
    \((300,6900)\), respectively.  Explicit and implicit signatures agreed
    throughout.}
    \label{tab:empirical-density}
    \scriptsize
    \setlength{\tabcolsep}{2.7pt}
    \begin{tabular*}{\textwidth}{@{\extracolsep{\fill}}lccrrrrrr@{}}
        \toprule
        & & & \multicolumn{4}{c}{End-to-end one-round time [s]}
        & \multicolumn{2}{c}{Peak RSS [MiB]} \\
        \cmidrule(lr){4-7}\cmidrule(lr){8-9}
        Method & States & Entries/round
        & \(d=2\) & \(d=6\) & \(d=12\) & \(d=24\)
        & \(d=12\) & \(d=24\) \\
        \midrule
        Root \(3\)-WL
        & \(n^2\) & \(n^3\)
        & 0.0138 & 0.0152 & 0.0160 & 0.0117 & 52.6 & 52.1 \\
        Root \(4\)-WL
        & \(n^3\) & \(n^4\)
        & 0.4926 & 0.5934 & 0.6698 & 0.3813 & 54.6 & 54.5 \\
        Explicit \(L(G)\) \(2\)-FWL
        & \(m^2\) & \(m^3\)
        & 0.0142 & 0.2655 & 1.8221 & 12.8667 & 54.9 & 58.4 \\
        Implicit ILG-\(3\)-WL
        & \(m^2\) & \(m^3\)
        & 0.0134 & 0.2693 & 1.7507 & 11.4022 & 55.3 & 60.8 \\
        \bottomrule
    \end{tabular*}
\end{table}

\paragraph{Density-sweep results.}
For the six points with \(m\geq75\), an ordinary least-squares fit of
\(\log(\text{time})\) against \(\log m\) gives slope \(2.60\)
(\(95\%\) CI \([2.18,3.03]\), \(R^2=0.986\)) for explicit line refinement
and \(2.60\) (\([2.33,2.87]\), \(R^2=0.995\)) for implicit ILG.  These
finite-range slopes describe the measured regime; the benchmark separately
records exactly \(m^3\) replacement entries per round, which is the claimed
combinatorial count.
Root \(3\)- and \(4\)-WL retain fixed state and
replacement counts throughout the sweep; their measured constants still vary
with the color alphabet.  The line-domain runtime crosses root \(4\)-WL
between \(m=100\) and \(m=150\).  At the SR25 scale \((n,m)=(25,150)\),
implicit ILG takes \(2.61\) times the one-round wall time of root \(4\)-WL.

\paragraph{Isolating structural storage.}
To separate graph representation from tuple state, we first fix \(m=150\)
and compare four connected roots with sharply different degree concentration.
Endpoint-table preparation and explicit line-graph materialization are timed
without refinement.  Allocation peaks are traced separately after root
generation, avoiding the interpreter baseline that would mask the small
endpoint table.

\begin{table}[H]
    \centering
    \caption{Fixed-\(m=150\) structural representation cost.  Each root has
    the same \(m^2=22{,}500\) subsequent line-domain pair states.  Times are
    medians of five calls after one warm-up; allocation is the peak traced
    Python heap in a separate fresh process.}
    \label{tab:empirical-structure}
    \footnotesize
    \setlength{\tabcolsep}{4pt}
    \begin{tabular*}{\textwidth}{@{\extracolsep{\fill}}lrrcccc@{}}
        \toprule
        & & & \multicolumn{2}{c}{Endpoint table}
        & \multicolumn{2}{c}{Explicit \(L(G)\)} \\
        \cmidrule(lr){4-5}\cmidrule(lr){6-7}
        Root & \(n\) & \(m_L\)
        & Time [ms] & Peak [MiB]
        & Time [ms] & Peak [MiB] \\
        \midrule
        \(P_{151}\) & 151 & 149
        & 0.099 & 0.026 & 0.510 & 0.058 \\
        Cubic & 100 & 300
        & 0.083 & 0.026 & 0.604 & 0.077 \\
        SR25 & 25 & 1,650
        & 0.091 & 0.025 & 2.357 & 0.368 \\
        \(K_{1,150}\) & 151 & 11,175
        & 0.089 & 0.025 & 15.151 & 2.017 \\
        \bottomrule
    \end{tabular*}
\end{table}

The same structural-only experiment on \(P_{m+1}\) and \(K_{1,m}\) for
\(m\in\{100,200,400,800\}\) gives endpoint-allocation slopes \(0.97\) on
both families.  Explicit materialization has slope \(0.95\) on paths, where
\(m_L=m-1\), and \(1.86\) on stars, where
\(m_L=\binom{m}{2}\).  At \(m=800\), endpoint preparation takes \(0.409\) ms
and \(0.130\) MiB on the star, versus \(454.7\) ms and \(59.0\) MiB for
explicit materialization.  Once the common pair refinement is included at
\(m=150\), however, the four implicit one-round times occupy only
\(1.38\)--\(1.46\) s and the explicit times \(1.41\)--\(1.76\) s despite a
\(75\)-fold range in \(m_L\).  Thus the endpoint representation removes the
structural term, but the \(m^2\) state and \(m^3\) replacement work dominate
the full \(k=3\) computation.

\paragraph{Stabilized SR25 case study.}
Finally, we run root \(3\)-WL, root \(4\)-WL, and ILG-\(3\)-WL to stabilization
on all \(15\) SR25 graphs.
For each graph and method, the reported time is the median of three calls
after one warm-up; totals sum the \(15\) graph medians.  Peak RSS is the
maximum over \(15\) separate memory processes.

\begin{table}[H]
    \centering
    \caption{End-to-end stabilized cost on the \(15\) SR25 graphs, yielding
    \(105\) pairwise separation verdicts.  All graphs have
    \((n,m,m_L)=(25,150,1650)\).}
    \label{tab:empirical-sr25}
    \footnotesize
    \setlength{\tabcolsep}{4pt}
    \begin{tabular*}{\textwidth}{@{\extracolsep{\fill}}lrrrrr@{}}
        \toprule
        Method & States & Entries/round & Separated
        & Total time [s] & Peak RSS [MiB] \\
        \midrule
        Root \(3\)-WL
        & 625 & 15,625 & 0/105 & 0.196 & 52.6 \\
        Root \(4\)-WL
        & 15,625 & 390,625 & 105/105 & 48.171 & 61.6 \\
        Implicit ILG-\(3\)-WL
        & 22,500 & 3,375,000 & 105/105 & 230.442 & 67.4 \\
        \bottomrule
    \end{tabular*}
\end{table}

On SR25, ILG has \(36\) times the root-\(3\)-WL state and \(216\) times its
replacement count.  Relative to root \(4\)-WL, it has \(1.44\) times the
state and \(8.64\) times the per-round replacements; stabilized implicit ILG
is \(4.78\) times slower and uses \(1.09\) times the peak RSS in this
reference implementation.  Root \(4\)-WL and ILG-\(3\)-WL separate all
\(105\) pairs.  This SR25 cost comparison does not order their
expressivity.  Root \(4\)-WL and ILG-\(3\)-WL are incomparable, so the higher
ILG cost can still yield separations that root \(4\)-WL misses
(Remark~\ref{rem:four-wl-incomparable}).

\section{Neural Evaluation Protocols}
\label{app:neural-evaluation-protocols}

This appendix specifies the pair-trained neural baselines and the untrained
ILG-\(3\)-GNN experiment in Table~\ref{tab:separation}.  The ILG-\(3\)-GNN
architecture and run seed are fixed across the substructure-counting witness
pairs, SR25, and BREC evaluations, and no parameter optimization is performed.

\paragraph{Pair-trained rows}
For locally computed pair-trained entries, we initialize and train a separate
model for each graph pair and evaluation seed; parameters are not shared
across pairs.  Training uses a Siamese cosine-embedding loss with target
\(-1\) and Adam.  The protocol uses \(q=32\) copies of each graph under
different vertex orderings; here \(q\) is the number of copies.  A pair counts as
separated when its paired-comparison statistic exceeds the BREC threshold
\(72.34\) and its isomorphic-copy reliability statistic remains below that
threshold \citep{wang2024brec}.

For the controls, substructure-counting witness pairs, SR25, and ILG-\(1\)-GNN BREC results, seeds
\(0,1,2\) are used and a pair counts as separated when at least two seeds pass
both checks.  The two controls are an isomorphic copy and the Whitney pair
\(K_3,K_{1,3}\).  Each SR25 pair is trained and evaluated separately, and the
table aggregates the \(105\) decisions.  The locally compared state domains
are root edges for ILG-\(1\)-GNN, unordered root \(3\)-subsets for Morris
\(3\)-GNN, and ordered root-vertex pairs for PPGN
\citep{morris2019weisfeiler,maron2019ppgn}.

\paragraph{ILG-\(1\)-GNN}
The ILG-\(1\)-GNN grid uses the pair-trained protocol on the controls,
substructure-counting witness pairs, SR25, and BREC.  None of the reported
pairs is separated.
This is consistent with the theoretical upper bound: under the canonical
edge/line-vertex identification, the architecture is an MPNN on \(L(G)\) and
is bounded by \(1\)-WL on that graph.

\paragraph{Morris \(3\)-GNN on BREC}
We use a fixed configuration with two shared full-replacement layers of width
\(64\), mean pooling over all unordered vertex triples, and seeds
\(100,200,300\).  After verifying that every seed produces valid results for
all \(400\) pairs, we select the seed with the largest number of separations;
pairs from different seeds are not combined.  The value \(207/400\) is
reported as a neural separation count rather than as equality with a computed
pairwise set-\(3\)-WL bound.

\paragraph{PPGN}
The BREC value \(233/400\) is sourced from
\citet[Table~2]{wang2024brec} under that paper's protocol.  The SR25 and
substructure-counting witness-pair cells are computed locally using the
pair-trained protocol above.
The row therefore reports graph-pair separation under these source-specific
protocols.

\paragraph{Untrained ILG-\(3\)-GNN configuration}
We fix run seed \(500\) and deterministically derive one fresh
initialization for every substructure-counting witness pair, SR25 pair, and
BREC pair.
The same initialized parameters are used for both graph sides, with no
parameter updates.  The global-edge architecture maintains dense ordered root-edge-pair states and
uses constant edge features, the invariant equality/incident/disjoint
relation codes, width \(32\), five replacement layers, a
\(16\)-dimensional graph output, \(\sigma(x)=\sin(5x)\), and multiscale sum
readout.  A parameter-free layer normalization is applied to the graph output
before distances are computed.  Evaluation uses CPU float64 arithmetic and
order-stabilized reductions.  The inputs contain no vertex identifiers, edge
identifiers, or exact refinement colors, and the single run seed is used
throughout.

\paragraph{Relabeling decision}
For each graph, we construct four copies: the stored ordering and three
independently relabeled copies.  Thus \(q=4\) for the reported result.  Let
\(z^G_a,z^H_b\) be their graph outputs and define
\[
\begin{split}
B&=\min_{a,b}\lVert z^G_a-z^H_b\rVert_2,\\
W_G&=\max_{a,a'}\lVert z^G_a-z^G_{a'}\rVert_2,\qquad
W_H=\max_{b,b'}\lVert z^H_b-z^H_{b'}\rVert_2.
\end{split}
\]
The pair is counted as separated only if
\[
 B-\max\{W_G,W_H\}-\tau>0,\qquad
 \tau=\sqrt{10^{-7}\cdot72.34}\approx2.6896\times10^{-3}.
\]
Here \(\tau\) is the direct-\(\ell_2\) floor induced by the BREC covariance
ridge \(10^{-7}\) and decision threshold \(72.34\).  We use this direct
\(\ell_2\) criterion rather than the BREC RPC statistic.
Every reported separation uses all four copies; the two-copy stage may reject
a pair early.  Nonfinite values and execution errors are treated as failures
to separate.

We apply the same configuration to the three substructure-counting witness
pairs, the \(105\) SR25 pairs, and the \(400\) BREC main pairs; the rule separates
\(3/3\), \(105/105\), and \(359/400\), respectively.  Every evaluation yields
valid finite outputs.
After computing the neural decisions, we compare each suite with deterministic
ILG-\(3\)-WL, and the pairwise verdicts coincide.  The same protocol does not
separate either the isomorphic \(C_6\) control or the Whitney pair
\(K_3,K_{1,3}\).  Since \(L(K_3)\cong L(K_{1,3})\cong K_3\), every invariant
line-graph model with the constant features used here must give the Whitney
pair the same output.

\paragraph{Dense execution}
All reported ILG-\(3\)-GNN results use the dense ordered edge-pair state space.
The implementation also supports exact-color state compression.  It groups
states with identical exact-refinement transition signatures while retaining
the size of each group.  For fixed weights, this is an execution optimization
of the same recurrence; compression-group indices are not supplied as
features, targets, or readout coordinates.  The table therefore
reports untrained separation counts under the rule above rather than BREC RPC
scores or supervised-training results.

Full configurations, source runners, tests, and reproduction commands are
provided in the accompanying repository:
\url{https://github.com/lukeyf/ilg-wl}.

\section*{GenAI Usage Statement}

Generative AI was used to survey related work, to develop the detailed
arguments of the proofs in Appendix~\ref{app:backward-transfer-proof}, to
assist in constructing the Lean~4 formalization, to assist in designing and
implementing the empirical complexity evaluation, and for language editing
and code assistance.  The choice of research direction and of proof strategy
was made by the author.  The author reviewed all AI-assisted material, and
the formal development is checked by the Lean~4 kernel.  The author takes
responsibility for the contents of the paper and artifact.

\end{document}